%% file: d2.tex
\documentclass[11pt]{article}
\usepackage{amsmath,amssymb,amsfonts,amsthm,epsfig,verbatim}
\usepackage{color,framed}
\usepackage{times}
\newif\ifhyper\IfFileExists{hyperref.sty}{\hypertrue}{\hyperfalse}
\hypertrue
\ifhyper\usepackage{hyperref}\fi

\def\verbose{0}
\def\confversion{0}
\def\colorful{0}

\newtheorem{theorem}{Theorem}

\newtheorem{lemma}[theorem]{Lemma}

\newtheorem{corollary}[theorem]{Corollary}
\newtheorem{claim}[theorem]{Claim}
\newtheorem{fact}[theorem]{Fact}

\newtheorem{definition}[theorem]{Definition}
\newtheorem{remark}[theorem]{Remark}
\newtheorem{observation}[theorem]{Observation}
\newcommand{\braket}[1]{\langle #1 \rangle}

\newcommand{\eps}{\epsilon}
\newcommand{\Tr}{\operatorname{Tr}}
\newcommand{\Var}{\operatorname{{\bf Var}}}
\newcommand{\Cov}{\operatorname{Cov}}

\newcommand{\Num}{\mathrm{Num}}
\newcommand{\Coeff}{\mathrm{Coeff}}
\newcommand{\be}{\mathbf{e}}
\newcommand{\sign}{\mathrm{sign}}

\newcommand{\ignore}[1]{}

\newcommand{\cref}[1]{Corollary~\ref{cor:#1}}

\newcommand{\R}{{\mathbb{R}}}

\newcommand{\Z}{{\mathbb Z}}
\newcommand{\E}{\operatorname{{\bf E}}}

\newcommand{\littlesum}{\mathop{\textstyle \sum}}

\newcommand{\poly}{\mathrm{poly}}

\newcommand{\Inf}{\mathbf{Inf}}

\newcommand{\eqdef}{\stackrel{\textrm{def}}{=}}

\renewcommand{\Pr}{\operatorname{{\bf Pr}}}
\newcommand{\dist}{\mathrm{dist}}

\newcommand{\Inner}{\mathrm{In}}
\newcommand{\Outer}{\mathrm{Out}}
\newcommand{\num}{\mathrm{num}}

\newcommand{\Sym}{{\mathrm{Sym}}}
\newcommand{\reg}{{\mathrm{reg}}}
\newcommand{\nneg}{{\mathrm{neg}}}

\newcommand{\dk}{d_{\mathrm K}}

\makeatletter
\renewcommand{\section}{\@startsection{section}{1}{0pt}{-8pt}{10pt}{\Large\bf}}
\makeatother

\newcommand{\rnote}[1]{\footnote{{\bf [[Rocco: {\red{{#1}}}\bf ]] }}}

\ifnum\colorful=1
\newcommand{\blue}[1]{{\color{blue} {#1}}}
\newcommand{\red}[1]{{\color{red} {#1}}}

\newcommand{\grade}[1]{{\color{green} {#1}}}

\newcommand{\new}[1]{{\color{red} {#1}}}
\fi
\ifnum\colorful=0
\newcommand{\blue}[1]{{{#1}}}
\newcommand{\red}[1]{{{#1}}}

\newcommand{\grade}[1]{{{#1}}}

\newcommand{\new}[1]{{{#1}}}
\fi

\newcommand{\etal}{{\em et al.\ }}
\newcommand{\T}{{{\cal T}}}
\newcommand{\depth}{{\mathrm{depth}}}

\author{Anindya De\thanks{{\tt anindya@math.ias.edu}.  Work  was partly done while the author was hosted by Oded Regev at NYU and partly while
 the author was a fellow at the Simons Institute, Berkeley. Partly supported by NSF grants CCF-1320188.   }\\
Institute for Advanced Study\\
\and Rocco A.\ Servedio\thanks{{\tt rocco@cs.columbia.edu}. Supported by NSF grants CCF-1115703 and CCF-1319788.}\\
Columbia University\\
}

\begin{document}

\setcounter{page}{0}

\title{Efficient deterministic
approximate counting\\ for low-degree polynomial threshold functions}

\maketitle

\thispagestyle{empty}

\begin{abstract}
We give a \emph{deterministic} algorithm for
approximately counting satisfying assignments of a degree-$d$ polynomial threshold function
(PTF).
Given a degree-$d$ input polynomial $p(x_1,\dots,x_n)$ over $\R^n$
and a parameter $\eps > 0$, our algorithm approximates
$
\Pr_{x \sim \{-1,1\}^n}[p(x) \geq 0]
$
to within an additive $\pm \eps$ in time $O_{d,\eps}(1)\cdot \poly(n^d)$.
(Since it is NP-hard to determine whether the above probability
is nonzero, any sort of efficient multiplicative approximation is
almost certainly impossible even for randomized algorithms.)
Note that the running time of our algorithm (as a function of $n^d$,
the number of coefficients of a degree-$d$ PTF)
is a \emph{fixed} polynomial.  The fastest previous algorithm for
this problem \cite{Kane12subpoly}, based on constructions of
unconditional pseudorandom generators for degree-$d$ PTFs, runs in time
$n^{O_{d,c}(1) \cdot \eps^{-c}}$ for all $c > 0$.

The key novel contributions of this work are

\begin{itemize}

\item A new multivariate central limit theorem, proved using tools
from Malliavin calculus and Stein's Method.  This new CLT shows that any
collection of Gaussian polynomials with small eigenvalues must have a
joint distribution which is very close to a multidimensional Gaussian
distribution.

\item A new decomposition of low-degree multilinear polynomials over Gaussian
inputs.  Roughly speaking we show that (up to some small error)
any such polynomial can be decomposed
into a bounded number of multilinear polynomials all of which have
extremely small eigenvalues.

\end{itemize}

We use these new ingredients to give a deterministic algorithm
for a Gaussian-space version of the approximate counting problem,
and then employ standard techniques for working with
low-degree PTFs (invariance principles and regularity lemmas)
to reduce the original approximate counting problem over the Boolean hypercube to the Gaussian
version.

\ignore{
Our main result follows from combining the above new ingredients with
standard techniques for working with low-degree PTFs.  These include
mollification (to go from the kind of ``closeness''
which our CLT establishes to the kind of closeness
which is required for analyzing polynomial threshold functions over
Gaussian inputs) and the use of invariance principles \cite{MOO10}
and regularity lemmas \cite{DSTW:10} for low-degree polynomials (to let us reduce the original
approximate counting problem over the Boolean hypercube to a Gaussian
version of the problem).}

As an application of our result, we give the first deterministic fixed-parameter tractable
algorithm for the following moment approximation problem:  given a degree-$d$ polynomial $p(x_1,\dots,x_n)$ over $\{-1,1\}^n$, a positive integer
$k$ and an error parameter $\eps$, output a $(1\pm \eps)$-multiplicatively
accurate estimate to $\E_{x \sim \{-1,1\}^n}[|p(x)|^k].$  Our algorithm
runs in time $O_{d,\eps,k}(1) \cdot \poly(n^d).$

%\rnote{Perhaaps we should omit the last paragraph?  }
\end{abstract}

\input{intro}

\ifnum\confversion=0

\input{prelim}

\fi

\input{multilinearize}

\input{clt}

\input{decomp}

\input{combine}

\ifnum\confversion=0
\input{from-Gaussian-to-Boolean}
\fi

\input{moments}

\medskip

\noindent {\bf Acknowledgement.}  We thank Ilias Diakonikolas for his contributions in the early stages of this project. We also thank Rafal Latala, Michel Ledoux, Elchanan Mossel, Ivan Nourdin and Krzysztof Oleszkiewicz for answering questions about the CLT. Part of this work was done when A.D. was hosted by Oded Regev and the Simons Institute. A.D. would like to thank them for their kind hospitality and support.

%\newpage
\bibliographystyle{alpha}
\bibliography{allrefs}

\end{document}

%% file: intro.tex
\newpage

\section{Introduction}

For decades a major research goal in computational complexity has been
to understand the computational power of randomization -- and perhaps to show
that randomness does not actually augment the abilities of polynomial-time
algorithms.  Towards this end, an important research goal within
unconditional derandomization has been the development of
\emph{deterministic approximate counting algorithms}.
This line of research started with the work of Ajtai and
Wigderson~\cite{AjtaiWigderson:85}, who gave a sub-exponential time
deterministic algorithm to approximately count the number
of satisfying assignments of a constant-depth circuit.
Since this early work many other classes of Boolean functions
have been studied from this perspective, including DNF
formulas, low-degree $GF[2]$ polynomials, linear
threshold functions, and degree-2 polynomial
threshold functions
\cite{LVW93, LubyVelickovic:96, Trevisan:04, GopalanMR:13,
Viola09, GKMSVV11, DDS13:deg2count,DDS14junta}.

In this paper we study the problem of deterministic approximate counting
for \emph{degree-$d$ polynomial threshold functions} (PTFs).  A
degree-$d$ PTF is a Boolean function $f: \{-1,1\}^n \to \{-1,1\}$ defined by
$f(x) = \sign(p(x))$ where $p: \{-1,1\}^n \to \R$ is a degree-$d$
polynomial.  In the special case where $d=1$, degree-$d$ PTFs are often
referred to as \emph{linear threshold functions} (LTFs).
While LTFs and low-degree PTFs have been researched for decades
(see e.g. \cite{
MyhillKautz:61, MTT:61, MinskyPapert:68, Muroga:71, GHR:92, Orponen:92,
Hastad:94,Podolskii:09} and many other works), they have recently been
the focus of renewed research attention in fields such as concrete complexity
theory \cite{Sherstov:08,Sherstov:09,DHK+:10,Kane:10,Kane12,Kane12GL,KRS12},
learning theory \cite{KKMS:08,SSS11,DOSW:11,DDFS:12stoc},
voting theory \cite{APL:07, DDS12icalp} and others.

\medskip

\noindent {\bf Our main result.}
The main contribution of this paper
is to give a fixed polynomial time deterministic approximate counting
algorithm for degree-$d$ PTFs.  We prove the following theorem:
\begin{theorem} \label{thm:main}
There is a deterministic algorithm $A$ with the following properties:
Let $A$ be given as input a degree-$d$ polynomial
$p$ over $\{-1,1\}^n$ and an accuracy parameter $\eps>0$.
Algorithm $A$ runs in time $O_{d,\eps}(1) \cdot \poly(n^d)$
and outputs a value $\tilde{v} \in [0,1]$ such that
$
\left|
\tilde{v} - \Pr_{x \in \{-1,1\}^n}[p(x) \geq 0]
\right| \leq \eps.
$
\end{theorem}

Note that the above result guarantees an
\emph{additive} approximation
to the desired probability.  While additive approximation
is not as strong as multiplicative
approximation, one should recall that the problem of
determining whether $\Pr_{x \in \{-1,1\}^n}[p(x) \geq 0]$ is nonzero
is well known to be NP-hard, even for degree-2 polynomials and even
if all nonconstant monomials in $p$ are
restricted to have coefficients from $\{0,1\}$ (this can be shown via a simple
reduction from Max-Cut).  Thus no efficient algorithm,
even allowing randomness, can give any
multiplicative approximation to $\Pr_{x \sim \{-1,1\}^n}[p(x) \geq 0]$
unless NP $\subseteq$ RP.  Given this, additive approximation is a
natural goal.

\medskip

\noindent {\bf Related work.}
Several previous works have given poly$(n^d)$-time
deterministic approximate counting algorithms for width-$d$
DNF formulas (see e.g. \cite{Trevisan:04,LubyVelickovic:96,GopalanMR:13}
as well as the approach of \cite{AjtaiWigderson:85} augmented with the
almost $t$-wise independent distributions of \cite{NN93},
as discussed in \cite{Trevisan:04}).
Degree-$d$ PTFs are of course a broad generalization of width-$d$ DNF formulas,
and the algorithms
for width-$d$ DNFs referenced above do not extend to degree-$d$ PTFs.
\ignore{ --- indeed, even for $d=1$ the first deterministic
approximate counting algorithm for degree-$d$ PTFs (based on
pseudorandom generators) was only relatively recently established in
\cite{DGJ+:10}.
}

The $d=1$ case for degree-$d$ PTFs (i.e. LTFs)
is qualitatively different from $d >1$.  For $d=1$ the satisfiability
problem is trivial, so one may reasonably hope for
a multiplicatively $(1 \pm \eps)$-accurate deterministic
approximate counting algorithm.
Indeed such an algorithm, running in fully polynomial time $\poly(n,1/\eps)$,
was given by Gopalan \etal and Stefankovic \etal in \cite{GKMSVV11}.
For $d \geq 2$, however, as noted above additive approximation is the best one
can hope for, even for randomized algorithms.
The only previous deterministic approximate counting
results for degree-$d$ PTFs for general $d$ follow
from known constructions of
\emph{unconditional pseudorandom generators} (PRGs) for degree-$d$ PTFs.
The first such construction was given by Meka and Zuckerman \cite{MZstoc10}, whose PRG yielded
an $n^{O_d(1) \cdot \poly(1/\eps^d)}$-time
deterministic approximate counting algorithm.  Followup works by
Kane \cite{Kane11ccc, Kane11focs, Kane12subpoly} improved the parameters
of these PRGs, with the strongest construction from \cite{Kane12subpoly} (for PTFs over
Gaussian inputs) giving a $n^{O_{d,c}(1) \cdot \eps^{-c}}$-time algorithm.  Thus
these prior works do not give a fixed polynomial-time algorithm.

For the special case of $d=2$, in separate work
\cite{DDS13:deg2count} the authors have given a deterministic approximate
counting algorithm for degree-2 PTFs
that runs in time $\poly(n,2^{\poly(1/\eps)})$.
In \cite{DDS14junta} the authors extended
the \cite{DDS13:deg2count} result and gave an algorithm
that does deterministic approximate counting for any $O(1)$-junta of
degree-$2$ PTFs.
As we explain in detail in the rest of this introduction,
much more sophisticated techniques and analyses
are required to obtain the results of the current paper for
general $d$.  These include \blue{a new central limit theorem} in probability
theory based on Malliavin calculus and Stein's method, and an intricate
new decomposition procedure that goes well beyond the decomposition
approaches employed in \cite{DDS13:deg2count, DDS14junta}.

\medskip

\noindent {\bf Our approach.}
The main step in proving Theorem \ref{thm:main} is to give a
deterministic approximate counting algorithm for the \emph{standard Gaussian
distribution $N(0,1)^n$} over $\R^n$ rather than the uniform
distribution over $\{-1,1\}^n.$  The key result that gives us
Theorem \ref{thm:main} is the following:

\begin{theorem} \label{thm:degd-main-gauss}
There is a deterministic algorithm $A$ with the following properties:
Let $A$ be given as input a degree-$d$ polynomial
$p$ over $\R^n$ and an accuracy parameter $\eps>0$.
Algorithm $A$ runs in time $O_{d,\eps}(1) \cdot \poly(n^d)$
and outputs a value $\tilde{v} \in [0,1]$ such that
$
\left|
\tilde{v} - \Pr_{x \sim N(0,1)^n}[p(x) \geq 0] \right| \leq \eps.
$
\end{theorem}

Theorem \ref{thm:main} follows from Theorem \ref{thm:degd-main-gauss} using the
invariance principle of \cite{MOO10} and the
``regularity lemma'' for polynomial threshold functions from
\cite{DSTW:10}.  The arguments that give Theorem \ref{thm:main} from
Theorem \ref{thm:degd-main-gauss} are essentially identical to the ones used
in \cite{DDS13:deg2count}, so we omit them in this extended abstract (see the full version).
In the rest of this introduction
we describe the main ideas behind the proof
of Theorem \ref{thm:degd-main-gauss}; as explained below, there are two
main contributions.

\medskip

\noindent {\bf First contribution:  A new multivariate CLT.} Our first contribution \blue{is}
a new multidimensional central limit theorem
 that we establish for $r$-tuples of degree-$d$ Gaussian polynomials,
i.e. $r$-tuples $(p_1(x),\dots,p_r(x))$ where each $p_i$ is a degree-$d$
polynomial and $x \sim N(0,1)^n$.
This CLT states that if each $p_i$ has ``small eigenvalues''
(as defined at the start of Section \ref{sec:CLT}), then the joint
distribution converges to the multidimensional Normal distribution
${\cal G}$ over $\R^r$ whose mean and covariance match $(p_1,\dots,p_r)$.
The closeness here is with respect to ``test functions''
that have globally bounded second derivatives; see
Theorem \ref{thm:mainclt} for a detailed statement of our CLT.  In Section
\ref{sec:combine} we use tools from
mollification to go from the aforementioned kind of
``closeness'' to the kind of closeness which is required to analyze
polynomial threshold functions.

Comparing with previous work, the degree-2 case \cite{DDS13:deg2count}
required a CLT for a single degree-2 Gaussian polynomial.
The main technical ingredient of the \cite{DDS13:deg2count} proof was
a result of Chatterjee \cite{Chatterjee:09}.
\cite{DDS14junta} established the $d=2$ case of our multidimensional CLT
via a relatively straightforward analysis (requiring just basic linear
algebra) of the central limit theorem from \cite{NPR2010}.  We note that
in the $d=2$ case it is clear what \blue{is} the correct notion of the eigenvalues of
a degree-$2$ polynomial, namely the eigenvalues of the quadratic form.  In contrast, 
it is far \red{from} clear what is the correct notion of the
eigenvalues of a degree-$d$ polynomial, \blue{especially since we require a notion that enables both a CLT and 
a decomposition as described later}.  (We note that the  tensor eigenvalue definitions
that are employed in \cite{FW95,CS13,Latala06} do not appear to be suitable for our purposes.)
Based on discussions with experts
\cite{Latala13,Nourdin13,Ledoux13,Oleszkiewicz13}, even the $d=2$ version of
our multidimensional CLT was not previously known, let alone the far
more general version of the CLT which we establish in this work.

It is instructive to consider our CLT in the context of a result of Latala
\cite{Latala06}, which shows that \red{(a somewhat different notion of)} tensor eigenvalues can be used to bound the growth of
moments of degree-$d$ Gaussian polynomials.  However, the moment bounds that can be obtained from
this approach are too weak to establish asymptotic normality \cite{Latala13}.

Like \cite{DDS14junta}, in this paper we also use the central limit theorem
from \cite{NPR2010} as a starting point.  However,
our subsequent analysis crucially relies on the fact
that there is a geometry-preserving isomorphism between the space of
symmetric  tensors and multivariate Gaussian polynomials.
This allows us to view Gaussian polynomials in terms of the associated
tensors and greatly facilitates the use of language and tools from tensor
algebra.  To \blue{establish} our condition for
asymptotic normality, we make significant use of tensor identities from
Malliavin calculus which were developed in the context of application to
Stein's method (see \cite{NourdinPeccati09,Nourdin-notes,NPR2010}).

\medskip \noindent {\bf Second contribution:  Decomposition.}
The second main contribution of this paper
is a novel decomposition that lets us transform a
multilinear degree-$d$ Gaussian polynomial $p$ into a polynomial of
the form $h(A_1,\dots,A_r)$, where (informally)

\begin{enumerate}

\item $p$ and $h(A_1,\dots,A_r)$ are $\eps$-close (i.e. $\E[p]=\E[h(A_1,\dots,A_r)]$
and $\Var[p - h(A_1,\dots,A-r)] \leq \eps$);

\item For each polynomial $A_i$, all of its eigenvalues are extremely small (at most $\eta$
for some very small $\eta$); and

\item $r=r(\eps,d,\eta)$ is independent of $n$ and depends only on the
approximation parameter $\eps$, the eigenvalue bound $\eta$, and the degree $d$ of $p$.

\end{enumerate}

This decomposition is useful for the following reasons:
Property (1) ensures that the distributions of $p$ and of $h(A_1,\dots,A_r)$ are close in cdf-distance, and thus
to in order to do approximate counting of Gaussian
satisfying assignments for $p$, it suffices to do approximate
counting of Gaussian satisfying assignments for $h(A_1,\dots,A_r)$.
Property (2)
ensures that we may apply our new CLT to the $r$-tuple
of polynomials $A_1,\dots,A_r$, and thus we may approximately count
satisfying assignments to $h(A_1,\dots,A_r) \geq 0$ by approximating
the fraction of assignments that satisfy $h({\cal G}_1,\dots,{\cal G}_r)$
where ${\cal G}=({\cal G}_1,\dots,{\cal G}_r)$ is the multidimensional
Normal distribution given by our CLT.
Finally, by
Property (3), approximating $\Pr[h({\cal G}_1,\dots,{\cal G}_r) \geq 0]$
is a ``constant-dimensional problem'' (independent of $n$) so it is
straightforward for a deterministic algorithm to approximate this
probability in time independent of $n$.

We note that there is a subtlety here which requires significant effort
to overcome.  As we discuss in Remark \ref{rem:needfewpoly}, in order for
our CLT to give a nontrivial bound it must be the case that the eigenvalue
bound $\eta$ is much smaller than $1/r$.  Mimicking decomposition approaches
previously used in literature~\cite{Servedio:07cc, MZ:09, DSTW:10} has
the problem that they will necessarily make $r \ge 1/\eta$, thus
rendering such decompositions useless for our purposes. (One exception is
the decomposition procedure  from \cite{Kane11ccc} where a similar problem
arises, but since the desired target conditions there are different from
ours, that work uses  a different approach to overcome the difficulty; we
elaborate on this below.) In our context, achieving a decomposition such
that $\eta \ll 1/r$ requires ideas that go beyond those used in previous
decompositions, and is responsible for the large \blue{``constant-factor''} overhead (captured by
$O_{d,\eps}(1)$) in the overall running time bound.

\ignore{
By combining our decomposition with our new CLT, we reduce
(in $O_{d,\eps}(1) \cdot \poly(n^d)$ time) the problem
of approximating
$\Pr_{x \sim N(0,1)^n}[p(x) \geq 0]$
to the problem of approximating $\Pr_{y \sim N(\mu,\Sigma)}[y \in \Gamma]$
for a region $\Gamma \in \R^r$ where $N(\mu,\Sigma)$ is some
$r$-dimensional Normal random variable with mean $\mu$ and covariance
matrix $\Sigma.$  The value of $r$ is a (very large) function
of $k,d$ and $1/\eps$ but it is --- crucially --- entirely independent of $n$,
so a straightforward ``brute force'' approach can approximate
the desired probability in time which is independent of $n$.
}

At a very high level our decomposition is
reminiscent of the regularity lemma for degree-$d$
polynomials over $\{-1,1\}^n$ that was given in \cite{DSTW:10}, in that both
procedures break a given degree-$d$ input polynomial into a collection of
``regular'' polynomials, but as we now explain, this resemblance is a
superficial one as there are many significant differences.
First, in the \cite{DSTW:10} setting the given input polynomials are over
$\{-1,1\}^n$ while here the polynomials are over Gaussian space;
this is a major distinction since the geometry of Gaussian space
plays a fundamental role in our proofs and techniques. Second, the notion
of ``regularity'' that is used is quite different between the two works;
in \cite{DSTW:10} a polynomial is regular if all variable influences
are small whereas here a polynomial is ``regular''
if all its ``tensor eigenvalues''
are small.  (We subsequently refer to this new notion of regularity
which is introduced and used in our work as \emph{eigenregularity.})
Third, in \cite{DSTW:10} each ``atomic step'' of
the decomposition is simply to restrict an individual
input variable to $+1$ or $-1$, whereas in this paper the atomic
``decomposition step'' now involves an eigenvalue computation
(to identify two lower-degree polynomials whose product is nontrivially
correlated with the polynomial being decomposed).
  Finally,
the \cite{DSTW:10} decomposition produces a decision tree over
input variables with restricted polynomials at the leaves, whereas in this
paper we produce a \emph{single degree-$d$ polynomial} $h(A_1,\dots,A_r)$
as the output of our decomposition.

Our decomposition has some elements that are reminiscent
of a decomposition procedure described in \cite{Kane11ccc}.
Kane's procedure, like ours, breaks a degree-$d$ polynomial
into a sum of product of lower degree polynomials.
However, there are significant differences between the procedures.
Roughly speaking,
Kane's decomposition starts with a polynomial $p$ and is aimed at upper
bounding the higher moments of the resulting constituent polynomials,
whereas our decomposition is aimed at upper bounding the eigenregularity
(magnitude of the largest eigenvalues) of the constituent polynomials.
To make sure that the number $r$ of constituent polynomials compares
favorably with the moment bounds, Kane divides these polynomials
into several classes such that the number of polynomials in any class
compares favorably with the moment bounds in that class (and some desired
relation holds between the number of polynomials in the different classes).
Instead, in our decomposition procedure, we want $r$ to compare
favorably with the eigenvalue bound $\eta$; given this requirement,
it does not seem possible to mimic Kane's approach of splitting the
constituent polynomials into several classes.
Instead, through a rather elaborate decomposition procedure,
we show that while it may not be possible to split the original
polynomial $p$ in a way so that $r$ compares favorably with $\eta$,
it is always possible to (efficiently) find a polynomial $\tilde{p}$
such that $p- \tilde{p}$ has small variance, and $\tilde{p}$ can be
decomposed so that the number of constituent polynomials compare
favorably with the eigenregularity parameter.

We note  that  it is possible for the polynomial $p-\tilde{p}$ to have small
variance but  relatively huge moments.  Thus our decomposition procedure
is not effective for the approach in \cite{Kane11ccc} which is based on
bounding moments. However, because $p-\tilde{p}$ has small variance,
the distributions of $p$ and $\tilde{p}$ are indeed close in cdf distance,
which suffices for our purposes. Thus our decomposition procedure
should be viewed as incomparable to that of \cite{Kane11ccc}.

We also remark that our decomposition is significantly more
involved than the decompositions used in \cite{DDS13:deg2count, DDS14junta}.
To see how this additional complexity arises, note that both these papers
need to decompose either a single degree-2 Gaussian polynomial or a
set of such polynomials; for simplicity assume we are
dealing with a single degree-$2$ polynomial $p$.
Then the \cite{DDS13:deg2count} decomposition procedure splits $p$ into a
sum of products of linear functions plus a degree-$2$ polynomial
which has small eigenvalues. Crucially, since a linear function of Gaussians
is itself a Gaussian, this permits a change of basis in which these
linear functions may be viewed as the new variables.
By ``restricting'' these new variables, one is essentially
left with a single degree-$2$ polynomial with a small eigenvalue.
In contrast, if $p$ has degree $d$ greater than $2$, then the
\cite{DDS13:deg2count} decomposition will split $p$ into a sum of products of
pairs of lower degree Gaussian polynomials plus a polynomial which has
small eigenvalues.
However, if $d>2$ then some or all of the new constituent lower degree
polynomials may have degree greater than $1$.
Since a polynomial of degree $d>1$ cannot itself be viewed as a Gaussian,
this precludes the possibility of ``restricting" this polynomial
as was done in \cite{DDS13:deg2count}.  Thus, one has to resort to an
iterative decomposition,
which introduces additional complications some of which were discussed above.

\ifnum\confversion=1

\medskip \noindent {\bf Organization.}
Because of space constraints proofs are omitted in this extended
abstract (see the full version which follows for all proofs).
In Section \ref{sec:multilinearize} we show that it is sufficient
to give an algorithm for deterministic approximate counting of degree-$d$
polynomials in the special case where all the polynomials are multilinear.
In Section \ref{sec:CLT} we state our new CLT for $k$-tuples of degree-$d$
Gaussian polynomials with ``small eigenvalues.''
In Section \ref{sec:decomp} we describe our decomposition procedure
that can be used to \blue{decompose a degree-$d$
multilinear polynomial over Gaussian inputs into an essentially
equivalent polynomial  that has} a highly
structured ``special form.''
In Section \ref{sec:combine} we show how the CLT from
Section \ref{sec:CLT} can be combined with the highly structured
polynomial from Section \ref{sec:decomp} to prove
Theorem \ref{thm:degd-main-gauss}.
We close in Section \ref{sec:moments} by briefly describing how
Theorem \ref{thm:main} can be applied to give the first deterministic
fixed-parameter tractable algorithm for the problem of
approximating the $k$-th absolute moment of a degree-$d$
polynomial over $\{-1,1\}^n.$

\fi

\ifnum\confversion=0

\medskip \noindent {\bf Organization.}
We begin in Section \ref{sec:prelim} by recording various useful
preliminaries, including some basics from the study of isonormal
Gaussian processes (in the context of finite-degree Gaussian polynomials)
that are required for the rest of the paper.
%\rnote{The purpose of this
%sentence is to scare the reviewer into not reading further.  Seriously,
%I think we should put this stuff in a preliminaries section
%since it will be pervasive throughout both the CLT and the decomposition
%and it is not a new contribution of this paper.}
In Section \ref{sec:multilinearize} we
show that it is sufficient
to give an algorithm for deterministic approximate counting of degree-$d$
polynomials in the special case where all the polynomials are multilinear.
In Section \ref{sec:CLT} we prove our new CLT for $k$-tuples of degree-$d$
Gaussian polynomials with ``small eigenvalues.''
In Section \ref{sec:decomp} we describe our decomposition procedure
that can be used to decompose a $k$-tuple of degree-$d$
multilinear polynomials over Gaussian inputs into an essentially
equivalent $k$-tuple of polynomials that have a highly
structured ``special form.''
In Section \ref{sec:combine} we show how the CLT from
Section \ref{sec:CLT} can be combined with the highly structured
polynomials from Section \ref{sec:decomp} to prove
Theorem \ref{thm:degd-main-gauss}.
In Section \ref{sec:from-Gaussian-to-Boolean} we sketch how
Theorem \ref{thm:main} follows from
Theorem \ref{thm:degd-main-gauss}.
We close in Section \ref{sec:moments} by briefly describing how
Theorem \ref{thm:main} can be applied to give the first deterministic
fixed-parameter tractable algorithm for the problem of
multiplicatively approximating the $k$-th absolute moment of a degree-$d$
polynomial over $\{-1,1\}^n.$
%\rnote{more detail on this, or no?  I think we don't want to make a
%big deal out of this moments stuff, but on the other hand this is the very
%first time we're mentioning it in the introduction.  We gave it 1 sentence
%in the abstract.}

\fi

% %%%%%%%%%%%%%%%%
% end of intro.tex
% %%%%%%%%%%%%%%%%

%% file: prelim.tex
% \newpage

\ifnum\confversion=0

\section{Preliminaries} \label{sec:prelim}

\subsection{Basic Definitions, Notation and
Useful Background}

\ignore{
\paragraph{Kolmogorov distance between $\R^r$-valued random variables.}
It will be convenient for us to
use a natural $r$-dimensional generalization of the Kolmogorov
distance between two real-valued random variables which we now describe.
Let $X=(X_1,\dots,X_r)$
and $Y=(Y_1,\dots,Y_r)$ be two $\R^r$-valued random variables.  We define the
\emph{$r$-dimensional Kolmogorov distance} between $X$ and $Y$ to be
\[
\dk(X,Y) = \sup_{(\theta_1,\dots,\theta_{{r}}) \in \R^{{r}}}
\left| \Pr [\forall \ i \in [{{r}}] \ X_i \leq \theta_i]
 - \Pr [\forall \ i \in [{{r}}] \ Y_i \leq \theta_i]  \right|.
\]

\paragraph{Other Notation and terminology}
}

For $A$ a real $N \times N$
matrix we write $\|A\|_2$ to denote the operator
norm
$
\|A\|_2 = \max_{0 \neq x \in \R^N} {\frac {\|Ax\|_2}{\|x\|_2}}.
$
Throughout the paper we write
$\lambda_{\max}(A)$ to denote the largest-magnitude eigenvalue of a
symmetric matrix $A$.

We will need the following standard concentration bound for
low-degree polynomials over independent  Gaussians.

\begin{theorem}[``degree-$d$ Chernoff bound'',  \cite{Janson:97}] \label{thm:dcb}
Let $p: \R^n \to \R$ be a degree-$d$ polynomial. For any
$t > e^d$, we have
\[
\Pr_{x \sim N(0,1)^n}[
|p(x) - \E[p(x)]| > t  \cdot \sqrt{\Var(p(x))} ] \leq
{d e^{-\Omega(t^{2/d})}}.
\]
\end{theorem}

\noindent We will also use the following anti-concentration bound for
degree-$d$ polynomials over Gaussians:

\begin{theorem}[\cite{CW:01}]
\label{thm:cw} Let $p: \R^n \to \R$ be a degree-$d$ polynomial
that is not identically 0.  Then for all $\eps>0$ and all
$\theta \in \R$, we have
\[
\Pr_{x \sim N(0,1)^n}\left[|p(x) - \theta| < \eps \sqrt{\Var(p)}
\right] \le O(d\eps^{1/d}).
\]
\end{theorem}

On several occasions we will require the following lemma,
which provides a sufficient condition for two degree-$d$ Gaussian polynomials
to have approximately the same fraction of satisfying assignments:

\begin{lemma} \label{lem:small-var-diff-kol-close}
Let $a(x),b(x)$ be degree-$d$ polynomials over $\R^n$.  For $x \sim N(0,1)^n$,
if $\E[a(x)-b(x)]=0$, $\Var[a]=1$ and $\Var[a-b] \leq (\tau/d)^{3d}$,
then $\Pr_{x \sim N(0,1)^n}[\sign(a(x)) \neq \sign(b(x))]\leq O(\tau).$
\end{lemma}

\begin{proof}
The argument is a straightforward consequence of Theorems \ref{thm:dcb}
and \ref{thm:cw}.
First note that we may assume $\tau$ is at most some sufficiently small
positive absolute
constant since otherwise the claimed bound is trivial.
By Theorem~\ref{thm:cw}, we have
$
\Pr[|a(x) |\le (\tau/d)^{d} ] \leq O(\tau).
$  Since
$\Var[a-b] \le (\tau/d)^{3d}$ and $a-b$ has mean 0,
applying the tail bound given by Theorem~\ref{thm:dcb}, we get
 $
 \Pr[ |a(x)-b(x)| > (\tau/d)^{d} ]  \leq O(\tau)
 $
(with room to spare, recalling that $\tau$ is at most some absolute constant).  Since $\sign(a(x))$ can disagree with $\sign(b(x))$ only if
$|a(x) |\le (\tau/d)^{d} $ or
$|a(x)-b(x)| > (\tau/d)^{d}$, a union bound gives that
$
\Pr_{x \sim N(0,1)^{n'}}[\sign(a(x)) \neq \sign(b(x)) ] = O(\tau),
$
and the lemma is proved.
\end{proof}

\subsection{A linear algebraic perspective}

We will often view polynomials over $N(0,1)^n$ as elements
from a linear vector space.  In this subsection we make this correspondence
explicit and establish some simple but useful linear algebra background
results.  In particular, we consider the finite-dimensional real vector space of all degree $d$ polynomials over $n$ variables.
This vector space is equipped with an inner product defined as follows:
for polynomials $P,Q$, we have $\langle P, Q \rangle = \mathbf{E}_{x \sim
N(0,1)^n} [ P(x) \cdot Q(x)]$.
For the rest of the paper we let $V$ denote this vector space.
\ignore{With this inner product, we now have a Hilbert space.
We will  sometimes be using Hilbert space terminology in our
subsequent discussion.
}

We will need to quantify the linear dependency between polynomials;
it will be useful for us to do this in the following way.
Let $W$ be a subspace of $V$ and $v \in V$.
We write $v^{\parallel W}$ to denote  the projection of $v$ on $W$
and $v^{\perp W}$ to denote the projection of $v$ on the space orthogonal
to $W$, so $v = v^{\parallel W} + v^{\perp W}$.
Equipped with this notation, we define the following
(somewhat non-standard) notion of linear dependency
for an ordered set of vectors:

\begin{definition}\label{def:vector-space}
Let $V$ be defined as above and let $\mathcal{A} = \{v_1, \ldots, v_m \}$
be an ordered set of unit vectors belonging to $V$.
Define $\mathcal{A}_i = \{v_1, \ldots, v_i \}$ and define $V^{i}$ to be
the linear span of $\mathcal{A}_i$. $\mathcal{A}$
is said to be \emph{$\zeta$-far from being linearly dependent}
if for every $1<i \le m$,
we have $\Vert (v_i)^{\perp V_{i-1}} \Vert_2 \ge \zeta$.
\end{definition}

Note that viewing the vector space $V$ as $\mathbb{R}^{t}$ (for some $t$),
we can associate a matrix $M_{\mathcal{A}} \in \mathbb{R}^{t \times m}$
with $\mathcal{A}$ where the $i^{th}$ column of $M_{\mathcal{A}}$ is $v_i$.
The smallest non-zero singular value of this matrix is  another measure
of the dependency of the vectors in
${\cal A}$.  Observe that this value (denoted by
$\sigma_{\min}(M_{\mathcal{A}})$) can alternately be characterized as
$$\sigma_{\min}(M_{\mathcal{A}}) = \inf_{\alpha \in \mathbb{R}^m:
\Vert \alpha \Vert_2=1}  \left\Vert \sum_{i=1}^m \alpha_i v_i
\right\Vert_2. $$

We next have the following lemma which shows that if $\mathcal{A}$ is $\zeta$-far from being linearly dependent, then the smallest non-zero singular value of $M_{\mathcal{A}}$ is noticeably large.
\begin{lemma}\label{lem:small-singular}
If $\mathcal{A}$ is $\zeta$-far from being linearly dependent (where $\zeta \le 1/4$), then $\sigma_{\min}(M_{\mathcal{A}}) \ge \zeta^{2m-2}$.
\end{lemma}

\begin{proof}
We will prove this by induction on $m$, by proving
a lower bound on $\Vert \sum_{i=1}^m \alpha_i v_i \Vert_2$
for any unit vector $\alpha \in \mathbb{R}^m$.
For $m=1$, the proof is obvious by definition. For the induction step,
observe that
$$
\left\Vert \sum_{i=1}^m \alpha_i v_i \right\Vert_2 \ge |\alpha_m| \cdot \left
\Vert v_m^{ V_{m-1}} \right\Vert_2
$$
where we use the notation from Definition~\ref{def:vector-space}. If $|\alpha_m| \ge \zeta^{2m-3}$, then we get the stated bound on  $\Vert \sum_{i=1}^m \alpha_i v_i \Vert_2$. In the other case, since $|\alpha_m| < \zeta^{2m-3}$, we
have
$$
\left \Vert \sum_{i=1}^m \alpha_i v_i \right \Vert_2
\ge \left \Vert \sum_{i=1}^{m-1} \alpha_i v_i \right \Vert_2 - |\alpha_m|
\ge  \left \Vert \sum_{i=1}^{m-1} \alpha_i v_i \right \Vert_2 - \zeta^{2m-3}.
$$
However, by the induction hypothesis, we get
$$
\left \Vert \sum_{i=1}^{m-1} \alpha_i v_i \right \Vert_2 \ge (1-\zeta^{2m-3})
\zeta^{2m-4} \ge \frac{\zeta^{2m-4}}{2}.
$$
Thus, $\Vert \sum_{i=1}^m \alpha_i v_i \Vert_2 \ge \zeta^{2m-4}/2  - \zeta^{2m-3} \ge \zeta^{2m-2}$
(provided $\zeta \le 1/4$).
\end{proof}

The next simple claim says that if $\mathcal{A}$ is $\zeta$-far from being
linearly dependent and $v$ lies in the linear span of $\mathcal{A}$, then
we can upper bound the size of the coefficients used to represent $v$.
\begin{claim} \label{claim:coeff-bound}
Let $v$ be a unit vector which lies in the span of $\mathcal{A}$ and let $\mathcal{A}$ be $\zeta$-far from being linearly dependent. Then, if $v = \sum_{i=1}^m \beta_i v_i$ is  the unique representation of $v$ as a linear combination of $v_i$'s, we have $\sqrt{\sum_{i=1}^m \beta_i^2} \le (1/\zeta)^{2m-2}$.
\end{claim}
\begin{proof}
Let $\gamma_i = \beta_{i}/\sqrt{\sum_{i=1}^m \beta_i^2}$. Since $\gamma$
is a unit vector, by Lemma~\ref{lem:small-singular} we have that
$$
\left \Vert \sum_{i=1}^m \gamma_i v_i \right \Vert_2 \ge \zeta^{2m-2}.
$$
Thus $\zeta^{2m-2} \cdot \sqrt{\sum_{i=1}^m \beta_i^2} \le 1$, giving
the claimed upper bound.
\end{proof}

We will also need another simple fact which we state below.
\begin{fact}\label{fact:no-dependence}
Let $\mathcal{A}_i$ be $\zeta$-far from being linearly dependent. Let $v_{i+1}$ and $v$ be unit vectors such that $|\langle v, v_{i+1} \rangle | \ge \zeta$ and $v$ is orthogonal to $V_i$. Then $\mathcal{A}_{i+1}$ is $\zeta$-far from being linearly dependent.
\end{fact}

\begin{proof}
Note that $v_{i+1}  = v_{i+1}^{\parallel V_i} + v_{i+1}^{\perp V_i}$ where $V_i = \mathop{span}(\mathcal{A}_i)$ (following Definition~\ref{def:vector-space}).
Hence we have
\begin{eqnarray*}
 \Vert v_{i+1}^{\perp V_i} \Vert_2 \ge |\langle v, v_{i+1}^{\perp V_i}
\rangle| =
| \langle v, v_{i+1}^{\parallel V_i} \rangle+ \langle v, v_{i+1}^{\perp V_i} \rangle| = |\langle v, v_{i+1} \rangle| \ge \zeta,
\end{eqnarray*}
where the first inequality is by Cauchy-Schwarz and the first
equality uses that $v$ is orthogonal to $V_i$.
\end{proof}

\subsection{The model}
Throughout this paper, our algorithms will repeatedly be performing basic linear algebraic operations, in particular SVD computation
and Gram-Schmidt orthogonalization.  In the bit complexity model, it is well-known that these linear algebraic operations can be performed
 (by deterministic algorithms) up to additive error $\epsilon$ in time $\poly(n, 1/\epsilon)$. For example, let $A \in \mathbb{R}^{n \times m}$ have $b$-bit rational
 entries.  It is known (see \cite{Golub} for details) that in time $\poly(n,m,b,1/\epsilon)$, it is possible to compute a value $\tilde{\sigma}_1$ and
 vectors $u_1 \in \R^n$, $v_1 \in \R^m$, such that $\tilde{\sigma}_1 = {\frac {u_1^T A v_1} {\|u_1\| \|v_1\|}}$ and $|\tilde{\sigma}_1 - \sigma_1| \leq \epsilon$,
 where $\sigma_1$ is the largest singular value of $A$. Likewise, given $n$ linearly independent vectors $v^{(1)},
 \dots,v^{(n)} \in \R^m$ with $b$-bit rational entries, it is possible
 to compute vectors $\tilde{u}^{(1)},\dots,\tilde{u}^{(n)}$ in
 time $\poly(n,m,b)$ such that
 if $u^{(1)},\dots,u^{(n)}$ is a Gram-Schmidt orthogonalization of
 $v^{(1)},\dots,v^{(n)}$ then we have $|u^{(i)} \cdot u^{(j)} -
 \tilde{u}^{(i)} \cdot
 \tilde{u}^{(j)}| \leq 2^{-\poly(b)}$ for all $i,j$.

 In this paper, we work in a unit-cost real number model of computation.
This allows us to assume that given a real matrix $A \in \mathbb{R}^{n \times m}$  with $b$-bit rational entries, we can compute the SVD of $A$ exactly in time
$\poly(n,m,b)$. Likewise, given $n$ vectors over $\mathbb{R}^m$, each of whose entries are $b$-bit rational numbers, we can perform an exact Gram-Schmidt orthogonalization
in time $\poly(n,m,b)$. Using high-accuracy approximations of the sort 
described above throughout our algorithms,
it is straightforward to translate our unit-cost real-number
algorithms into the bit complexity setting, at the cost of some additional
 error in the resulting bound.
 Note that the final guarantee we require from
Theorem \ref{thm:degd-main-gauss} is only that $\tilde{p}$
 is an additively accurate approximation to the unknown
 probability.  Note further that our Lemma \ref{lem:small-var-diff-kol-close} 
 gives the following: for $p(x)$ a degree-$d$ Gaussian polynomial with
$\Var[p]=1$, and $\tilde{p}(x)$ a degree-$d$ polynomial so that for each
 fixed monomial the coefficients of $p$ and $\tilde{p}$ differ
 by at most $\kappa$, then taking $\kappa= (\epsilon^3/(d^3 \cdot n))^{d}$,  we have that
 $|\Pr[p(x) \geq 0] - \Pr[\tilde{p}(x) \geq 0]| \leq \eps.$

 Using these two observations, it can be shown that by making
 sufficiently accurate approximations at each stage where a numerical
computation is performed by our ``idealized'' algorithm,
 the cumulative error resulting from all of the approximations
 can be absorbed into the final $O(\eps)$ error bound.  Since inverse
 polynomial levels of error can be achieved in polynomial time
 for all of the approximate numerical computations that our algorithm
performs,
and since only poly$(n^d)$ many such approximation steps are performed by
 poly$(n^d)$-time algorithms, the resulting approximate implementations
 of our algorithms in a bit-complexity model
 also achieve the guarantee of Theorem \ref{thm:degd-main-gauss},
 at the cost of a  fixed $\poly(n^d)$ overhead in the running time.
 Since working through the details of such an analysis is as tedious for the reader
 as it is for the authors, we content ourselves with this brief
 discussion.

\subsection{Polynomials and tensors:  some basics from isonormal
Gaussian processes} \label{sec:basics}
We start with some basic background; a more detailed discussion of the topics
we cover here can be found in
\cite{NourdinPeccati09,Nourdin-notes,NPR2010}.
\ifnum\verbose=1
The notes ``An introduction to Stein's Method'' by Salvador Ortiz-Latorre
are a good source for this material (starting on slide 70).
\fi

\medskip \noindent {\bf Gaussian processes.}
A \emph{Gaussian process} is a collection of jointly distributed
random variables $\{X_t\}_{t \in T}$, where $T$ is an index set, such that if $S \subset T$ is finite then the random variables $\{X_s\}_{s \in S}$ are distributed as a multidimensional normal.  Throughout this paper we will only deal with \emph{centered} Gaussian processes, meaning that  $\mathbf{E}[X_t]=0$ for every $t \in T$. It is well known that a centered Gaussian process is completely characterized by  its set of covariances $\{\mathbf{E}[X_s X_t ]\}_{s,t \in T}$. An easy but important observation is that the function $d : T \times T \rightarrow \mathbb{R}^+$ defined by $d(s,t) = \sqrt{\mathbf{E}[(X_s- X_t)^2]}$ forms a (pseudo)-metric on the set $T$.

\medskip \noindent {\bf Isonormal Gaussian processes.}   For $\mathcal{H}$ any separable Hilbert space, the Gaussian process $\{X(h)\}_{h \in \mathcal{H}}$ over ${\cal H}$ is \emph{isonormal} if it is centered and $\mathbf{E}[X(h) \cdot X(h')] = \langle h, h' \rangle$ for every $h, h' \in \mathcal{H}$. It is easy to see that for an isonormal Gaussian process, the metric $d$ induced by the process $\{X(h)\}_{h \in \mathcal{H}}$ is the same as the metric on the Hilbert space $\mathcal{H}$, and thus there is an isomorphism between the Gaussian process $\{X(h)\}$ and the Hilbert space $\mathcal{H}$. This isomorphism allows us to reason about the process $\{X(h)\}$ using the geometry of $\mathcal{H}$.

Throughout this paper the Hilbert space $\mathcal{H}$ will be the $n$-dimensional
Hilbert space $\R^n$, and we will consider the
isonormal Gaussian process $\{X(h)\}_{h \in {\cal H}}$.  This Gaussian process
can be described explicitly as follows: We equip $\mathbb{R}^n$ with
the standard normal measure $N(0,1)^n$, and corresponding to every
element $h \in \R^n$ we define $X(h) = h \cdot x$, where $x \sim N(0,1)^n$.
The resulting $\{X(h)\}$ is easily seen to be a Gaussian process
with the property that $\mathbf{E}[X(h) \cdot X(g)] = \langle h,g \rangle$.

\medskip \noindent {\bf Tensors.}
We write ${\cal H}^{\otimes q}$
to denote the $q$-tensor product of ${\cal H}$.
Fixing a basis $e_1,\dots,e_n$ of ${\cal H}$, recall that every element of
${\cal H}^{\otimes q}$ may be uniquely written as
\begin{equation} \label{eq:tensor}
f = \sum_{i_1,\dots,i_q=1}^n f(i_1, \dots, i_q)\cdot e_{i_1} \otimes
\cdots \otimes e_{i_q}
\quad
\text{where~}f(i_1,\dots,i_q) \in \R.
\end{equation}

\ifnum\verbose=1
\blue{
We define a symmetric inner product between tensors as follows.
For $s \leq q$, we define
\[
\langle e_{i_1} \otimes \dots \otimes e_{i_s},
e_{j_1} \otimes \dots \otimes e_{j_q} \rangle =
\delta_{i_1 j_1} \cdots \delta_{i_s j_s} \cdot
e_{j_{s+1}} \otimes \cdots \otimes e_{j_q}
\]
This definition extends by multilinearity to specify
$\langle f, g \rangle$ for all $f \in {\cal H}^{\otimes s},
g \in {\cal H}^{\otimes q}$.
Note that the inner product of an $s$-dimensional tensor and a
$q$-dimensional tensor is a $(q-s)$-dimensional tensor, and in particular
for $f,g \in {\cal H}^{\otimes q}$ the inner product
$\langle f,g \rangle$ is a scalar
\[
\langle f,g \rangle =
\sum_{i_1,\dots,i_q=1}^n
f(i_1, \dots, i_q)
g(i_1, \dots, i_q).
\]
}
\fi

We write ${\cal H}^{\odot q}$ to denote the subspace of
symmetric $q$-tensors.
These are the elements $f \in {\cal H}^{\otimes q}$ such that,
writing $f$ as in (\ref{eq:tensor}) above,
we have $f(i_1,\dots, i_q) = f(i_{\sigma(1)}, \dots , i_{\sigma(q)})$
for every permutation $\sigma \in S_q.$
Given a tensor $f \in {\cal H}^{\otimes q}$ written as in (\ref{eq:tensor}),
we write $\tilde{f}$ to denote the symmetrization of $f$,
\[
\tilde{f} \eqdef {\frac 1 {q!}} \sum_{\sigma \in S_q}
\sum_{i_1,\dots,i_q=1}^n f(i_{\sigma(1)}, \dots, i_{\sigma(q)}) \cdot
e_{i_1} \otimes \cdots \otimes e_{i_q},
\]
which lies in ${\cal H}^{\odot q}$;
we alternately write $\Sym(f)$ for $\tilde{f}$ if the latter is
typographically inconvenient.

\ifnum\verbose=1
We write $\alpha$ to denote a multi-index
$(\alpha_1,\dots,\alpha_n) \in (\Z_{\geq 0})^n$.  We
write $|\alpha|$ to denote the $L_1$-norm of $\alpha$, i.e. the sum
of its coordinates; intuitively
having $\alpha_i=j$ means that there are $j$ copies of $e_i$.
As an example, if $q=3$ and $n=7$ and $\alpha = (0,1,0,0,2,0,0)$ then
$\be(\alpha)$ would stand for $\Sym(e_2 \otimes e_5 \otimes e_5).$
Note that $\{\be(\alpha)\}_{\{\alpha: \ |\alpha|=q\}}$ is an orthogonal
basis of ${\cal H}^{\odot q}$.
\fi

Given an element $h \in {\cal H}$, we write $h^{\odot q}$ to denote
the $q$-th tensor product of $h$ with itself,
$h \otimes \cdots \otimes h$.  Note that this is a symmetric tensor;
we sometimes refer to $h^{\odot q}$ as a \emph{pure symmetric $q$-tensor}.
Note that every symmetric $q$-tensor can be expressed as a finite
linear combination of pure symmetric $q$-tensors $h_i^{\odot q}$
where each $h_i \in {\cal H}$ has $\|h_i\|=1$ (but this
representation is not in general unique).

We say that a tensor $f$ as in (\ref{eq:tensor}) is
\emph{multilinear} if $f(i_1,\dots,i_q)=0$ whenever $i_a=i_b$
for any $a \neq b$, $a,b \in [q]$ (i.e. all the diagonal entries of $f$
are zero).

\medskip \noindent {\bf From tensors to Gaussian polynomials and back.}
We write $H_q(x)$ to denote the $q^{th}$
Hermite polynomial; this is the univariate polynomial
\[
H_0(x) \equiv 1, H_1(x) = x, \quad
H_q(x) = {\frac {(-1)^q} {q!}} e^{x^2/2} {\frac {d^q}{dx^q}} e^{-x^2/2}.
\]
Note that these have been normalized so that
that $\E_{x \sim N(0,1)}[H_q(x)^2]= 1/q!$.
We write ${\cal W}^q$ to denote the $q$-th \emph{Wiener chaos};
this is the space spanned by all random variables of the form
$H_{q}(X(h))$
(intuitively, this is the span of all homogenous degree-$q$
multivariate Hermite polynomials.).
We note that
it can be shown (see e.g. Section 2.2 of \cite{NourdinPeccati09}) that
for $h,g \in {\cal H}$ with $\|h\|=\|g\|=1$ we have
\begin{equation}
\label{eq:hermite}
\E[H_q(X(h))\cdot H_q(X(g))] = {\frac 1 {q!}}
\langle h^{\odot q},g^{\odot q} \rangle.
\end{equation}

The \emph{iterated Ito integral} is a map which takes us from symmetric
tensors to Gaussian polynomials as follows.  Given $q \geq 1$ and
$h \in {\cal H}$ which satisfies $\|h\|=1$, we define

\begin{equation} \label{eq:ito}
I_q(h^{\odot q}) = q! \cdot H_q(X(h)).
\end{equation}

(We define $I_0(c)=c$ for $c \in \R.$) Note that with the definition of Ito integral, we can rephrase the guarantee of (\ref{eq:hermite}) as saying
$$
\mathbf{E}[I_q(h^{\odot q}) \cdot I_q (g^{\odot q}) ] = q! \cdot \langle h^{\odot q}, g^{\odot q} \rangle.
$$
So far, the map $I_q(\cdot)$ has been defined only for pure symmetric $q$-tensors of unit norm. However, equipped with the fact that every $x \in \mathcal{H}^{\odot q}$ can be written as a linear combination of such tensors, the map $I_q$ can be linearly extended to the whole space $\mathcal{H}^{\odot q}$.
Using the multiplication formula for Ito integrals (equation (\ref{eq:ito-mult}) below)
and standard identities for Hermite polynomials, it can be shown (see again Section 2.2 of \cite{NourdinPeccati09}) that such an extension
is consistent and unambiguous, i.e. it does not depend on the particular linear combination of the pure symmetric tensors we use to represent $x \in {\cal H}$.  Thus every element of ${\cal H}^{\odot q}$ maps to  an element of ${\cal W}^q$  and further this mapping can be shown to be bijective. In fact, the map $I_q$ is an isomorphism (up to scaling) between the space ${\cal W}^q$ of Gaussian chaoses  and the Hilbert space $\mathcal{H}^{\odot q}$, as is shown by the following relation: for $f, g \in \mathcal{H}^{\odot q}$, we have
$$
\mathbf{E}[I_q(f) \cdot I_q(g) ] = q! \cdot \langle f, g \rangle
$$
(ee Claim~\ref{claim:ipito} below for a proof). This relation forges a connection between the $q$-th Wiener chaos ${\cal W}^{q}$ and the geometry of the space $\mathcal{H}^{\odot q}$. This connection is crucial for us as we extensively use operations in the space $\mathcal{H}^{\odot q}$ to reason about the Wiener chaos ${\cal W}^{q}$.

Let $F=F(x_1,\dots,x_n)$  be any degree-$d$ Gaussian polynomial
over ${\cal H}=\R^n.$  Since $\E[F^2] < \infty$, the Wiener chaos decomposition implies that there exists
a unique sequence $f_0,\dots,f_d$ where $f_q \in {\cal H}^{\odot q}$ such that
\begin{equation} \label{eq:wcd}
F = \sum_{q=0}^d I_q(f_q),
\end{equation}
where by convention $I_0(f_0)=\E[F].$ Moreover, this decomposition is easily seen to be effective,
in the sense that given $F$ we can deterministically construct the tensors $f_0, \ldots, f_d$ in time
$n^{O(d)}$. In a bit more detail, let $J_q$ be the operator which maps $F : \mathbb{R}^n \rightarrow \mathbb{R}$ (with $\mathbf{E}[F^2] < \infty$) to its projection on the space $\mathcal{W}^{q}$. If $F$ is explicitly presented as a polynomial of degree $d$, then the action of each operator $J_1, \ldots, J_d$ can easily be computed in time $n^{O(d)}$, and given an explicit representation (as a polynomial) of any point $F_q$ in the image of $J_q$, it is straightforward to compute $f_q$ (in time $n^{O(q)}$) such that $I_q(f_q) = F_q$.  (In the other direction it is also straightforward, given
$f_q \in {\cal H}^{\odot q}$ for $q=0,\dots,d$, to output the degree-$d$ polynomial $F$
satisfying (\ref{eq:wcd}).)

\begin{remark} \label{rem:multilin}
If $F$ is a multilinear degree-$d$ Gaussian polynomial
over ${\cal H}$, then it can be shown that in the Wiener chaos decomposition
$F=I_0(f_0) + \cdots + I_d(f_d)$, each $f_q$ is a
multilinear symmetric tensor.  Conversely, if $f_q \in {\cal H}^{\odot q}$ is a
multilinear symmetric tensor then it can be shown that $I_q(f_q)$ is a multilinear
Gaussian polynomial.
\end{remark}

\subsection{Some background results from isonormal Gaussian processes.}
\label{sec:results}

\ignore{
We do not give proofs of all the results in this section; missing proofs
can be found in \cite{NourdinPeccati09,Nourdin-notes,NPR2010}.
\rnote{\red{I think it is best for us to give specific citations to lemmas
equations etc in specific papers for the things we claim
without proof.}}
}

We begin with a simple claim about the inner product between
Ito integrals:

\ignore{
We will use the following claim many times.
The first part states that Ito integrals
are orthogonal unless they belong to the same space and the second
gives us a formula for the inner product between two Ito integrals
that are in the same space.  (This claim is the Proposition on p. 75 of the
Ortiz-Latorre notes.)}
\begin{claim} \label{claim:ipito}
[{\bf Inner product between Ito integrals.}]  Let
$f \in {\cal H}^{\odot p}$ and $g \in {\cal H}^{\odot q}$.
Then
\[
\mathbf{E}[I_p(f) \cdot I_q(g)]=
\begin{cases}
0 & \text{if }p \neq q,\\
p! \cdot \langle f,g \rangle & \text{if }p=q.\\
\end{cases}
\]
\end{claim}

%\ifnum\verbose=1
\begin{proof}
For $p \neq q$ the claim follows from the fact that different levels
of Wiener chaos are orthogonal to each other.
For $p=q,$
we may write $f,g$ in terms of pure symmetric tensors as
$f= \sum_{i=1}^t \alpha_i \cdot f_i^{\odot p}$,
$g= \sum_{j=1}^t \beta_j \cdot g_j^{\odot p}$,
and hence
\ignore{
$I_p(f)= \sum_{i=1}^t \alpha_i \cdot I_p(f_i^{\odot p})$,
$I_p(g)= \sum_{j=1}^t \beta_j \cdot I_p(g_j^{\odot p})$.
Hence
}
\begin{eqnarray*}
\E[I_p(f)\cdot I_p(g)]
&=& \sum_{i,j=1}^t \alpha_i \beta_j\E[I_p(f_i^{\odot p})
\cdot I_p(g_j^{\odot p})]\\
&=& (p!)^2 \sum_{i,j=1}^t \alpha_i \beta_j
\E[H_p(X(f_i)) \cdot H_p(X(g_j))]\\
&=& p! \sum_{i,j=1}^t \alpha_i \beta_j
\langle f_i^{\odot p}, g_j^{\odot p} \rangle\\
&=& p! \langle f,g \rangle,
\end{eqnarray*}
where the first equality is by linearity of $I_p(\cdot)$, the second
is by (\ref{eq:ito}), the third is by (\ref{eq:hermite}),
and the fourth is by the bilinearity of $\langle \cdot , \cdot \rangle$.
\end{proof}
%\fi

As a consequence of this we get the following useful fact:
\begin{fact} \label{fact:var-ito}
Let $f_1, \ldots $ be symmetric tensors where $f_i \in \mathcal{H}^{\odot i}$. Then we have $\Var[\sum_i I_i(f_i)] = \sum_i \Var[ I_i(f_i)]$.
\end{fact}
\begin{proof}
All the random variables $I_i(f_i)$ are centered for $i \geq 1$, so it suffices to show
that $\E[\left(\sum_i I_i(f_i)\right)^2] = \sum_i \E[I_i(f_i)^2]$.  This follows directly
from Claim \ref{claim:ipito}.
\end{proof}

\paragraph{Contraction products.}
Consider symmetric tensors $f \in \mathcal{H}^{\odot q}$ and
$ g \in \mathcal{H}^{\odot r}$. For $0 \leq s \leq \min\{q,r\}$
we define the \emph{$s$-th contraction product}
$f \otimes_s g \in \mathcal{H}^{\otimes q+r-2s}$ as follows:
\ignore{(see p. 74 of Ortiz-Latorre)}
$$
(f \otimes_s g)_{ (t_1, t_2, \ldots, t_{q+r-2s}) } = \sum_{i_1, \ldots,
i_s}^{\infty} \langle f, e_{i_1} \otimes  \ldots \otimes e_{i_s} \rangle
\otimes \langle g, e_{i_1} \otimes  \ldots \otimes e_{i_s} \rangle.
$$
One way to visualize the contraction product is as a matrix
multiplication. We may view $f \in {\cal H}^q$ as a matrix $f_{q-s,s}$
where the rows of $f_{q-s,s}$ are identified with the elements of
$[n]^{q-s}$ and the columns with the elements of $[n]^s$,
and we may likewise view $g \in {\cal H}^r$ as an
$[n]^s \times [n]^{r-s}$ matrix.  A matrix
multiplication between $f_{q-s,s}$ and $g_{s,r-s}$ results in a
matrix of dimension $[n]^{q-s} \times [n]^{r-s}$, which can be viewed as an
element of $\mathcal{H}^{\otimes q+r-2s}$; this element is the
$s$-th contraction product $f \otimes_s g$.

Note that the contraction product $f \otimes_s g$ of two symmetric tensors
may not be a symmetric tensor. We write
$f \tilde{\otimes}_s g$ to denote $\Sym(f \otimes_s g)$;
the resulting symmetric tensor is an element of  $\mathcal{H}^{\odot q+r-2s}$.

\medskip

We will make heavy use of the following multiplication formula for
Ito integrals (see p. 4 of \cite{NPR2010}):
\begin{theorem} \label{thm:itomult}
[{\bf Multiplication formula for Ito integrals.}]
If $f \in \mathcal{H}^{\odot p}$ and $g \in \mathcal{H}^{\odot q}$, then
\begin{equation}
\label{eq:ito-mult}
I_p(f) \cdot I_q(g) = \sum_{r=0}^{\min\{p,q\}} r! \cdot \binom{p}{r}
\binom{q}{r} I_{p+q-2r} (f \tilde{\otimes}_r g).
\end{equation}
\end{theorem}

\fi

% %%%%%%%%%%%%%%%%%
% end of prelim.tex
% %%%%%%%%%%%%%%%%%

%% file: multilinearize.tex
%\newpage

\section{Dealing with non-multilinear polynomials}
\label{sec:multilinearize}

\ifnum\confversion=1
The decomposition procedure that we use relies heavily on the fact that
the input polynomials $p_i$ are multilinear.  To 
handle general (non-multilinear) degree-$d$ polynomials, the first step
of our algorithm is to transform them to (essentially) equivalent
multilinear degree-$d$ polynomials.  This is accomplished
by a simple procedure whose performance is
described below. Note that
given Theorem \ref{thm:multilinearize}, in subsequent sections we can (and do)
assume that the polynomial $p$ given as input in Theorem \ref{thm:degd-main-gauss}
is multilinear.
\fi

\ifnum\confversion=0
The decomposition procedure that we use relies heavily on the fact that
the input polynomials $p_i$ are multilinear.  To 
handle general (non-multilinear) degree-$d$ polynomials, the first step
of our algorithm is to transform them to (essentially) equivalent
multilinear degree-$d$ polynomials.  This is accomplished
by a simple procedure whose performance is
described in the following theorem.\footnote{A similar ``multilinearization''
procedure is analyzed in \cite{Kane11ccc}, but since the setting and required
guarantees are somewhat different here we give a self-contained algorithm and
analysis.} Note that
given Theorem \ref{thm:multilinearize}, in subsequent sections we can (and do)
assume that the polynomial $p$ given as input in Theorem \ref{thm:degd-main-gauss}
is multilinear.
\fi

\begin{theorem} \label{thm:multilinearize}
There is a deterministic procedure {\bf Linearize} with the following properties:
The algorithm takes as input a (not necessarily multilinear)
variance-1 degree-$d$ polynomial $p$ over $\R^n$
and an accuracy parameter $\delta > 0$.
It runs in time $O_{d,\delta}(1)\cdot poly(n^d)$ and outputs
a multilinear degree-$d$ polynomial $q$
over $R^{n'}$, with $n' \leq O_{d,\delta}(1) \cdot n$,
such that
\ifnum\confversion=1
$
\left|
\Pr_{x \sim N(0,1)^n}[p(x) \geq 0] -
\Pr_{x \sim N(0,1)^{n'}}[q(x) \geq 0]
\right|
\leq O(\delta).
$
\fi
\ifnum\confversion=0
\[
\left|
\Pr_{x \sim N(0,1)^n}[p(x) \geq 0] -
\Pr_{x \sim N(0,1)^{n'}}[q(x) \geq 0]
\right|
\leq O(\delta).
\]
\fi
\end{theorem}

\ifnum\confversion=0

\begin{proof}
The procedure {\bf Linearize} is given below.
(Recall that a diagonal entry of a $q$-tensor
$f = \sum_{i_1,\dots,i_q=1} f(i_1, \dots, i_q)\cdot e_{i_1} \otimes
\cdots \otimes e_{i_q}$ is a coefficient $f(i_1,\dots,i_q)$
that has $i_a=i_b$ for some $a \neq b.$)

\begin{framed}
\noindent {\bf Linearize}

\smallskip

\noindent {\bf Input:}  A degree-$d$ polynomial $p : \mathbb{R}^n \rightarrow \mathbb{R}$ such that $\Var (p) =1$.

\noindent {\bf Output:} A degree-$d$ multilinear polynomial $q : \mathbb{R}^{n'} \rightarrow \mathbb{R}$ such that $$\left|\Pr_{x \sim N(0,1)^n} [p(x) \ge 0] -\Pr_{x \sim N(0,1)^{n'}} [q(x) \ge 0] \right| \le \delta.$$  where $n'=n \cdot K$.

%\noindent Define $\mathcal{S}_{n,d} = \{(\alpha_1, \ldots, \alpha_n) : \littlesum_{i} \alpha_i \le d \textrm{ and } \forall i \in [n] \ \alpha_i \ge 0 \}$.

\noindent

\begin{enumerate}

\item Let $p(x_1,\dots,x_n) = \sum_{j=0}^d I_j(f_j)$ be the Wiener chaos decomposition of $p$.
Let $K = {d^2 \cdot (d/\delta)^{3d}}$.

\item  Construct polynomial $\tilde{q} : \mathbb{R}^{n'} \rightarrow \mathbb{R}$ from $p$ by replacing each $x_i$ by $(y_{i,1} + \ldots + y_{i,K})/\sqrt{K}$.

\item Let $\tilde{q} = \sum_{j=0}^d I_j(\tilde{f}_j)$ be the Wiener chaos decomposition of $
\tilde{q}$, where $\tilde{f}_j \in \mathcal{H}_{n'}^{\odot q}$ and $\mathcal{H}_{n'} = \mathbb{R}^{n'}$.

\item For each $0 \leq j \leq d$, obtain $g_j$ from $\tilde{f}_j$ by zeroing out all the diagonal entries from $\tilde{f}_j$.

\item Output $q = \sum_{j=0}^d I_j (g_j)$.

\end{enumerate}

\end{framed}

It is clear that all the tensors $g_j$ are multilinear, so by Remark \ref{rem:multilin}
the polynomial $q$ that the procedure {\bf Linearize} outputs is multilinear. The main step in
proving Theorem \ref{thm:multilinearize} is to bound the variance of $\tilde{q}-q$, so we will
establish the following claim:

\begin{claim} \label{claim:varbound}
$\Var[\tilde{q}-q] \leq {\frac d K} \cdot \Var[\tilde{q}].$
\end{claim}

\begin{proof}
We first observe that
\begin{equation} \label{eq:varbd}
\Var[\tilde{q}-q] = \Var\left[
\sum_{j=1}^d (I_j(\tilde{f}_j) - I_j(g_j)) \right] =
\sum_{j=1}^d \E\left[(I_j(\tilde{f}_j) - I_j(g_j))^2\right],
\end{equation}
where the first equality is because $g_0=f_0$ and the second is by Claim \ref{claim:ipito} and
and the fact that each $I_j(\tilde{f}_j), I_j(g_j)$ has mean 0 for $j \geq 1.$  Now fix a
value $1 \leq j \leq d.$
Since each $g_j$ is obtained from $\tilde{f}_j$ by zeroing out diagonal elements,
again using Claim \ref{claim:ipito} we see that
\begin{equation} \label{eq:varbd2}
\E\left[(I_j(\tilde{f}_j) - I_j(g_j))^2\right] = j! \cdot \|\tilde{f}_j - g_j\|^2_F,
\end{equation}
where the squared Frobenius norm $\|\tilde{f}_j - g_j\|^2_F$ equals the sum of squared
entries of the tensor $\tilde{f}_j - g_j$.  Now,observe that the entry $\alpha_{i_1,\dots,i_j}=f_j(i_1, \ldots, i_j)$ of the tensor $f_j $ maps to the entry
$$
\alpha_{i_1, \ldots, i_j} \frac{(e_{i_1,1} + \ldots + e_{i_1, K})}{\sqrt{K}} \otimes \ldots \otimes  \frac{(e_{i_j,1} + \ldots +e_{i_j, K})}{\sqrt{K}} = \alpha_{i_1, \ldots, i_j} \cdot \sum_{(\ell_1, \ldots, \ell_{{j}}) \in [K]^{{j}}} \frac{1}{K^{{j}/2}} \otimes_{{{a=1}}}^{{j}} e_{i_{{a}},\ell_{{a}}}
$$
when $\tilde{q}$ is constructed from $p$.
Further observe that all $K^j$ outcomes of
$\otimes_{{{a=1}}}^{{j}} e_{i_{{a}},\ell_{{a}}}$ are distinct.
Since $g_j$ is obtained by zeroing out the diagonal entries of $\tilde{f}_j$, we get that
$$
\Vert \tilde{f}_j - g_j \Vert_F^2 = \sum_{{(i_1, \ldots, i_j) \in [n]^j}}
(\alpha_{i_1, \ldots, i_j})^2 \cdot \frac{1}{K^j} \cdot  |\mathcal{S}_{K, j}|
$$
where the set $\mathcal{S}_{K,j} = \{{(\ell_1, \ldots, \ell_j)
\in [K]^j : \ell_1, \ldots, \ell_j} \textrm{ are not all distinct}\}$. It is easy to see that $|\mathcal{S}_{K, j}| \le (j^2 \cdot K^j)/{K}$, so we get
$$
\Vert \tilde{f}_j - g_j \Vert_F^2 \le \sum_{{(i_1, \ldots, i_j) \in [n]^j}}
(\alpha_{i_1, \ldots, i_j})^2 \cdot \frac{j^2}{{K}}.
$$
Returning to (\ref{eq:varbd}) and (\ref{eq:varbd2}), this yields
$$
\Var [q-\tilde{q}] \leq \sum_{j=1}^d j! \cdot \sum_{{(i_1, \ldots, i_j) \in [n]^j}}
(\alpha_{i_1, \ldots, i_j})^2 \cdot \frac{j^2}{{K}} \le \frac{d^2}{K} \cdot
\left(\sum_{j=1}^d j! \cdot \sum_{{(i_1, \ldots, i_j) \in [n]^j}} (\alpha_{i_1, \ldots, i_j})^2\right).
$$
Using Fact \ref{fact:var-ito} and Claim \ref{claim:ipito}, we see that
\[
\Var[p] = \sum_{j=1}^d \Var[I_j(f_j)]
= \sum_{j=1}^d \E[I_j(f_j)^2] =
\sum_{j=1}^d j! \cdot \sum_{{(i_1, \ldots, i_j) \in [n]^j}} (\alpha_{i_1, \ldots, i_j})^2.
\]
It is easy to see that $\Var[\tilde{q}]=\Var[p]$, which equals 1 by assumption, so we have
that $\Var[q-\tilde{q}]\leq {\frac {d^2} K} \cdot \Var[\tilde{q}]$ as desired.
\end{proof}

To finish the proof of Theorem \ref{thm:multilinearize}, observe that by our
choice of $K$ we have
$\Var[q-\tilde{q}] \le (\delta/d)^{3d} \cdot \Var[ \tilde{q}]$. Since $q-\tilde{q}$ has mean 0
and $\Var[\tilde{q}]=1$ we may apply Lemma \ref{lem:small-var-diff-kol-close},
and we get that
$|\Pr_{x \sim N(0,1)^{n'}}[q(x) \geq 0] -
\Pr_{x \sim N(0,1)^{n'}}[\tilde{q}(x) \geq 0]| \leq O(\delta).$  The theorem follows
by observing that the two distributions
$p(x)_{x \sim N(0,1)^n}$ and $\tilde{q}(x)_{x \sim N(0,1)^{n'}}$ are identical.
\end{proof}

\fi

%% file: clt.tex
%\newpage

\section{A multidimensional CLT for low-degree
Gaussian polynomials} \label{sec:CLT}

Our goal in this section is to prove a CLT
(Theorem \ref{thm:mainclt} below) which says, roughly, the
following:  Let $F_1,\dots,F_r$ be eigenregular
low-degree Gaussian polynomials over $\R^n$ (here the meaning
of ``eigenregular'' is that the polynomial has
``small eigenvalues''; more on this below).
Then the distribution of $(F_1,\dots,F_r)$ is close --- as measured
by test functions with absolutely bounded second derivatives ---
to the $r$-dimensional Normal distribution
with matching mean and covariance.

To make this statement more precise, let us begin by explaining what exactly is meant by the
eigenvalues of a polynomial --
this is clear enough for a quadratic polynomial, but not so clear for
degrees 3 and higher.

\medskip

\ifnum\confversion=1
\noindent {\bf Eigenvalues of tensors and polynomials.}  Let ${\cal H}$ denote the
Hilbert space $\R^n$, and let ${\cal H}^{\odot p}$ denote the space of symmetric $p$-tensors
over ${\cal H}.$  (See the full version for detailed background on tensors.)
We begin by defining the largest eigenvalue of a symmetric tensor.
\fi

\ifnum\confversion=0
\noindent {\bf Eigenvalues of tensors and polynomials.}
We begin by defining the largest eigenvalue of a symmetric tensor.
\fi

\begin{definition}
For any $p \ge 2$ and $g \in \mathcal{H}^{\odot p}$, define
$\lambda_{\max}(g)$, the \emph{largest-magnitude
 eigenvalue} of $g$, as follows. Consider a partition of $[p]$ into $S$
and $\overline{S}=[p] \setminus S$ where both $S$ and $\overline{S}$ are
non-empty.\ifnum\confversion=0\footnote{(Note that since we are only dealing with symmetric
tensors we could equivalently have considered only partitions into
$[1,\dots,k]$, $[k+1,\dots,p]$ where $1 \leq k \leq p-1.$)}\fi
 We define
\ifnum\confversion=1 $\lambda_{S,\overline{S}}(g) = \sup_{x \in \mathcal{H}^{S}, y \in
\mathcal{H}^{\overline{S}}}
{\frac {\langle g , x \otimes y \rangle}{\Vert x \Vert_F \cdot \Vert y \Vert_F}}$
and
$\lambda_{\max}(g)  = \max_{S,\overline{S} \neq \emptyset}
\lambda_{S,\overline{S}}(g)$. \red{(Here $\Vert x \Vert_F$ denotes the Frobenius norm of $x$.)} \fi
\ifnum\confversion=0
\[
\lambda_{S,\overline{S}}(g) = \sup_{x \in \mathcal{H}^{S}, y \in
\mathcal{H}^{\overline{S}}}
{\frac {\langle g , x \otimes y \rangle}{\Vert x \Vert_F \cdot \Vert y \Vert_F}}
\quad \quad \text{and} \quad \quad
\lambda_{\max}(g)  = \max_{S,\overline{S} \neq \emptyset}
\lambda_{S,\overline{S}}(g).
\]  \red{(Here $\Vert x \Vert_F$ denotes the Frobenius norm of $x$.)}
\fi 
For $p\in \{0,1\}$ and $g \in {\cal H}^{\odot p}$ we say that
$\lambda_{\max}(g)=0.$
\end{definition}

\ifnum\confversion=1
Let ${\cal W}^q$ (referred to as the $q$-th Wiener chaos)
denote the linear subspace of polynomials spanned by the Hermite polynomials
of degree exactly $q$ over $\R^n$.  Fix a Gaussian polynomial $F$ of degree $d$ and recall
that $F$ admits a unique Wiener chaos decomposition $F = \sum_{q=0}^d I_q(f_q)$; here
$f_q \in {\cal H}^{\odot q}$ and is the natural tensor associated with the projection of
$F$ onto ${\cal W}^q$.  Thus $I_q(\cdot)$ maps the tensor $f_q \in {\cal H}^{\odot q}$ to
a polynomial in ${\cal W}^q$.  (While the precise definition of $I_q$ is not required for the
rest of this extended abstract, the curious reader is encouraged to consult Section 2 of the full
version for additional details.)

The following definition plays a crucial role in the rest of the paper.
\fi

\ifnum\confversion=0
Fix a Gaussian polynomial $F$ of degree $d,$ and recall that $F$ has a unique
Wiener chaos decomposition as
$F = \sum_{q=0}^d I_q(f_q)$ with $f_q \in {\cal H}^{\odot q}$.
The following definition plays a crucial role in the rest of the paper.
\fi

\begin{definition} \label{def:poly}
We define the \emph{largest-magnitude eigenvalue} of $F$ to be
\ifnum\confversion=1 $
\lambda_{\max}(F) = \max\{\lambda_{\max}(f_2),\dots,$ $\lambda_{\max}(f_{\blue{d}})\}.$
\fi
\ifnum\confversion=0
\[
\lambda_{\max}(F) = \max\{\lambda_{\max}(f_2),\dots,\lambda_{\max}(f_{\blue{d}})\}.
\]
\fi
We say that $F$ is \emph{$\eps$-eigenregular} if
$
{\frac
{\lambda_{\max}(F)}
{\sqrt{\Var[F]}}}
\leq \eps,
$
and we sometimes refer to
$
{\frac
{\lambda_{\max}(F)}
{\sqrt{\Var[F]}}}$ as the \emph{eigenregularity of $F$.}
\end{definition}

\begin{remark} \label{rem:deg1}
If $F$ is a polynomial of degree at most 1 then we say that
the polynomial $F$ is $0$-eigenregular
(and hence $\eps$-eigenregular for every $\eps > 0$).
\end{remark}

Now we can give a precise statement of our new CLT:

\begin{theorem} \label{thm:mainclt}
Fix $d \geq 2$ and let $F=(F_1,\dots,F_r)$ be Gaussian polynomials over $\R^n$,
each of degree at most $d$,
such that for each $i$ we have $\E[F_i]=0$, $\Var[F_i] \leq 1$
and $F_i$ is $\eps$-eigenregular.
Let $C$ denote the covariance matrix of $F$, so $C(i,j)=
\Cov(F_i,F_j)=\E_{\blue{x \sim N(0,1)^n}}[F_i\blue{(x)} F_j\blue{(x)}].$  Let ${\cal G}$ be a mean-zero $r$-dimensional
Gaussian random variable with covariance matrix $C$.  Then
for any $\alpha : \mathbb{R}^{r} \rightarrow \mathbb{R},{\alpha \in
{\cal C}^2}$ such that all second derivatives of $\alpha$ are at most $\Vert \blue{\alpha''} \Vert_{\infty} < {\infty}$,
we have
\ifnum\confversion=1 $\left| \E  [\alpha(F_1,\dots,F_r)] -  \E [\alpha({\cal G})] \right| <   2^{O(d \log d)} \cdot r^2 \cdot
\sqrt{\eps} \cdot
\Vert \alpha'' \Vert_{\infty}. $
\fi
\ifnum\confversion=0
$$
\left| \E  [\alpha(F_1,\dots,F_r)] -  \E [\alpha({\cal G})] \right| <   2^{O(d \log d)} \cdot r^2 \cdot
\sqrt{\eps} \cdot
\Vert \alpha'' \Vert_{\infty}.
$$
\fi
 \end{theorem}

\ifnum\confversion=1
The proof of Theorem \ref{thm:mainclt} is somewhat involved,
using Malliavin calculus in the context of Stein's method; it builds on
recent work by Nourdin, Peccati and R\'{e}veillac
\cite{NourdinPeccati09,Nourdin-notes,NPR2010}.
We give the proof in the full version.
\fi

\ifnum\confversion=0

The rest of Section \ref{sec:CLT} is dedicated to the proof of
Theorem \ref{thm:mainclt}.
The proof of Theorem \ref{thm:mainclt} is somewhat involved,
using Malliavin calculus in the context of Stein's method; it builds on
recent work by Nourdin, Peccati and R\'{e}veillac
\cite{NourdinPeccati09,Nourdin-notes,NPR2010}.
In Section \ref{sec:malliavin} we first give the necessary background
ingredients from Malliavin calculus which will serve as the tools in our
proof, and in Section \ref{sec:mainclt-proof} we give our proof of Theorem
\ref{thm:mainclt}.

\fi

\begin{remark} \label{rem:needfewpoly}
It is clear from the statement of Theorem \ref{thm:mainclt} that in order for
the theorem to yield a meaningful bound, it must be the case that the number
of polynomials $r$ is small compared to $1/\sqrt{\eps}.$
Looking ahead, in our eventual application of Theorem \ref{thm:mainclt},
the $r$ polynomials $F_1,\dots,F_r$ will be obtained
by applying the decomposition procedure described in Section \ref{sec:decomp}
to the original degree-$d$ input polynomial.
Thus it will be crucially important for our decomposition procedure
to decompose the original polynomial into
$r$ polynomials all of which are \emph{extremely} eigenregular, in particular
$\eps$-eigenregular for a value $\eps \ll 1/r^2.$  \blue{Significant work will be required in
Section \ref{sec:decomp} to surmount this challenge.}
\end{remark}

\ifnum\confversion=0

\subsection{Background from Malliavin calculus.} \label{sec:malliavin}

\paragraph{Malliavin derivatives.}
Let $F = f(X(h_1), \ldots, X(h_m))$ where $h_1, \ldots, h_m \in \mathcal{H}$
and $f$ is a differentiable function.
The \emph{Malliavin derivative} is a $\mathcal{H}$ valued random variable
defined as
$$
DF \eqdef
\sum_{i=1}^n \frac{\partial f(X(h_1), \ldots, X(h_n))}{\partial x_i} h_i.
$$
Note that if $F=f(x_1,\dots,x_n)$ (i.e. $m=n$ and $h_i$  is the canonical
basis vector $e_i \in \R^n$) then we have
\[
DF = \left(
\frac{\partial f(x_1, \ldots, x_n)}{\partial x_{1}},
\ldots,
\frac{\partial f(x_1, \ldots, x_n)}{\partial x_n}\right),
\]
where as usual we have $x \sim N(0,1)^n$.

\if\verbose=1
\new{
As a simple example, if $f(x_1,x_2,x_3) = 2x_1^2 + 3 x_2 x_3$
and $F = f(X(h_1),X(h_2),X(h_3))$
then $DF = 4X(h_1)h_1 + 3 X(h_3)h_2 + 3X(h_2)h_3$
which is indeed an ${\cal H}$-valued
random variable.
}
\fi

Our proof of Theorem \ref{thm:mainclt} will involve a lot of manipulation
of inner products of Malliavin derivatives.  The following results
will be useful:

\begin{claim} \label{claim:npr-ip}
[\cite{NPR2010}]
Let $q \ge p$ and $f \in \mathcal{H}^{\odot p}$ and $g \in \mathcal{H}^{\odot q}$. Let $F = I_p(f)$ and $G = I_q(g)$.
$$
\langle DF, DG \rangle=  pq \sum_{r=1}^{\min\{p,q\}} (r-1)! \binom{p-1}{r-1} \binom{q-1}{r-1} I_{p+q-2r} (f \widetilde{\otimes}_r g)
$$
\end{claim}

\begin{theorem} \label{thm:npr-exp}
[\cite{NPR2010}]
Let $q \ge p$ and $f \in \mathcal{H}^{\odot p}$ and $g \in \mathcal{H}^{\odot q}$. Let $F = I_p(f)$ and $G = I_q(g)$.
$$
\mathbf{E}[\langle DF, DG \rangle^2]=  p^2 q^2 \sum_{r=1}^p (r-1)!^2 \binom{p-1}{r-1}^2\binom{q-1}{r-1}^2 (p+q-2r)! \Vert f \widetilde{\otimes}_r g \Vert^2
\quad \quad \textrm{if }p<q
$$
$$
\mathbf{E}[\langle DF, DG \rangle^2] = p^2 p!^2 \braket{f , g}^2 + p^4 \sum_{r=1}^{p-1}(r-1)!^2 \binom{p-1}{r-1}^4 (2p-2r)! \Vert f \widetilde{\otimes}_r g \Vert^2 \quad \quad \textrm{if }p=q
$$
$$
\mathbf{E}[\langle DF, DG \rangle^2]=  p^2 q^2 \sum_{r=1}^{q} (r-1)!^2 \binom{p-1}{r-1}^2\binom{q-1}{r-1}^2 (p+q-2r)! \Vert f \widetilde{\otimes}_r g \Vert^2
\quad \quad \textrm{if }p>q.
$$
\end{theorem}
(The last equality above is not explicitly stated in \cite{NPR2010} but it
follows easily from their proof of the first equality; see Equation 3.12
in the proof of Lemma 3.7 of \cite{NPR2010}.)

\medskip

We recall (see \cite{NPR2010,NourdinPeccati09}) that the operator $L$
(which is called the generator of the Ornstein-Uhlenbeck semigroup)
is defined by
\[
LF = \sum_{q=0}^{\infty} -q J_q(F).
\]
We also recall the that the \emph{pseudo-inverse of $L$}
is defined to be the operator
\[
L^{-1}F = \sum_{q=1}^{\infty} -J_q(F)/q.
\]
Both operators are well-defined for all finite-degree Gaussian polynomials $F$.
\ignore{
%It is not difficult to show the following:

%\begin{fact}  \label{fact:ip}
%\new{(cite)}
%$\mathbf{E}[\langle DF, DG \rangle] = \mathbf{E} [-F \times LG]
%$
%\end{fact}
}

We recall the following key identity which provides the fundamental connection
between Malliavin Calculus and Stein's method:
\begin{claim} [see e.g. Equation (2.22) of \cite{NourdinPeccati09}]
Let $h : \mathbb{R} \rightarrow \mathbb{R}$ be a continuous function
with a bounded first derivative. Let $p$ and $q$ be
polynomials over $\mathcal{X}$ with $ \mathbf{E}[q]=0$. Then
$\mathbf{E} [q h(p)] = \mathbf{E} [h'(p) \cdot \langle Dp \
, \ -DL^{-1} q\rangle] $.
\end{claim}
Specializing to the case $h(x)=x$, we have
 \begin{corollary}\label{corr:a}
 Let $p$ and $q$ be finite degree polynomials over $\mathcal{X}$ with $ \mathbf{E}[q]=0$. Then,
 $\mathbf{E} [q p] = \mathbf{E} [ \langle Dp \ , \ -DL^{-1} q\rangle] $.
 \end{corollary}

\subsection{Proof of Theorem \ref{thm:mainclt}}
\label{sec:mainclt-proof}

We recall the following CLT due to Nourdin and Peccati:

 \begin{theorem}\label{thm:Nourdin} [\cite{NourdinPeccati09}, see also
\cite{Nourdin-notes}, Theorem~6.1]
Let $F=(F_1, \ldots, F_r)$ where each $F_i$ is a Gaussian polynomial
with $\E[F_i]=0$.
Let $C$ be a symmetric PSD matrix in $\mathbb{R}^{r \times r}$
and let ${\cal G}$ be a mean-0 $r$-dimensional
Gaussian random variable with
covariance matrix $C$. Then for any $\alpha : \mathbb{R}^{r}
\rightarrow \mathbb{R}, \alpha \in {\cal C}^2$
such that $\Vert \alpha'' \Vert_{\infty} < \infty$, we have
$$
\left| \E [\alpha(F)] - \E [\alpha({\cal G})] \right| < {\frac 1 2}
\Vert \alpha'' \Vert_{\infty} \cdot \left( \sum_{i=1}^r \sum_{j=1}^r
\E [|C(i,j) - Y(i,j)|]\right)
$$
where $Y(i,j) =  \langle DF_i ,  -DL^{-1} F_j\rangle$.
% \end{theorem}
 \end{theorem}

We now use Theorem \ref{thm:Nourdin} to prove Theorem
\ref{thm:mainclt}.

\begin{proof}
As in Theorem~\ref{thm:Nourdin}, we write $Y(a,b)$ to denote
$\langle D F_a, -DL^{-1}F_b \rangle.$
For any $1 \leq a,b \leq r$, we have
\begin{equation} \label{eq:CY}
C(a,b)=\Cov(F_a,F_b) = \E[F_a F_b] = \E[Y(a,b)],
\end{equation}
where the second equality is because $F_a$ and $F_b$ have mean 0
and the third equality is by Corollary~\ref{corr:a}.
Since $C$ is a covariance matrix and every covariance matrix
is PSD, we may apply Theorem \ref{thm:Nourdin}, and we get
that
\[
\left| \E  [\alpha(F)] -  \E [\alpha({\cal G})] \right| < {\frac {r^2} 2}
\Vert \alpha'' \Vert_{\infty} \cdot \max_{1 \leq a,b \leq r}
\E[|C(a,b)-Y(a,b)|]
= {\frac {r^2} 2}
\Vert \alpha'' \Vert_{\infty} \cdot \max_{1 \leq a,b \leq r}
\E[|Y(a,b)-\E[Y(a,b)]|],
\]
where we used (\ref{eq:CY}) for the equality.
By Jensen's inequality we have
$\E[|Y(a,b)-\E[Y(a,b)]|] \leq \sqrt{\Var[Y(a,b)]}.$
Lemma \ref{lem:variance} below gives us that
$\Var[Y(a,b)] \leq 2^{O(d \log d)}\eps$,
and the theorem is proved.
\end{proof}

So to prove Theorem \ref{thm:mainclt},
it remains to establish the following lemma:

\begin{lemma} \label{lem:variance}
For each $1 \leq a,b \leq k$,
we have that $\Var[Y(a,b) ] = 2^{O(d \log d)}\eps$
where $Y(a,b) =
\langle DF_a ,  -DL^{-1} F_b\rangle$.
\end{lemma}

\subsubsection{Proof of Lemma \ref{lem:variance}}
We begin with the following useful facts about contraction products:

\begin{fact} \label{fact:contract-eigen}
Let $h \in {\cal H}^{\odot q_1}$,
$g \in {\cal H}^{\odot q_2}$ where $q_1 \geq q_2.$
Then for $1 \leq r \leq \min\{q_1-1,q_2\}$,
we have
$\|h \widetilde{\otimes}_r g\|  \leq \lambda_{\max}(h) \|g\|.
$
\end{fact}
\begin{proof}
We first observe that the range of allowed values on $r$ ensures that
the contraction product $h \widetilde{\otimes}_r g$ is well defined.
Next, we note that since symmetrizing can never increase the norm of a tensor,
we have
$\|h \widetilde{\otimes}_r g\|^2  \leq \|h {\otimes}_r g\|^2.$
As mentioned
in our earlier discussion about contraction products we may
view $h$ as an $[n]^{q_1-r} \times [n]^r$ matrix $H$
and $g$ as an $[n]^r \times [n]^{q_2-r}$ matrix $G$ with columns $G_i$.
Since $1 \leq r \leq q_1-1$ the matrix $H$ is non-degenerate (neither a single
row nor a single column), and
we have
\[
\|h {\otimes}_r g\|^2 = \|HG\|^2_F = \sum_i \|H G_i\|^2_2
\leq
\sum_i \lambda_{\max}(h)^2 \|G_i\|^2_2 =
\lambda_{\max}(h)^2 \|G\|^2_2 =
\lambda_{\max}(h)^2 \|g\|^2
\]
as claimed.
\end{proof}

\begin{fact} \label{fact:contraction-ip}
Fix $a \in {\cal H}^{\odot q_1},$
$b \in {\cal H}^{\odot q_2}$ where $q_1 \geq q_2$ and
$c \in {\cal H}^{\odot q_3},$ $d \in {\cal H}^{\odot q_4}$
where $q_3 \geq q_4.$ Then for $1 \leq r_1 \leq \min\{q_1 -1,
q_2\}$ and
$1 \leq r_2 \leq \{q_3 - 1,q_4\}$
satisfying $q_1 + q_2 - 2r_1=q_3 + q_4 - 2r_2$, we have
$
\langle
a \widetilde{\otimes}_{r_1} b,
c \widetilde{\otimes}_{r_2} d
\rangle \leq
\lambda_{\max}(a) \lambda_{\max}(c) \cdot \|b\| \cdot \|d\|.
$
\end{fact}
\begin{proof}
By Cauchy-Schwarz we have that
\[
\langle
a \widetilde{\otimes}_{r_1} b,
c \widetilde{\otimes}_{r_2} d
\rangle \leq
{\| a \widetilde{\otimes}_{r_1} b \|} \cdot
{\| c \widetilde{\otimes}_{r_2} d \|},
\]
and using Fact \ref{fact:contract-eigen} twice
this is at most the claimed bound.
\end{proof}

Fix $a,b \in [k]$.  We may write
\[
F_a = \sum_{q=1}^d I_q(a_q) \quad \quad
\text{and}
\quad \quad
F_b = \sum_{q=1}^d I_q(b_q) \quad \quad
\]
where each $a_q,b_q \in {\cal H}^{\odot q}$, and by assumption each
$2 \leq q \leq d$ has $\lambda_{\max}(a_q),\lambda_{\max}(b_q) \leq \eps.$
(Note that there is no
contribution of the form $I_0(a_0)$ because by assumption we have $\E[F_a]=0$
and $\E[I_q(a_q)]=0$ for $q>0$, and likewise for $b$.)
Recall also that by assumption we have $\Var[F_a] \leq 1$, and hence
$\E[F_a^2] \leq 1$.  Using Claim \ref{claim:ipito}, we have that
\[
\E[F_a^2] = \E\left[\left(\sum_{q=1}^d I_q(a_q)\right)^2\right] =
\sum_{q=1}^d \E[I_q(a_q)^2] = \sum_{q=1}^d q! \cdot \langle a_q,a_q\rangle \leq
1,
\]
which immediately implies that
\begin{equation} \label{eq:lownorm}
\|a_q\|^2 \leq {\frac 1 {q!}} \quad \text{for all~}q\in[d]
\text{~(and likewise~}
\|b_q\|^2 \leq {\frac 1 {q!}}\text{).}
\end{equation}

Recall that $\Var[Y(a,b)]=\E[Y(a,b)^2] - \E[Y(a,b)]^2.$  We begin by giving a
compact expression for $\E[Y(a,b)]^2$ as follows:

\begin{eqnarray} \label{eq:100}
\E[Y(a,b)]^2&=&\E[
\langle DF_a ,  -DL^{-1} F_b\rangle
]^2
= \E[F_a F_b]^2 \quad \text{(by Corollary \ref{corr:a})} \nonumber\\
&=& \E
\left[\left(
\sum_{q=1}^d I_q(a_q) \right) \left(
\sum_{q=1}^d I_q(b_q) \right)\right]^2 \nonumber\\
&=&
\left( \sum_{q=1}^d q! \langle a_q,b_q \rangle \right)^2
\quad \text{(by linearity of expectation and Claim \ref{claim:ipito})}.
\nonumber\\
\end{eqnarray}

Thus to prove Lemma \ref{lem:variance} it suffices to show that
\begin{equation} \label{eq:goal}
{
\E[Y(a,b)^2]   \leq
\left( \sum_{q=1}^d q! \langle a_q,b_q \rangle \right)^2
+
2^{O(d \log d)}\eps;
}
\end{equation}
we do this below.
We begin by writing
\begin{equation} \label{eq:Yab}
Y(a,b) =
\langle DF_a ,  -DL^{-1} F_b\rangle =
\left\langle
\sum_{q=1}^d D I_q(a_q),
\sum_{q=1}^d D I_q(b_q)/q
\right\rangle
= X + Y,
\end{equation}
where
\begin{equation} \label{eq:XY}
X = \sum_{q=1}^d {\frac 1 q} \langle D I_q(a_q),D I_q(b_q) \rangle
\quad \text{and} \quad
Y = \sum_{d \geq q_1 > q_2 \geq 1} \left({\frac 1 {q_1}} +
{\frac 1 {q_2}}\right) \langle D I_{q_1}(a_{q_1}),D I_{q_2}(b_{q_2}) \rangle.
\end{equation}
Thus our goal is to upper bound $\E[(X+Y)^2]$ by the RHS of
(\ref{eq:goal}); we do this via the following two claims.

\begin{claim} \label{claim:X}
We have
\begin{equation} \label{eq:X}
{
\E[X^2]  \leq
\left( \sum_{q=1}^d q! \langle a_q,b_q \rangle \right)^2  +
2^{O(d \log d)}\eps^2.
}
\end{equation}
\end{claim}

\begin{claim} \label{claim:Y}
We have
\begin{equation} \label{eq:Y}
{
\E[Y^2] \leq 2^{O(d \log d)}\eps^2.
}
\end{equation}
\end{claim}

Given Claims \ref{claim:X} and \ref{claim:Y} we have
\begin{eqnarray}
\E[(X+Y)]^2 &=& \E[X^2]+2\E[XY] + \E[Y^2] \leq
\E[X]^2 + \E[Y^2] + 2 \sqrt{\E[X^2]\E[Y]^2} \\
& {\leq} &
\left( \sum_{q=1}^d q! \langle a_q,b_q \rangle \right)^2  +
2^{O(d \log d)}\eps^2 + 2 \sqrt{\E[X^2]\E[Y]^2}.
\label{eq:101}
\end{eqnarray}
Now note that
\begin{eqnarray*}
\sqrt{\E[X^2]\E[Y]^2} &=&
\sqrt{\left( \left( \sum_{q=1}^d q! \langle a_q,b_q \rangle \right)^2  +
2^{O(d \log d)}\eps^2
\right) \cdot 2^{O(d\log d)}\eps^2}
\quad
\text{(by Claims \ref{claim:X} and \ref{claim:Y})}\\
&\leq& 2^{O(d \log d)}\eps^2 + 2^{O(d\log d)} {\eps}  \cdot
\sum_{q=1}^d q! \langle a_q,b_q \rangle \quad \quad \quad \text{(by~}\sqrt{x + y} \leq
\sqrt{x} + \sqrt{y}\text{)}\\
&\leq& 2^{O(d \log d)}\eps^2 + 2^{O(d \log d)}{\eps} \cdot
\sum_{q=1}^d \left(\sqrt{q!} \|a_q\|
\right) \cdot \left(\sqrt{q!} \|b_q\|\right)
\quad \text{(by Cauchy-Schwarz)}\\
&=& 2^{O(d \log d)} {\eps} \quad \text{(by (\ref{eq:lownorm})).}
\end{eqnarray*}

Combining this with (\ref{eq:101})
we indeed get (\ref{eq:goal}) as desired.  Thus it remains only to
prove Claims \ref{claim:X} and \ref{claim:Y}.

\paragraph{Proof of Claim \ref{claim:X}.}

We may write $X^2$ as $A+B$, where

\[
A = \sum_{q=1}^d {\frac 1 {q^2}}
\langle D I_q(a_q),DI_q(b_q) \rangle^2
\]
and
\[
B =
\sum_{d \geq q_1 > q_2 \geq 1} {\frac 2 {q_1 q_2}}
\langle D I_{q_1}(a_{q_1}),DI_{q_1}(b_{q_1}) \rangle \cdot
\langle D I_{q_2}(a_{q_2}),DI_{q_2}(b_{q_2}) \rangle.
\]

First we analyze $\E[A]$.
Using Theorem \ref{thm:npr-exp} we have that
\begin{eqnarray}
\E[A] &=&
\sum_{q=1}^d {\frac 1 {q^2}}
\E[\langle D I_q(a_q),DI_q(b_q) \rangle^2]\nonumber\\
&=&
\sum_{q=1}^d
(q!)^2 \langle a_q,b_q\rangle^2 + q^2 \sum_{r=1}^{q-1}((r-1)!)^2
{q-1 \choose r-1}^4 \cdot (2q -2r)
\|a_q \widetilde{\otimes}_r b_q\|^2.
\label{eq:A}
\end{eqnarray}
Now observe that for $1 \leq r \leq q-1,$ we have
\[
\|a_q \widetilde{\otimes}_r b_q\|^2 \leq  \lambda_{\max}(a_q)^2 \cdot
\|b_q\|^2 \leq \eps^2/q!
\]
where we have used Fact \ref{fact:contract-eigen} for the first inequality
and the eigenvalue bound and (\ref{eq:lownorm}) for the second.
Hence from (\ref{eq:A}) we obtain
\begin{equation} \label{eq:Abound}
\E[A] \leq
\sum_{q=1}^d
(q!)^2 \langle a_q,b_q\rangle^2 + 2^{O(d \log d)} \eps^2.
\end{equation}

We turn now to bound $\E[B].$
Using Claim \ref{claim:npr-ip} we get
\begin{eqnarray} \label{eq:B}
\E[B] &=&
\sum_{d \geq q_1 > q_2 \geq 1} {\frac 2 {q_1 q_2}}
\E\left[
\left(
q_1^2 \sum_{r_1=1}^{q_1} (r_1-1)! {q_1 - 1 \choose r_1-1}^2
I_{2q_1 - 2r_1}
(a_{q_1} \widetilde{\otimes}_{r_1}b_{q_1})
\right)
\cdot \nonumber \right.\\
& &\text{~~~~~~~~~~~~~~~~~~~~~~~} \left.
\left(
q_2^2 \sum_{r_2=1}^{q_2} (r_2-1)! {q_2 - 1 \choose r_2-1}^2
I_{2q_2 - 2r_2}
(a_{q_2} \widetilde{\otimes}_{r_2}b_{q_2})
\right)
\right] \nonumber\\
&=&
\sum_{d \geq q_1 > q_2 \geq 1} 2 q_1 q_2
\sum_{r_1=1}^{q_1} \sum_{r_2=1}^{q_2}
(r_1-1)! (r_2-1)!{q_1-1 \choose r_1-1}^2 {q_2-1 \choose r_2-1}^2
\cdot \nonumber\\
& & \text{~~~~~~~~~~~~~~~~~~~~~~~~~~~~~~~~}
\E[
I_{2q_1 - 2r_1}
(a_{q_1} \widetilde{\otimes}_{r_1}b_{q_1})
I_{2q_2 - 2r_2}
(a_{q_2} \widetilde{\otimes}_{r_2}b_{q_2})
].
\end{eqnarray}

Let us fix a given outcome of $q_1 > q_2$.  Recalling
Claim \ref{claim:ipito}, we see that
the only $(r_1,r_2)$ pairs that will give a nonzero expectation
are those
such that $2q_1-2r_1=2q_2-2r_2$, i.e. $r_2=q_2-q_1+r_1$.  For such an
$(r_1,r_2)$
pair, by Claim \ref{claim:ipito} we get
that
$
\E[
I_{2q_1 - 2r_1}
(a_{q_1} \widetilde{\otimes}_{r_1}b_{q_1}) \cdot
I_{2q_2 - 2r_2}
(a_{q_2} \widetilde{\otimes}_{r_2}b_{q_2})]
$
equals
$
(2q_1 - 2r_1)!
\langle a_{q_1} \widetilde{\otimes}_{r_1} b_{q_1} ,
a_{q_2} \widetilde{\otimes}_{r_2} b_{q_2} \rangle
$, which in turn satisfies
\begin{eqnarray*}
\langle a_{q_1} \widetilde{\otimes}_{r_1} b_{q_1} ,
a_{q_2} \widetilde{\otimes}_{r_2} b_{q_2} \rangle
&=&
\langle a_{q_1},b_{q_1} \rangle \cdot \langle a_{q_2},b_{q_2} \rangle
\text{~if~}r_1=q_1,\\
\langle a_{q_1} \widetilde{\otimes}_{r_1} b_{q_1} ,
a_{q_2} \widetilde{\otimes}_{r_2} b_{q_2} \rangle
&\leq &
\eps^2 \cdot {\frac 1 {\sqrt{q_1!}}} \cdot {\frac 1 {\sqrt{q_2!}}}
\text{~if~}1 \leq r_1 \leq q_1-1,
\end{eqnarray*}
where the inequality follows from
Fact \ref{fact:contraction-ip}, the eigenvalue bound,
and (\ref{eq:lownorm}).
We thus arrive at
\begin{eqnarray*}
E[B]&\leq&
\sum_{d \geq q_1 > q_2 \geq 1}
2(q_1)!(q_2)!
\langle a_{q_1},b_{q_1} \rangle \cdot \langle a_{q_2},b_{q_2} \rangle\\
&& +
\sum_{d \geq q_1 > q_2 \geq 1}
2q_1q_2
\sum_{r_1=1}^{q_1-1}
(r_1-1)! (q_2-q_1+r_1-1)!{q_1-1 \choose r_1-1}^2 {q_2-1 \choose
q_2-q_1+r_1-1}^2
\cdot \eps^2 \cdot {\frac
{(2q_1 - 2r_1)!} {\sqrt{(q_1!)(q_2!)}}}\\
&\leq&
\sum_{d \geq q_1 > q_2 \geq 1}
2(q_1)!(q_2)!
\langle a_{q_1},b_{q_1} \rangle \cdot \langle a_{q_2},b_{q_2} \rangle
+ 2^{O(d \log d)}\eps^2 .
\end{eqnarray*}
Combining this with (\ref{eq:A}), we get that
\[
\E[X^2] \leq
\left(\sum_{q=1}^d q! \langle a_q,b_q \rangle\right)^2 + 2^{O(d \log d)}\eps^2.
\]
This concludes the proof of Claim \ref{claim:X}. \qed

\paragraph{Proof of Claim \ref{claim:Y}.}
We have
\begin{eqnarray*}
\E[Y^2]
&=&
\sum_{d \geq q_1 > q_2 \geq 1}
\sum_{d \geq q_3 > q_4 \geq 1}
\left({\frac 1 {q_1}} + {\frac 1 {q_2}}\right)
\left({\frac 1 {q_3}} + {\frac 1 {q_4}}\right)
\cdot
\E[
\langle DI_{q_1}(a_{q_1}),DI_{q_2}(b_{q_2})\rangle
\cdot
\langle DI_{q_3}(a_{q_3}),DI_{q_4}(b_{q_4})\rangle
]\\
&< &
4
\sum_{1 \leq q_1 < q_2 \leq d}
\sum_{1 \leq q_3 < q_4 \leq d}
\sqrt{\E[
\langle DI_{q_1}(a_{q_1}),DI_{q_2}(b_{q_2})\rangle^2]} \cdot
\sqrt{
\E[
\langle DI_{q_3}(a_{q_3}),DI_{q_4}(b_{q_4})\rangle^2
]},
\end{eqnarray*}
where we have used Cauchy-Schwarz and the fact that
$
({\frac 1 {q_1}} + {\frac 1 {q_2}})
({\frac 1 {q_3}} + {\frac 1 {q_4}})
$ is always strictly less than 4.
Fix any $d \geq q_1 > q_2 \geq 1$; to prove
Claim  \ref{claim:Y} it suffices to show that
$
\E[
\langle DI_{q_1}(a_{q_1}),DI_{q_2}(b_{q_2})\rangle^2] \leq
2^{O(d \log d})\eps^2.$
For this we use the third bound of Theorem  \ref{thm:npr-exp}, which gives
\begin{eqnarray*}
\E[ \langle DI_{q_1}(a_{q_1}),DI_{q_2}(b_{q_2})\rangle^2]
&=& (q_1)^2 (q_2)^2 \sum_{r=1}^{q_2} ((r-1)!)^2 {q_1-1 \choose r-1}^2{q_2-1 \choose
r-1}^2 (q_1 + q_2 - 2r)
\|a_{q_1} \widetilde{\otimes}_r b_{q_2} \|^2.
\end{eqnarray*}
For any $1 \leq r \leq q_2$ we have that $r \leq q_1-1$ (since $q_1 > q_2$),
and hence by Fact \ref{fact:contract-eigen}, the eigenvalue
bound and (\ref{eq:lownorm}) we get that
$
\|a_{q_1} \widetilde{\otimes}_r b_{q_2} \|^2 \leq \eps^2 / q_2!$.
Thus each summand in the previous expression is at most $2^{O(q_1 \log q_1)}
\eps^2 = 2^{O(d \log d)}\eps^2$, as required.
This concludes the proof of Claim \ref{claim:Y}, and with it
the proof of Lemma \ref{lem:variance} and of Theorem \ref{thm:mainclt}.
 \qed

\fi

%% file: decomp.tex
\section{Decomposing $k$-tuples of multilinear
degree-$d$ Gaussian polynomials}
\label{sec:decomp}

In this section we prove our main decomposition result for $k$-tuples
of multilinear Gaussian polynomials, Theorem \ref{thm:main-decomp}.
We begin by giving a precise statement of the result, followed by a
discussion of how the result fits into our broader context.

\begin{theorem} \label{thm:main-decomp}
Fix $d \geq 2$ and fix any non-increasing computable
function $\beta: [1,\infty) \to (0,1)$ that satisfies
$\beta(x) \leq 1/x.$
There is a procedure {\bf \blue{Regularize-Poly}}$_\beta$
with the following properties.
The procedure takes as input a degree-$d$ multilinear Gaussian polynomial $p$
with $\Var[p]=1$ and a parameter $\tau>0$.
It runs in $\poly(n^d) \cdot O_{d,\tau}(1)$
time and outputs a collection of polynomials $\{h_{q}\}_{q=0,\dots,d}$
and $\{A_{q,\ell}\}_{q=0,\dots,d,\ell=
1,\dots,m_{q}}.$

Write $p(x)$ as
$\sum_{q=0}^d c_q p_q(x)$ where $p_q \in {\cal W}^q$ for all $q$ and
$\Var[p_q]=1$ for $1 \leq q \leq d$.
For $0 \leq q \leq d$ and $x \in \R^n$,
let
\ifnum\confversion=1 $\tilde{p}_q(x) = c_q h_{q}(A_{q,1}(x),\dots,A_{q,m_{q}}(x))$ and let $
\tilde{p}(x) =  \sum_{q=0}^d \tilde{p}_q (x).$
\fi
\ifnum\confversion=0
\[
\tilde{p}_q(x) = c_q h_{q}(A_{q,1}(x),\dots,A_{q,m_{q}}(x)),
\quad \quad \text{and let} \quad \quad
\tilde{p}(x) =  \sum_{q=0}^d \tilde{p}_q (x).
\]
\fi
The following conditions hold:

\begin{enumerate}

\item For each $q \in \{0,\dots,d\}$ the polynomial
$\tilde{p}_q$ belongs to ${\cal W}^q$.  Moreover, for $q \in \{1,\dots,d\}$,  each polynomial $A_{q,\ell}$ belongs to ${\cal W}^j$
for some $1 \leq j \leq d$ and has $\Var[A_{q,\ell}]=1.$

\item We have
$
\left|
\Pr_{x \sim N(0,1)^n}[p(x) \geq 0] -
\Pr_{x \sim N(0,1)^n}[\tilde{p}(x) \geq 0]
\right|
\leq O(\tau),
$
and moreover $\Var_{x \sim N(0,1)^n}[p(x)-\tilde{p}(x)] \leq (\tau/d)^{3d}.$

\item
Each polynomial $h_q$ is a multilinear polynomial in its $m_q$
arguments.  Moreover, there exist functions $N_\beta(d,\tau)$ and $M_\beta(d,\tau)$
such that if $\Coeff(h_q)$ denotes the sum of the absolute values of
the coefficients of $h_q$, then $\sum_{q=1}^d \Coeff(h_q) \leq M_\beta(d,\tau)$
and the number of arguments to all $h_q$'s, $\sum_{q=1}^d m_q$, is at most
$N_\beta(d,\tau).$ Also, \blue{the} degree of the polynomial $h_q$ (for all $1 \le q \le d$) is upper bounded by $d$.

\item Further, let $\Num = \sum_{q=1}^d m_q$ and $\Coeff = \sum_{q=1}^d \Coeff(h_q)$.
Then each polynomial $A_{q,\ell}(x)$ is $\beta(\Num + \Coeff)$-eigenregular.

\end{enumerate}

\end{theorem}

\medskip \noindent {\bf Discussion.}
Intuitively, Condition (2)  means that it will be sufficient to do
deterministic approximate counting for the polynomial
$\tilde{p}$ rather than the original polynomial $p$.
Condition (4) ensures that
the eigenregularity of each
polynomial $A_{q,\ell}$ compares favorably both with the number of
\blue{polynomials} produced and with the size of the coefficients in the
``outer'' polynomials $h_q$.  As discussed earlier, having the eigenregularity be small
relative to the number of polynomials is crucial since it
\blue{is required in order for} our CLT, Theorem \ref{thm:mainclt}, to yield
a good bound.  We need to deal with the size of the coefficients for technical
reasons -- as we will see in Section \ref{sec:combine}, we will
apply our CLT where its ``test function'' $\alpha$ is
a smooth approximator to the 0/1-valued function which, on
input $\{A_{j,\ell}\}_{j,\ell}$, outputs 1 if and only if
$
\sign\left(\blue{\sum_{q=0}^d h_{q}(A_{q,1},\dots,A_{q,m_{q}})}\right) = 1.
$
Our CLT's quantitative bound depends on the second derivative of $\alpha$, and to bound this
we need coefficient size bounds on the $h_q$ polynomials.

\medskip

We build up to Theorem \ref{thm:main-decomp}
in a sequence of incremental stages.
In Section \ref{sec:onepoly} we begin
by describing how to decompose a single element of
a given Wiener chaos.  Because of our requirement that the number
of polynomials produced in the decomposition must be very small
relative to the eigenregularity that is achieved ---
see Remark \ref{rem:needfewpoly} --- even this is a non-trivial task,
requiring two ``layers'' of decomposition and an approach that
goes well beyond the decomposition techniques in previous
work \cite{DDS13:deg2count,
DDS14junta}.  In Section \ref{sec:manypolyonewiener} we extend this and
describe how to  simultaneously decompose a $k$-tuple
of elements of the same Wiener chaos.
(See the beginning of Section \ref{sec:manypolyonewiener} for an explanation
of why we need to be able to simultaneously decompose many polynomials
at once.) In Section \ref{sec:wholething} we describe how to handle
a $k(d+1)$-tuple of elements where there are $k$ elements
from each of the $d+1$ Wiener chaoses ${\cal W}^{0},\dots,
{\cal W}^{d}$.  \ifnum\confversion=1 In the full version, we show how the decomposition of Section 
\ref{sec:wholething} yields Theorem \ref{thm:main-decomp}. \fi
\ifnum\confversion=0 Finally, in Section
\ref{sec:main-decomp} we specialize this decomposition for a $k(d+1)$-tuple to obtain Theorem \ref{thm:main-decomp}.
\fi

\subsection{Decomposing a single multilinear
element of the $q$-th Wiener chaos}
\label{sec:onepoly}

\ifnum\confversion=0

We begin with a useful definition and fact:

\begin{definition} \label{def:disjoint}
For $S \subseteq [n]$,
we say that a tensor
\[
f = \sum_{i_1,\dots,i_q=1}^n f(i_1,\dots,i_q) \cdot e_{i_1} \otimes
\cdots \otimes e_{i_q}
\]
is \emph{supported on $S$} if $f(i_1,\dots,i_q)=0$
whenever any of $i_1,\dots,i_q$ do not belong to $S$.
\end{definition}

\begin{fact} \label{fact:disjoint-vanish}
Let $f \in {\cal H}^{\odot p}$ be supported on $S \subseteq [n]$
and $g \in {\cal H}^{\odot q}$ be supported on $T \subseteq [n]$
where $S,T$ are disjoint.  Then for any $0 < r \leq \min\{p,q\}$
we have that the contraction product $f \otimes_r g$ equals 0.
\end{fact}
\begin{proof}
We may write
$$
f = \sum_{j_1, \ldots, j_p \in S} f(j_1, \ldots, j_p) e_{j_1} \otimes
\ldots \otimes
e_{j_p}
\quad \text{and} \quad
g = \sum_{j_1, \ldots, j_q \in T} g(j_1, \ldots, j_q ) e_{j_1} \otimes \ldots \otimes
e_{j_q}.
$$
Note that
$$
f \otimes_r g = \sum_{i_1, \ldots, i_r \in [n]} \langle f , e_{i_1}
\otimes \ldots \otimes e_{i_r}
\rangle \otimes \langle g, e_{i_1} \otimes \ldots \otimes e_{i_r} \rangle
$$
The value $\langle f , e_{i_1} \otimes \ldots \otimes e_{i_r} \rangle$ is 0
unless all the $i_j$'s lie in $S$, and
likewise $ \langle g, e_{i_1} \otimes \ldots \otimes e_{i_r} \rangle  =0$
unless all the $i_j$'s lie in $T$.
Since $S \cap T = \emptyset$, the fact is proved.
\end{proof}

As our lowest-level decomposition tool,
we show that given a symmetric tensor with a large-magnitude
eigenvalue, it is possible to efficiently find two
lower-dimensional symmetric tensors $g_1$ and $g_2$,
supported on disjoint subsets of $[n]$,
such that $f$ is correlated with the product of $g_1$ and $g_2$:

\begin{lemma}\label{lemma:1-d}
Fix any $q \geq 2$.
There is a $\poly(n^q)$-time deterministic procedure
with the following properties:
Let $f \in \mathcal{H}^{\odot q}$ be a multilinear symmetric tensor
that has $\Var[I_q(f)]=1$ and $\lambda_{\max}(f) \ge \eta>0$.
On input $f$, the procedure outputs
multilinear symmetric tensors $g_1 \in \mathcal{H}^{\odot q_1}$ and
$g_2 \in \mathcal{H}^{\odot q_2}$ such that the following
conditions hold:

\begin{enumerate}

\item $\Var [I_{q_1}(g_1)] = \Var[I_{q_2}(g_2)] =1$, and
$q_1 , q_2 >0$ with $q_1 + q_2 = q$;

\item $g_1$ and $g_2$ are supported on disjoint sets $S,T \subset [n]$;

\item
$\E[I_q(f) \cdot I_{q_1}(g_1) \cdot I_{q_2}(g_2)] \geq \eta/2^q$.

\end{enumerate}

\end{lemma}

\begin{proof}
The procedure begins as follows:

\begin{enumerate}

\item For each partition of $[q]$ into $S$ and $\overline{S} = [q] \setminus S$
where $|S|,|\overline{S}|>0$, it computes the value
$\lambda_{S,\overline{S}}(f) = \sup_{x \in \mathcal{H}^{S}, y \in
\mathcal{H}^{\overline{S}}}
{\frac {\langle g , x \otimes y \rangle}{\Vert x \Vert \Vert y \Vert}}
$ and the associated tensors $\alpha = x/\|x\|$,
$\beta = y/\|y\|$.

\item For the $(\alpha,\beta)$ pair from the iteration of step (1)
which achieves the largest value for $\lambda_{S,\overline{S}}(f)$,
let $\alpha' = \tilde{\alpha}/\|\tilde{\alpha}\|$ and $q_1=|S_1|$,
and let $\beta' = \tilde{\beta}/\|\tilde{\beta}\|$ and $q_2=|S_2|.$

\end{enumerate}

Recalling the variational characterization of singular values, each
iteration of Step (1) (for a given choice of $S,\overline{S}$) is an
SVD computation and can be performed in time $\poly(n^q)$.
Since there are at most $2^q$ partitions $S,\overline{S}$ to consider,
the overall running time of these first two steps is
$\poly(n^q)$.  

We proceed to analyze these first two steps of the procedure.\ignore{
Note first that since $\Var I_q(f) =1$, we have
$\Vert f \Vert_2 = 1/\sqrt{q!}$.}
The fact that $\lambda_{\max}(f) \ge \eta$ means that the maximizing
$\alpha,\beta$ identified in Step (2) will satisfy
$\Vert \alpha \Vert_2 = \Vert \beta \Vert_2=1$
and $\langle f , \alpha \otimes \beta \rangle \ge \eta$.
Since $f$ is a multilinear tensor these tensors $\alpha,\beta$
will also be multilinear (if, say, $\alpha$
had a nonzero diagonal entry then a larger value of
$\langle f, \alpha \otimes \beta \rangle$ could be obtained
by zeroing out that entry and rescaling).
Since $f$ is a symmetric tensor, it follows that
\[
\langle f, \alpha \otimes \beta \rangle =
\langle f, \tilde{\alpha} \otimes \beta \rangle =
\langle f, \tilde{\alpha} \otimes \tilde{\beta} \rangle  \geq \eta.
\]
Since symmetrizing cannot increase the norm of a tensor, we have that
$\|\tilde{\alpha}\| \leq \|\alpha\|=1$ and likewise for
$\beta$.  It follows that
$\langle f, \alpha' \otimes \beta' \rangle \geq \eta$,
$\|\alpha'\|=\|\beta'\|=1$, $\alpha' \in {\cal H}^{q_1},$ and
$\beta' \in {\cal H}^{q_2}$.

From this point it is not difficult to achieve conditions (1) and (3) of the
lemma; to achieve condition (2) we proceed as follows.  Note that
\[
\langle f,  \alpha' \otimes \beta' \rangle =
\sum_{S_1 \in [n]^{q_1}, S_2 \in [n]^{q_2}} f(S_1, S_2) \alpha'(S_1)
\beta'(S_2).\]
Now consider a randomized process that divides $[n]$
into two sets $\mathcal{A}_1$ and $\mathcal{A}_2$ by independently
assigning each $i \in [n]$ to $\mathcal{A}_1$ with probability $1/2$
and to $\mathcal{A}_2$ with probability $1/2$ (we will later
derandomize this process below).
Given an outcome of ${\cal A}_1$ and ${\cal A}_2$,
we consider $\nu_1 \in \mathcal{H}^{\odot q_1}$ and
$\nu_2 \in \mathcal{H}^{\odot q_2}$ defined as follows:
for each $S_1 \in [n]^{q_1},S_2 \in [n]^{q_2}$,
$$ \nu_1 (S_1) = \alpha'(S_1) \cdot \mathbf{1} [S_1 \subseteq \mathcal{A}_1]
\quad \textrm{ and } \quad \nu_2(S_2) = \beta'(S_2) \cdot \mathbf{1} [S_2
\subseteq \mathcal{A}_2],$$
where ``$S_i \subseteq {\cal A}_i$'' means that each coordinate of $S_i$ lies in
${\cal A}_i$.
We have that
\begin{equation}
\label{eq:rand}
\mathbf{E} [ \langle f,  \nu_1 \otimes \nu_2 \rangle ] =
\sum_{S_1 \in  [n]^{q_1}, S_2 \in [n]^{q_2}} f(S_1, S_2)
\alpha'(S_1) \beta'(S_2) \cdot (1/2)^{|S_1|} \cdot (1/2)^{|S_2|},
\end{equation}
where $|S_i|$ denotes the number of distinct values that are present in
the coordinates of $S_i$.  Since $\alpha'$ and $\beta'$ are
multilinear, the only nonzero contributions to the sum
(\ref{eq:rand}) are from $(S_1,S_2)$ pairs
with $|S_1|=q_1$ and $|S_2|=q_2=q-q_1$.
Hence we have
\[
\mathbf{E} [ \langle f,  \nu_1 \otimes \nu_2 \rangle ] =
{\frac 1 {2^q}}
\sum_{S_1 \in  [n]^{q_1}, S_2 \in [n]^{q_2}} f(S_1, S_2)
\alpha'(S_1) \beta'(S_2) \geq {\frac \eta {2^q}}
.
\]
The above analysis requires only $q$-wise independence, so constructing
${\cal A}_1$ and ${\cal A}_2$ (and the resulting
$\nu_1,\nu_2$) using a $q$-wise independent
distribution ${\cal D}$, we get that
\[
\mathbf{E}_{{\cal D}} [ \langle f,  \nu_1 \otimes \nu_2 \rangle ] \geq
{\frac \eta {2^q}},
\]
and thus some outcome in the support of ${\cal D}$ must achieve at least
the expected value.  The third step of the algorithm is to deterministically
enumerate over all points in the support of the $q$-wise independent
distribution (using well-known constructions of $q$-wise independent
distributions \cite{ABI85} this can be done in time $\poly(n^q)$)
and to output a pair $g_1={\frac {\nu_1}
{\sqrt{q_1!} \|\nu_1\|}}$,
$g_2 = {\frac {\nu_2}{\sqrt{q_2!}\|\nu_2\|}}$ that has
$\langle f,  \nu_1 \otimes \nu_2 \rangle \geq
{\frac \eta {2^q}}.$

We now verify that $g_1$ and $g_2$ satisfy all the required conditions.
First, it is clear that $g_1$ and $g_2$ satisfy $g_1 \in {\cal H}^{\odot
q_1}$, $g_2 \in {\cal H}^{\odot q_2}$,
and it is obvious from the construction
that $g_1$ and $g_2$ are supported
on disjoint sets ${\cal A}_1$ and ${\cal A}_2$,
so condition (2) indeed holds.
Turning to condition (1),
since $q_1 > 0$ we have
that $\Var[I_{q_1}(g_1)] = \E[I_{q_1}(g_1)^2]$, which equals 1 by
Claim \ref{claim:ipito} (and similarly we get
that $\Var[I_{q_2}(g_2)] = 1$).
For Condition (3),
we first show that
$I_{q_1}(g_1) \cdot I_{q_2}(g_2) = I_{q}(\Sym(g_1 {\otimes} g_2))$
(and hence $I_{q_1}(g_1) \cdot I_{q_2}(g_2)$
lies in the Wiener chaos of degree $q$).  To see this,
recall from the multiplication formula
for Ito integrals (Theorem \ref{thm:itomult}) that we have
$$I_{q_1} (g_1) \cdot I_{q_2}(g_2) = \sum_{r=0}^{\min\{q_1,q_2\}} r! \cdot
\binom{q_1}{r} \binom{q_2}{r} I_{q_1+q_2-2r}
\Sym(g_1 \otimes_r g_2)).
$$
Since $g_1$ and $g_2$ are supported on disjoint sets, by Fact
\ref{fact:disjoint-vanish} all summands with $r \neq 0$ vanish, and
we get
$I_{q_1} (g_1) \cdot I_{q_2}(g_2) = I_{q_1+q_2}(\Sym(g_1 \otimes
g_2))$ as claimed.

With this identity in hand, we have that
\begin{eqnarray*}
\E[I_q(f) I_{q_1}(g_1) I_{q_2}(g_2)] &=&
\E[I_q(f) I_q(\Sym(g_1 \otimes g_2))]\\
&=& q! \langle f, \Sym(g_1 \otimes g_2) \rangle
\quad \quad \text{(by Claim \ref{claim:ipito})}\\
&=& q!  \langle f, g_1 \otimes g_2 \rangle
\quad \quad \quad \quad \text{(since $f$ is symmetric)}\\
&\geq&
\langle f, \nu_1 \otimes \nu_2 \rangle
\quad \quad \text{(since $\|\nu_1\| \leq \|\alpha'\|=1$ and
$\|\nu_2\| \leq \|\beta'\|=1$)}\\
&\geq& {\frac \eta {2^q}},
\end{eqnarray*}
and Lemma \ref{lemma:1-d} is proved.
\end{proof}

We are now ready to define our first algorithmic primitive, the procedure
{\bf Split-One-Wiener}.  This procedure either certifies that its input polynomial
(an element $I_q(f)$ of the $q$-th Wiener chaos) is eigenregular, or else
it ``splits off'' a product $P \cdot Q$ from its input polynomial.
(Here and subsequently the suffix ``\blue{{\bf -One-Wiener}}'' indicates that
the procedure applies only to \blue{one element 
$I_q(f)$ belonging to one level of the Wiener chaos.})

\fi

\ifnum\confversion=1

Our first algorithmic primitive is the procedure
{\bf Split-One-Wiener}.
(Here and subsequently the suffix ``\blue{{\bf -One-Wiener}}'' indicates that
the procedure applies only to \blue{one element 
$I_q(f)$ belonging to one level of the Wiener chaos.})
This procedure either certifies that its input polynomial
(an element $I_q(f)$ of the $q$-th Wiener chaos) is eigenregular, or else
it ``splits off'' a product $P \cdot Q$ from its input polynomial 
\blue{and expresses $I_q(f)$ as $c \cdot P \cdot Q + R$ for some $c \in \R$.}
We omit the formal
description here because of space constraints.
\fi

\ifnum\confversion=0

\begin{lemma}\label{lemma:algo-primitive}
Fix any $q \geq 2$.
There is a deterministic procedure {\bf Split-One-Wiener} which takes as input
a polynomial $I_q(f) \in \mathcal{H}_q$ that has $\Var[I_q(f)]=1$
and a parameter $\eta>0$.  {\bf Split-One-Wiener} runs in deterministic
$\poly(n^q,1/\eta)$ time and has the following guarantee:

\begin{itemize}

\item If $\lambda_{\max}(f) < \eta$, then {\bf Split-One-Wiener}
returns ``\textsf{eigenregular}''.

\item Otherwise, if $\lambda_{\max}(f) \geq \eta$, then {\bf Split-One-Wiener}
outputs a quadruple $(P, Q ,R, c)$ with the following properties:

 \begin{enumerate}

 \item $P =I_{q_1}(g_1) \in {\cal W}^{q_1}$ and $Q = I_{q_2}(g_2)
\in {\cal W}^{q_2}$ where $q_1 +
q_2 = q$, $q_1, q_2>0$, and $\Var[I_{q_1}(g_1)]=
\Var[I_{q_2}(g_2)] =1$.
\item The tensors $g_1 \in {\cal H}^{\odot q_1}$ and
$g_2 \in {\cal H}^{\odot q_2}$
are supported on disjoint sets $S,T \subset [n].$

 \item $P \cdot Q \in {\cal W}^{q}$ and $\Var[P\cdot Q]=1$,
and all of $P,Q,R$ are multilinear.
 \item The value $c \eqdef \E[I_q(f) \cdot P \cdot Q]$
satisfies $c \geq \eta/2^q.$
 \item $R \in {\cal W}^{q}$ and
$I_q(f) = c P \cdot Q + R$ and $\E[P \cdot Q \cdot R]=0$.
 \item $\Var (R) = 1-c^2$.
 \end{enumerate}

\end{itemize}

  \end{lemma}

\begin{proof}
{\bf Split-One-Wiener} runs the procedure from Lemma \ref{lemma:1-d} and
checks whether the largest value $\lambda_{S,\overline{S}}(f)$
achieved in Step (1) is at least $\eta.$  If it is less than
$\eta$ then it outputs ``\textsf{eigenregular}.'' Otherwise it
sets $P= I_{q_1}(g_1)$, $Q = I_{q_2}(g_2)$,
sets $c = \E[I_q(f)\cdot P \cdot Q] = q! \langle f,
g_1 \otimes g_2 \rangle,$ and sets $R = I_q(f)=c \cdot P \cdot Q.$

Lemma \ref{lemma:1-d} directly gives properties (1),(2)
and (4), and property (3) follows from the
fact that $\E[Q]=\E[P]=0$ and $P$ and $Q$ are independent
random variables (observe that by property (2) they are
polynomials over disjoint sets of variables).
The first two parts of (5) are immediate; for the last part,
recalling that $R = I_q(f) - c P \cdot Q$, we have that $R$
is simply the component of $I_q(f)$ that is orthogonal to $P \cdot Q$.
Since $R$ lies in ${\cal W}^q$ its mean is zero, so by linear
algebra we have that $\Var[R]=\E[R^2]=1-c^2$ as claimed.
\end{proof}

\fi

\ignore{

\begin{definition}\label{def:linear-dependence}
Let $\mathcal{S} = \{v_1, \ldots, v_m \}$ be an ordered set of unit vectors belonging to $\mathbb{R}^\ell$. Define $W^i = \mathop{span}(v_1, \ldots, v_i)$. Then, $\mathcal{S}$ is said to be $\zeta$-far from being linearly dependent
if for every $1<i \le m$, the projection of $v_i$ on $W^{i-1}$ has
\end{definition}
}

Building on the algorithmic primitive {\bf Split-One-Wiener}, we now describe
a procedure {\bf Decompose-One-Wiener} which works by iteratively executing
{\bf Split-One-Wiener} on the ``remainder'' portion $R$ that was ``left over''
from the previous call to {\bf Split-One-Wiener}.  Intuitively, the overall
effect of this procedure is to break its input polynomial
into a sum of products of pairs of polynomials, plus a remainder term
which is either eigenregular or else has variance which is negligibly small.

\ifnum\confversion=0
(We draw the reader's attention to the quantitative bound on coefficients
given by property (6) of Lemma \ref{lemma:algo-primitive2}.
This coefficient bound will play a crucial role in the mollification
procedure of Section \ref{sec:combine}.)

\fi

\begin{lemma}\label{lemma:algo-primitive2}
Fix any $q \geq 2$.
There is a deterministic procedure {\bf Decompose-One-Wiener}
which  takes as input a polynomial $I_q(f) \in \mathcal{\blue{W}}_q$ that has
$\Var[I_q(f)]=1$ and parameters $\eta$ and $\epsilon$.
{\bf Decompose-One-Wiener} runs in poly$(n^q,1/\eta,\log(1/\eps))$ time and
has the following guarantee:

\ifnum\confversion=1
\begin{enumerate}

\item It outputs a set $L$ of
triples $\{(c_i, P_i, Q_i) \}_{i=1}^{m}$ and a polynomial $R$
such that $I_q(f) = \sum_{i=1}^{m}
c_i P_i Q_i + R$.

\item For each $i=1,\dots,m$ we have
$P_i \in {\cal W}^{q_{i,1}}$ and $Q_i \in {\cal W}^{q_{i,2}}$
with $q_{i,1},q_{i,2}>0$ and $q_{i,1}+q_{i,2}=q$.
\blue{Moreover}
$\Var[P_i] = \Var[Q_i] = \Var[P_i \cdot Q_i] =1$ for all $i \in [m]$,
$R \in {\cal W}^{q}$, 
all $P_i,Q_i$ and $R$ are multilinear,
\blue{
and $P_i$ and $Q_i$ are defined over disjoint sets of variables.
}

\item $m \leq O((4^q/\eta^2) \log (1/\epsilon))$ and
$\sum_{j=1}^m c_j^2 \le (2^q/\eta)^{4(m-1)}$.

\item
Either $R$ is $\eta$-eigenregular,
in which case {\bf Decompose-One-Wiener} returns ``\textsf{eigenregular remainder}'',
or else $\Var[R] \le \epsilon$,
in which case {\bf Decompose-One-Wiener} returns ``\textsf{small remainder}''.

\item $\mathbf{E}[ (\littlesum_{j=1}^m c_j P_j \cdot Q_j) \cdot
R]=0$. As a consequence, we
have $\Var[\littlesum_{j=1}^m c_j P_j \cdot Q_j] + \Var
[R]= 1$.

  \end{enumerate}

\fi

\ifnum\confversion=0

\begin{enumerate}

\item It outputs a set $L$ of
triples $\{(c_i, P_i, Q_i) \}_{i=1}^{m}$ and a polynomial $R$
such that $I_q(f) = \sum_{i=1}^{m}
c_i P_i Q_i + R$.

\item For each $i=1,\dots,m$ we have
$P_i \in {\cal W}^{q_{i,1}}$ and $Q_i \in {\cal W}^{q_{i,2}}$
with $q_{i,1},q_{i,2}>0$ and $q_{i,1}+q_{i,2}=q$;
moreover
$\Var[P_i] = \Var[Q_i] = \Var[P_i \cdot Q_i] =1$ for all $i \in [m]$,
$R \in {\cal W}^{q}$,
and all $P_i,Q_i$ and $R$ are multilinear.

\item $m \leq O((4^q/\eta^2) \log (1/\epsilon))$.

\item
Either $R$ is $\eta$-eigenregular,
in which case {\bf Decompose-One-Wiener} returns ``\textsf{eigenregular remainder}'',
or else $\Var[R] \le \epsilon$,
in which case {\bf Decompose-One-Wiener} returns ``\textsf{small remainder}''.

\item $\mathbf{E}[ (\littlesum_{j=1}^m c_j P_j \cdot Q_j) \cdot
R]=0$. As a consequence, we
have $\Var[\littlesum_{j=1}^m c_j P_j \cdot Q_j] + \Var
[R]= 1$.

\item $\sum_{j=1}^m c_j^2 \le (2^q/\eta)^{4(m-1)}$.
  \end{enumerate}

\fi

\end{lemma}

\ifnum\confversion=0

\begin{proof}
The procedure {\bf Decompose-One-Wiener} is defined below.  It is helpful to keep the following invariant
in mind:  At any stage during the execution of the algorithm, we let $V_L$ denote the linear subspace spanned by $\{P_i \cdot Q_i \}$ where $L= \{(c_i, P_i, Q_i)\}$.
The algorithm maintains the invariant that $I_q(g)$ is orthogonal to $V_L$ (as is clear by
construction).

  \begin{itemize}

  \item[i.] Initialize $L$ to the empty set of triples and the index $m$ to $0$.

  \item[ii.] Initialize $g = f$ and hence $I_q(g)  = I_q(f)$.

  \item[iii.] If $\Var[ I_q(g)]  \le \epsilon$, then output
the set $L=\{(c_i,P_i,Q_i)\}_{i=1}^m$ and the polynomial
$R = I_q(f) - \littlesum_{i=1}^m c_i P_i \cdot Q_i$,
and return ``\textsf{small remainder}.''

  \item[iv.] Else, choose a constant $\zeta$ so that $\Var [I_q(\zeta g)] =1$.

  \item[v.] Run procedure {\bf Split-One-Wiener} (using parameter $\eta$)
  on $I_q(\zeta g)$. If it
returns ``\textsf{eigenregular}'', then stop the procedure and output
the set $L=\{c_i,P_i,Q_i\}_{i=1}^m$ and the polynomial
$R = I_q(f) - \littlesum_{i=1}^m c_i P_i \cdot Q_i$,
and return ``\textsf{eigenregular remainder}''.
%$I_q(g)$ and $I_q(f) - I_q(g) = \littlesum_{i} c_i P_i \cdot Q_i$.

  \item[vi.] Else if the output of {\bf Split-One-Wiener} is $(P, Q, R,c')$,
then append $(c_{m+1},P_{m+1},Q_{m+1})$ to the list $L$
where $c_{m+1}=c'$, $P_{m+1}=P$ and $Q_{m+1}=Q.$
Now, project the polynomial $I_q(f)$
to $V_L$ and let $I_q(g)$ denote the part of $I_q(f)$ that is orthogonal to $V_L$,
i.e. $I_q(g) = (I_q(f))^{\perp V_L}.$
Recompute the constants $c_1,\dots,c_{m+1}$ so that
with the recomputed constants we have
$I_q(f) = \sum_{i=1}^{m+1}c_i P_i Q_i + I_q(g).$
Increment $m$ and go to Step [iii].

  \end{itemize}

We now establish the claimed properties. The first and fourth properties are obvious.
The second property follows directly from Lemma \ref{lemma:algo-primitive2}.
For the fifth property, note that $I_q(g)$ is orthogonal to
$I_q(f) - I_q(g)$ by construction. It remains to justify the third and the sixth properties.
We do this using the following claim:

\begin{claim} \label{claim:36}
At each stage in the execution of {\bf Decompose-One-Wiener},
when $L=\{(c_i,P_i,Q_i)\}_{i=1,\dots,k}$, the set $\{P_i \cdot Q_i\}_{i=1}^k$ is $\eta/2^q$-far
from being linearly dependent.
\end{claim}

\begin{proof}
The proof is by induction.  The claim is trivially true for $k=1.$
For the inductive step, observe that by construction, just
before $(c_{k+1},P_{k+1},Q_{k+1})$ is appended to the list $L$, we have
(by property (2)) that $P_{k+1}\cdot Q_{k+1}$ is a unit vector, and that
$|\langle I_{q}(\zeta g), P_{k+1} \cdot Q_{k+1} \rangle| \geq \eta/2^{q}.$
Since $I_q(\zeta g)$ is orthogonal to $V_L$ (before appending
$(c_{k+1},P_{k+1},Q_{k+1})$), by Fact
\ref{fact:no-dependence} we get the stated claim.
\end{proof}

\ignore{
%Treating the polynomial $P_i \cdot Q_i$ as elements of the linear vector space spanned by degree $d$ %polynomials, we claim that the set $\{P_i \cdot Q_i\}_{i=1}^m$ is $\eta$-far from being
%linearly dependent.  We will prove this by induction. It is obvious for the singleton set $\{P_1 \cdot %Q_1\}$. Now, assume that the hypothesis holds at the $j^{th}$ step. In other words, the set $\{P_i %\cdot Q_i\}_{i=1}^j$ is $\eta$-far from being linearly dependent. Note that by definition $\Var %I_q(\zeta g)=1$ and $I_q(\zeta g)$ is orthogonal to the span of the set $\{P_i \cdot Q_i\}_{i=1}^j$. %Further, the inner product of the polynomial $P_{j+1} \cdot Q_{j+1}$ with $I_q(\zeta g)$ is at least $\eta$. Thus, applying Fact~\ref{fact:no-dependence}, we get that
% $\{P_i \cdot Q_i\}_{i=1}^{j+1}$ is $\eta$-far from being linearly dependent.
}

When the {\bf Decompose-One-Wiener} procedure terminates, note that by property (5) we have
that $\Var[\sum_{i=1}^m c_i P_i \cdot Q_i] \leq 1$.  Hence applying Claim \ref{claim:36}
with Claim \ref{claim:coeff-bound}, we get property (6).

It remains only to establish property (3).  This follows immediately from the following claim:

\begin{claim}
At each stage in the execution of {\bf Decompose-One-Wiener},
when $L=\{(c_i,P_i,Q_i)\}_{i=1,\dots,k}$ and $I_q(g)=(I_q(f))^{\perp V_L}$, we have
$\Var[I_q(g)] \leq (1-\red{\eta^2/4^q})^k.$
\end{claim}
\begin{proof}
As before the proof is by induction and the base claim (when $k=0$) is immediate. For the inductive step, just before appending $(c_{k+1},P_{k+1},Q_{k+1})$ to the list $L$
in Step (vi), note that if we define $I_q(h)=I_q(g)-c_{k+1}P_{k+1}Q_{k+1}$, then by the
{\bf Split-One-Wiener} guarantee (property (6) of Lemma \ref{lemma:algo-primitive})
we have that $\Var[I_q(h)] \leq (1 - \eta^2/4^q) \cdot \Var[I_q(g)]$, which by the
inductive hypothesis is at most $(1-\eta^2/4^q)^{k+1}$.  Since the vector $I_q(f)-I_q(h)$
lies in the linear span of $V_L \cup \{P_{k+1} \cdot Q_{k+1}\}$, and
$\|I_q(h)\| \leq (1-\red{\eta^2/4^q})^{k+1}$, hence after
appending $(c_{k+1},P_{k+1},Q_{k+1})$ to $L$, we have that the new polynomial $I_q(g)$
defined in step (vi) has $\|I_q(g)\| \leq (1-\eta/2^q)^{k+1}.$  This concludes the proof.
\end{proof}

This concludes the proof of Lemma \ref{lemma:algo-primitive2}.
\end{proof}

\fi

We note that the guarantees of the {\bf Decompose-One-Wiener} procedure bear some resemblance to the
decomposition that is used in \cite{DDS13:deg2count} for degree-2 Gaussian
polynomials.  However, in our current context of working with degree-$d$
polynomials, {\bf Decompose-One-Wiener} is not good enough, for the following
reason:  Suppose that {\bf Decompose-One-Wiener} returns
``\textsf{eigenregular remainder}'' and outputs a decomposition of
$I_q(f)$ as $\sum_{i=1}^m c_i P_i Q_i + R$.  While the polynomial
$R$ is $\eta$-eigenregular, it is entirely possible that the number of polynomials
$P_i,Q_i$ in the decomposition (i.e. $2m$) may be as large as $
\Omega({\frac 1 {\eta^2}} \log(1/\eps))$.  We would like to apply
our CLT to conclude that the joint distribution of $R$ and the polynomials
obtained from the subsequent decomposition of
\blue{$P_1,Q_1,\dots,P_m,Q_m$} is close to a normal
distribution, but since the number $2m$ of polynomials is already
too large when compared to the inverse of the eigenregularity parameter,
we cannot use our CLT
(recall Remark \ref{rem:needfewpoly}).
\footnote{Note that
the reason this problem did not arise in the degree-2 polynomial
decompositions of
\cite{DDS13:deg2count} is because each polynomial $P_i,Q_i$ obtained
from {\bf Decompose-One-Wiener} in that setting must have degree 1
(the only way to break the number 2 into a sum of non-negative integers is
as 1+1).
Degree-1 polynomials may be viewed as having ``perfect eigenregularity''
(note that any degree-1 polynomial in Gaussian variables is itself
distributed precisely as a Gaussian) and so having any number of such
degree-1 polynomials did not pose a problem in \cite{DDS13:deg2count}.}

We surmount this difficulty by using {\bf Decompose-One-Wiener} as a tool
within an improved ``two-level'' decomposition procedure
which we present and analyze below.
This improved decomposition procedure has a stronger guarantee than
{\bf Decompose-One-Wiener} in the following sense:  it breaks
its input polynomial into a sum of products of pairs of polynomials plus
\emph{two} remainder terms $R_{\reg}$ (for ``eigenregular'')
and $R_{\nneg}$ (for ``negligible'').  The $R_{\nneg}$
remainder term is guaranteed to have negligibly small variance,
and the $R_{\reg}$ remainder term is guaranteed to either be zero or else
to be \emph{extremely} eigenregular --
in particular, for an appropriate setting of the input parameters, its
eigenregularity is much ``stronger'' than the number of pairs of polynomials
that are produced in the decomposition.  We
term this improved decomposition procedure {\bf Regularize-One-Wiener} because
of this extremely strong eigenregularity guarantee.

Before giving the formal statement, we note that
intuitively this procedure will be useful because it
``guarantees that we make progress''
for the following reason:  We can always erase the small-variance
$R_{\nneg}$ term at the cost of a small and affordable error,
and \blue{the degree-$q$ $R_{\reg}$  remainder
term is so eigenregular that it will not pose an obstacle to our ultimate
goal of applying the CLT. Thus we have reduced
the original polynomial to a
sum of pairwise products of lower-degree polynomials,
which can each be tackled inductively using similar methods
(more precisely, using the generalization of
procedure {\bf Regularize-One-Wiener} to simultaneously decompose
multiple polynomials which we describe in the next subsection).
}

\begin{theorem} \label{thm:regularize}
Fix any $q \geq 2$.
There is a procedure {\bf Regularize-One-Wiener} which takes as input a polynomial
$I_q(f)$ such that $\Var[I_q(f)] =1$ and input parameters
$\eta_0= 1 \ge \eta_1 \ge \ldots\ge \eta_K$ and $\epsilon$, where
$K = O(1/\epsilon \cdot \log (1/\epsilon))$.  {\bf Regularize-One-Wiener}
runs in poly$(n^q,1/\eta_K,1/\eps)$ time and
has the following guarantee:

\begin{enumerate}

\item Define $M(i) = \frac{O(4^q)}{\eta_i^2} \log(1/\epsilon)$.
{\bf Regularize-One-Wiener} outputs a value
$1 \leq \ell \leq k$, a
set $L=\{(a_{i,j},P_{i,j},Q_{i,j})\}_{i=1,\dots,\ell,j=1,\dots,M(i)}$
of triples, and a pair of polynomials $R_{\reg},R_{\nneg}$
such that $I_q(f) = \sum_{i=1}^\ell \sum_{j=1}^{M(i)} a_{i,j}
P_{i,j}\cdot Q_{i,j} + R_{\reg}+R_{\nneg}$.

\item For each $i,j$ we have $P_{i,j} \in {\cal W}^{q_{i,j,1}}$
and $Q_{i,j} \in {\cal W}^{q_{i,j,2}}$ with
$q_{i,j,1},q_{i,j,2}>0$ and $q_{i,j,1} + q_{i,j,2}=q$
and $\Var[P_{i,j}]=\Var[Q_{i,j}]=\Var[P_{i,j} \cdot Q_{i,j}] = 1$;
moreover, $P_{i,j}$ and $Q_{i,j}$ are over disjoint sets of variables.
In addition,
$R_{\reg},R_{\nneg} \in {\cal W}^q$ and all of $P_{i,j},
Q_{i,j}, R_{\reg}, R_{\nneg}$ are multilinear.

\item The polynomial $R_{\nneg}$ satisfies $\Var[R_{\nneg}] \leq \eps$
and the polynomial $R_{\reg}$ is $\eta_{\ell+1}$-eigenregular,
where we define $\eta_{K+1}=0.$

\item 
For $1 \leq i \leq \ell$ we have
$\sum_{j=1}^{M(i)} (a_{i,j})^2 \leq (2^q/\eta_i)^{4(M(i)-1)}.$

\end{enumerate}

\end{theorem}

We stress that it is crucially important that condition 3 provides
$\eta_{\ell+1}$-eigenregularity rather than
$\eta_{\ell}$-eigenregularity.

\ifnum\confversion=0

\begin{proof}
The procedure {\bf Regularize-One-Wiener} is given below.  We note that
it maintains the invariant
$I_q(f) = \sum_{(a,P,Q) \in L} a \cdot P \cdot Q + I_q(g_i)$
throughout its execution (this is easily
verified by inspection).

    \begin{itemize}

\item [i.] Initialize $L$ to the empty set of triples.

\item [ii.]  Initialize $g_1=f$, so $I_q(g_1) = I_q(f).$

\item [iii.] For $i=1$ to $K$ do the following:

\begin{itemize}

\item [iii(a).]  If $\Var[I_q(g_i)] \leq \eps$ then set $R_{\nneg}=I_q(g)$,
set $R_{\reg}=0$, output $L$, $R_{\reg}$ and $R_{\nneg}$, and exit.

\item [iii(b).] Otherwise, run {\bf Decompose-One-Wiener} with parameters
$\eta_i$ and $\epsilon$ on the polynomial $I_q(\lambda_i g)$,
where $\lambda_i$ is chosen so that $\Var[I_q(\lambda_i g)] =1$.
Let $L_i = \{(c_{i,j}, P_{i,j}, Q_{i,j})\}$
be the set of (at most $M(i)$ many, by Lemma \ref{lemma:algo-primitive2})
triples and $R_i$ be the
polynomial that it outputs.

    \item [iii(c).] If the call to {\bf Decompose-One-Wiener} in step iii(b) returned
    ``\textsf{small remainder}'' then set $L$ to $L \cup L'_i$
where $L'_i = \{({\frac {c_{i,j}}{\lambda_i}}, P_{i,j}, Q_{i,j})\}_{(c_{i,j},
P_{i,j},Q_{i,j}) \in L_i}$, set $R_{\nneg}$
to $R_i/\lambda_i$, set $R_{\reg}$ to 0, output $L$, $R_{\reg}$
and $R_{\nneg}$, and exit.

    \item [iii(d).] Otherwise it must be the case that {\bf Decompose-One-Wiener} returned
``\textsf{eigenregular remainder}.'' In this case, if $\Var [ \sum_{j=1}^{M(i)} c_{i,j} P_{i,j} \cdot Q_{i,j} ] \le \epsilon$, then set $R_{\nneg}$ to $\sum_{(c_{i,j}, P_{i,j}, Q_{i,j}) \in L_i} {\frac {c_{i,j}}{\lambda_i}} \cdot P_{i,j} \cdot Q_{i,j}$ and
$R_{\reg}$ to $R_i/\lambda_i$, output $L$, $R_{\reg}$ and $R_{\nneg}$, and exit.

    \item [iii(e).] Otherwise, set
    $g_{i+1}$ to satisfy
    $I_q(g_{i+1}) = R_i/\lambda_i$, set
   $L$ to $L \cup L'_i$
where $L'_i = \{({\frac {c_{i,j}}{\lambda_i}}, P_{i,j}, Q_{i,j})\}_{(c_{i,j},P_{i,j},Q_{i,j}) \in L_i}$, increment $i$, and go to the next
    iteration of step (iii).

\end{itemize}

\end{itemize}

For Property (1), we observe that the claimed bound on $M(i)$ follows immediately
from  part (3) of Lemma \ref{lemma:algo-primitive2}).  The rest of Property (1)
follows from the invariant and inspection of steps iii(c) and iii(d).
Property (2)
follows directly from part (2) of Lemma \ref{lemma:algo-primitive2}.

To establishing the remaining properties we will use the following claim:

\begin{claim} \label{claim:lambda-bound}
For each $i$ we have $\Var[I_q(g_i)] \leq (1-\eps)^{i-1}$.
\end{claim}
\begin{proof}
The proof is by induction on $i$.  The claim clearly holds for $i=1.$
For the inductive step, observe that the only way the procedure reaches step iii(e)
and increments $i$ is if the execution of {\bf Decompose-One-Wiener} on $I_q(\lambda_i g)$
returned ``\textsf{eigenregular remainder}'' and the decomposition
$
I_q(\lambda_i g) = \sum c_{i,j} P_{i,j} Q_{i,j} + R_i$ has $\Var[\sum c_{i,j} P_{i,j}
Q_{i,j}] > \eps$, and hence (by part (5) of Lemma \ref{lemma:algo-primitive2})
$\Var[R_i] \leq (1-\eps)\Var[I_q(\lambda_i g_i)].$  Consequently in this case we have
$\Var[I_q(g_{i+1})]=\Var[R_i/\lambda_i] \leq (1-\eps) \Var[I_q(g_i)]$, which inductively is at most $(1-\eps)^i$ as desired.
\end{proof}

Note that this claim immediately gives that $\lambda_i \geq 1$ for all $i$,
which together with part (6) of Lemma (\ref{lemma:algo-primitive2}) gives Property (4).

It remains only to establish Property (3).  Note that by Claim
\ref{claim:lambda-bound} it must be the case that the algorithm halts and
exits at some iteration of either
step iii(a), iii(c), or iii(d) --- if it has not already exited by the
time $i$ reaches $K$, since $\Var[I_q(g_i)] \leq (1-\eps)^{i-1}$
once it reaches $i=K$ it will exit in step iii(a).
We consider the three possibilities in turn.
If it exits at Step iii(a) then clearly Property (3) is satisfied.
If it exits at Step iii(c) then by Lemma (\ref{lemma:algo-primitive2})
we have that $\Var[R_i] \leq \eps$; since $\lambda_i \geq 1$
this means that $\Var[R_{\nneg}]=\Var[R_i/\lambda_i] \leq \eps$
and again Property (3) holds.  Finally, if it exits at Step iii(d)
during the $i$-th iteration of the loop then observe that the value
of $\ell$ is $i-1$ (since $L_i$ is \emph{not} added on to $L$).
Lemma (\ref{lemma:algo-primitive2}) guarantees that $R_i$ (and hence
$R_{\reg})$ is $\eta_i$-eigenregular, i.e. $\eta_{\ell+1}$-eigenregular,
and as above the fact that $\lambda_i \geq 1$ ensures that
$\Var[R_{\nneg}] \leq \eps$, and Property (3) holds in this case as well.
This concludes the proof of
Theorem \ref{thm:regularize}.
\end{proof}

\fi

\ignore{

%Essentially, in every iteration, we
%either end up in the ``\textsf{small remainder}" case in which case we
%stop the algorithm. The other stopping condition is that in  a particular
%iteration (say the $i^{th}$),  $\Var  \sum_{j=1}^{M(i)} c_{i,j} P_{i,j}
%\cdot Q_{i,j} \le \epsilon$. Then observe that
%
%$$P = \sum_{1 \le \ell <i} \sum_{j=1}^{M(\ell)} \frac{ c_{\ell j} P_{\ell j}
%\cdot Q_{\ell j}}{\lambda_{\ell}}  +  \sum_{j=1}^{M(i)} \lambda_{i} c_{i,j}
%P_{i,j} Q_{i,j} + \lambda_{i} R. $$
%
%We now want to control the variance of
%$ \sum_{j=1}^{M(i)} \lambda_i c_{i,j} P_{i,j} Q_{i,j}$.
%Observe that by the correctness of {\bf Decompose-One-Wiener}, we have that
%
%$$
%\Var  \left[ \sum_{j=1}^{M(i)} \lambda_ic_{i,j} P_{i,j} Q_{i,j} + \lambda_i
%R \right] \le (1-\eta)^i
%$$
%Of course, we normalize here so that $\Var [\sum_{j=1}^{M(i)} c_{i,j} P_{i,j}
%Q_{i,j} +  R] =1$ (In other words, $\lambda \le (1-\eta)^{i/2}$.) Now, by the
%fourth bullet of Claim~\ref{clm:algo-primitive2}, we can deduce
%that $\Var\sum_{j=1}^{M(i)} c_{i,j} P_{i,j} Q_{i,j} \le \epsilon$.
%
%Thus, the only way an iteration continues is if $\Var P_i \le (1-\epsilon)
%\Var P_{i-1}$. This means that the algorithm can go on for at most
%$(1/\epsilon) \cdot \log (1/\epsilon)$ steps.
%
}

\subsection{Decomposing a $k$-tuple of multilinear elements of the $q$-th
Wiener chaos} \label{sec:manypolyonewiener}

In this section we generalize the {\bf Regularize-One-Wiener} procedure to
simultaneously decompose \blue{multiple} polynomials that all belong to ${\cal W}^q$.
Even though our ultimate goal is to decompose a single degree-$d$ Gaussian
polynomial, we require a procedure that is capable of handling
many polynomials because even decomposing a single degree-$d$
polynomial using {\bf Regularize-One-Wiener} will give rise to many
lower-degree polynomials which all need to be
decomposed in turn.
\ifnum\confversion=0
\blue{(This is also the reason why we must prove Theorem
\ref{thm:general}, which deals with $k$ Gaussian polynomials, in order to
ultimately obtain Theorem \ref{thm:main-decomp}, which decomposes a single Gaussian
polynomial.)}
\fi

A natural approach to decompose $r$ polynomials
$I_q(f_1),\dots,I_q(f_r) \in {\cal W}^q$ is simply to
run {\bf Regularize-One-Wiener} $r$ separate times.  However, this simpleminded
approach could well result in different values $\ell_1,\dots,\ell_r$
being obtained from the $r$ calls, and hence in
different levels of eigenregularity for the $r$ ``remainder''
polynomials $R_{1,\reg},\dots,R_{r,\reg}$ that are constructed.
This is a problem because some of the calls may yield a relatively large
eigenregularity parameter, while other calls may generate
very many polynomials (and a much smaller eigenregularity parameter).
Since the CLT can only take advantage of the largest eigenregularity parameter,
the key advantage of {\bf Regularize-One-Wiener} --- that the
number of polynomials it produces compares favorably with the eigenregularity
of these polynomials --- is lost.

\ifnum\confversion=1
We get around this with a procedure called
{\bf MultiRegularize-One-Wiener}.  It takes as input an
$r$-tuple of polynomials
$(I_q(f_1),\dots,I_q(f_r))$ \blue{(that all belong to one fixed Wiener chaos)} and
input parameters $\eta_0= 1 \ge \eta_1 \ge \ldots\ge \eta_K$ and $\epsilon$.
Crucially, it guarantees that the \emph{overall} number of polynomials
that are produced from all the $r$ decompositions compares favorably with
the overall eigenregularity parameter that is obtained.
Intuitively, {\bf MultiRegularize-One-Wiener} augments the {\bf Regularize-One-Wiener} procedure
with ideas from the decomposition procedure for $k$-tuples of degree-2
polynomials that was given in \cite{DDS14junta}
(and which in turn built on ideas from \cite{GOWZ10}).  See the full version for
details.
\fi

\ifnum\confversion=0
We get around this with the
{\bf MultiRegularize-One-Wiener} procedure that is presented and
analyzed below.  It augments the {\bf Regularize-One-Wiener} procedure
with ideas from the decomposition procedure for $k$-tuples of degree-2
polynomials that was presented and analyzed in \cite{DDS14junta}
(and which in turn built on ideas from the decomposition of
\cite{GOWZ10} for simultaneously dealing with multiple degree-1
polynomials, using a different notion of ``eigenregularity'').
Crucially, it guarantees that the \emph{overall} number of polynomials
that are produced from all the $r$ decompositions compares favorably with
the overall eigenregularity parameter that is obtained.

\begin{theorem} \label{thm:regularize-many}
Fix any $q \geq 2$.
There is a procedure {\bf MultiRegularize-One-Wiener} which takes as input an
$r$-tuple of polynomials
$(I_q(f_1),\dots,I_q(f_r))$ such that $\Var[I_q(f_i)]=1$ for all $i$,
and input parameters
$\eta_0= 1 \ge \eta_1 \ge \ldots\ge \eta_K$ and $\epsilon$, where
$K = O(r/\epsilon \cdot \log (1/\epsilon))$.  {\bf MultiRegularize-One-Wiener}
runs in poly$(n^q,1/\eta_K,r/\eps)$ time and
has the following guarantee:

\begin{enumerate}

\item Define $M(i) = O(\frac{4^q}{\eta_i^2} \log(1/\epsilon))$.
{\bf MultiRegularize-One-Wiener} outputs an index $t$ with $0 \leq t \leq K$
and for each $s \in [r]$ a set $L_s$ of triples
$
\{(a_{s,i,j},P_{s,i,j},Q_{s,i,j})\}_{i=1,\dots,t,j=1,\dots,M(i)}$
and a pair of polynomials $R_{s,\reg},
R_{s,\nneg}$,
such that
\ifnum\confversion=0
\begin{equation}
\label{eq:good}
I_q(f_s) = \sum_{i=1}^t \sum_{j=1}^{M(i)} a_{s,i,j}
P_{s,i,j}\cdot Q_{s,i,j} + a_{s,\reg} \cdot R_{s,\reg}+R_{s,\nneg}.
\end{equation}
\fi
\ifnum\confversion=1
$I_q(f_s) = \sum_{i=1}^t \sum_{j=1}^{M(i)} a_{s,i,j}
P_{s,i,j}\cdot Q_{s,i,j} + a_{s,\reg} \cdot R_{s,\reg}+R_{s,\nneg}.
$
\fi

\item For each $s,i,j$ we have $P_{s,i,j} \in {\cal W}^{q_{s,i,j,1}}$
and $Q_{s,i,j} \in {\cal W}^{q_{s,i,j,2}}$ with
$q_{s,i,j,1},q_{s,i,j,2}>0$ and $q_{s,i,j,1} + q_{s,i,j,2}=q$
and $\Var[P_{s,i,j}]=\Var[Q_{s,i,j}]= \Var[P_{s,i,j} \cdot
Q_{s,i,j}] = 1$.
Similarly we have $R_{s,\reg},R_{s,\nneg} \in {\cal W}^q$, and $\Var[R_{s,\reg}]=1.$
Moreover
$P_{s,i,j}$ and $Q_{s,i,j}$ are over disjoint sets of variables,
and all of $P_{s,i,j},Q_{s,i,j},R_{s,\reg}$ and $R_{s,\nneg}$
are multilinear.

\item For each $s$ we have that
 $\Var[R_{s,\nneg}] \leq \eps$ and that
 $a_{s,\reg} \cdot R_{s,\reg}$ is $\eta_{t+1}$-eigenregular,
where we define $\eta_{K+1}=0.$
\ignore{
%\red{
%For each $S \in {\cal S}$ we have that $R_{s,\reg}=0$
%and for each $s \in \overline{S}$ we have that
%that $R_{s,\reg}$ is $\eta_{t+1}$-eigenregular.}
}

\item \red{
For $1 \leq s \leq r$ and $1 \leq i \leq t$ we have
$\sum_{j=1}^{M(i)} (a_{s,i,j})^2 \leq (2^q/\eta_i)^{4(M(i)-1)}.$
}

\end{enumerate}

\end{theorem}

\begin{proof}
Similar to {\bf Regularize-One-Wiener},
the procedure {\bf MultiRegularize-One-Wiener} maintains the
invariant that for each $s \in [r]$, we have
$I_q(f_s) = \sum_{(a,P,Q) \in L_s} a \cdot P \cdot Q + I_q(g_{s,i})$ throughout
its execution.

Before giving the detailed description we provide some
useful points to keep in mind.
The set $[r] \setminus \mathbf{live}$ contains the indices
of those polynomials for which the desired decomposition has already been
achieved, while $\mathbf{live}$ contains those polynomials that are
still being decomposed.
The variable $\mathbf{hit}_s$ maintains the number of
times that the decomposition procedure {\bf Decompose-One-Wiener}
has been applied to $I_q(g_{s,i})$ for some $i$.

Here is the procedure {\bf MultiRegularize-One-Wiener}:
\begin{itemize}

\item[i.] For all $s \in [r]$ initialize $\mathbf{hit}_s$  to be 0,
initialize $L_s$ to be the empty set of triples, and initialize
$g_{s,1}=f_s$, so $I_q(g_{s,1})=I_q(f_s).$  Initialize
the set $\mathbf{live}$ to be $[s]$.

\item[ii.]  For $i=1$ to $K$ do the following:

\begin{itemize}

\item [ii(a).]  For each $s \in \mathbf{live}$,
if $\Var[I_q(g_{s,i})] \leq \eps$ then set
$R_{s,\nneg}=I_q(g_{s,i})$, set $a_{s,\reg}=0$ and set $R_{s,\reg}$ to be any unit variance element of $\mathcal{W}^{q}$ (the choice of $R_{s, \reg}$ is immaterial),
and remove $s$ from $\mathbf{live}$.

\item [ii(b).]
If $\mathbf{live}$ is empty then for each $s \in [r]$
output the set $L_s$ and the pair $R_{s,\reg},R_{s,\nneg}$, and exit.
Otherwise,
for each $s \in \mathbf{live}$,
run {\bf Decompose-One-Wiener} with parameters
$\eta_i$ and $\epsilon$ on the polynomial $I_q(\lambda_{s,i} g_s)$,
where $\lambda_{s,i}$ is chosen so that $\Var[I_q(\lambda_{s,i} g_s)] =1$.
Let $L_{s,i} = \{(c_{s,i,j}, P_{s,i,j}, Q_{s,i,j})\}$
be the set of (at most $M(i)$ many, by Lemma \ref{lemma:algo-primitive2})
triples and $R_{s,i}$ be the
polynomial that it outputs.

\item[ii(c).] If the call to {\bf Decompose-One-Wiener}
returned
``\textsf{small remainder}'' for any polynomial
$I_q(\lambda_{s,i}g_s),$ then \ignore{set the Boolean
variable $\mathbf{flag}$ to 1, and} for each such $s$  set $L_s$
to $L_s \cup L_{s,i}'$
where $L'_{s,i} = \{({\frac {c_{s,i,j}}{\lambda_{s,i}}},
P_{s,i,j}, Q_{s,i,j})\}_{(c_{s,i,j},
P_{s,i,j},Q_{s,i,j}) \in L_{s,i}}$,
set $R_{s,\nneg}$
to $R_{s,i}/\lambda_{s,i}$, set $a_{s,\reg}=0$ and  $R_{s,\reg}$ to be any unit variance element of $\mathcal{W}^{q}$ (as before, the choice of $R_{s, \reg}$ is immaterial), and remove $s$
from $\mathbf{live}.$
\ignore{
%If the call to {\bf Decompose-One-Wiener} did not return
%``\textsf{small remainder}'' for any polynomial
%$I_q(\lambda_{s,i}g_s),$ then set $\mathbf{flag}$ to 0.
}

\item[ii(d).] \ignore{If $\mathbf{flag}=0$, then:  }If for
every $s \in \mathbf{live}$
it is the case that
$\Var[
\sum_{(c_{s,i,j},P_{s,i,j},Q_{s,i,j}) \in L_{s,i}}
c_{s,i,j} P_{s,i,j} \cdot Q_{s,i,j}
] \le \epsilon$,
then set $R_{s,\nneg}$ to
$\sum_{(c_{s,i,j},P_{s,i,j},Q_{s,i,j}) \in L_{s,i}}
{\frac {c_{s,i,j}}{\lambda_{s,i}}} P_{s,i,j} \cdot Q_{s,i,j}$. Also, set $a_{s,\reg} = \sqrt{\Var (R_{s,i}/\lambda_{s,i}) }$ and
and $R_{s,\reg}=R_{s,i}/(\lambda_{s,i}  \cdot a_{s,\reg})$. For each
$s \in [r]$ output the set $L_s$, and the triple $a_{s,\reg},R_{s,\reg},R_{s,\nneg}$,
and exit.

\item[ii(e).] \ignore{If $\mathbf{flag}=1$, then:}Otherwise,
for each $s \in \mathbf{live}$
such that
$\Var[
\sum_{(c_{s,i,j},P_{s,i,j},Q_{s,i,j}) \in L_{s,i}}
c_{s,i,j} P_{s,i,j} \cdot Q_{s,i,j}
] > \epsilon$,
increase $\mathbf{hit}_s $ by $1$, set $L_s$ to $L_s \cup L_{s,i}'$
where $L'_{s,i} = \{({\frac {c_{s,i,j}}{\lambda_{s,i}}},
P_{s,i,j}, Q_{s,i,j})\}_{(c_{s,i,j},
P_{s,i,j},Q_{s,i,j}) \in L_{s,i}}$, and
set $g_{s,i+1}$ to satisfy $I_q(g_{s,i+1}) = R_{s,i}/\lambda_{s,i}.$
Increment $i$ and go to the next iteration of step (ii).

\end{itemize}

\end{itemize}

Property (1) follows from the discussion preceding the algorithm
description and inspection of step ii(d).
Property (2)
follows directly from part (2) of Lemma \ref{lemma:algo-primitive2} (Note that the algorithm ensures that $R_{s,\reg}$ has unit variance).

We have the following analogue of Claim \ref{claim:lambda-bound}:
\begin{claim} \label{claim:lambda-bound-many}
At each stage in the execution of the algorithm,
for each $s \in \mathbf{live}$ we have $\Var[I_q(g_{s,i})]
\leq (1-\eps)^{\mathbf{hit}_s}$.
\end{claim}
\begin{proof}
The proof is an easy adaptation of the proof of Claim \ref{claim:lambda-bound},
using the criterion for incrementing $\mathbf{hit}_s$ that is employed
in step ii(e).
\end{proof}
Claim \ref{claim:lambda-bound-many} implies that for each $s \in \mathbf{live}$
we have $\lambda_{s,i} \geq 1$, so as in the
proof of Claim \ref{claim:lambda-bound} we get that Property (4) holds.

Observe that if an index $s$ is removed from
$\mathbf{live}$ (either in Step ii(a) or Step ii(c)),
then the polynomial $R_{s,\reg}$ is 0-eigenregular, and since
$\lambda_{s,i} \geq 1$, the polynomial $R_{s,\nneg}$ has
$\Var[R_{s,\nneg}] \leq \eps$.
Hence as a consequent of the above-mentioned invariant,
it is easily verified that
each $s \in [r] \setminus \mathbf{live}$ satisfies (\ref{eq:good}).

The last step is to establish Property (3).  The key observation is that
each time the algorithm increments $i$ in step ii(e) and returns to step ii(a),
at least one $s \in [r]$ must have had $\mathbf{hit}_s$ incremented.
Once a given value of $s$ has $\mathbf{hit}_s$ reach $O(1/\eps \cdot
\log(1/\eps))$, by Claim \ref{claim:lambda-bound-many} it will be the
case that $s$ is removed from $\mathbf{live}$ in Step ii(a).
Since $K=O(r/\eps \cdot \log(1/\eps))$, it follows that the algorithm
must halt and exit in some iteration of step ii(b) or ii(d).
If the algorithm exits in step ii(b) then it is clear from the above
discussion that Property (3) holds.  Finally, if the algorithm
exits in step ii(d), then similar to the final paragraph of the proof of
Theorem \ref{thm:regularize}, the value of $t$ is $i-1$
(since for the elements $s \in \mathbf{live}$ at the start of
that execution of step ii(d), the elements of $L_{s,i}$ are not
added on to $L_s$).  Similar to before we get that
Lemma (\ref{lemma:algo-primitive2}) guarantees that $R_{s,i}$ (and hence
$R_{s,\reg})$ is $\eta_i$-eigenregular, i.e. $\eta_{t+1}$-eigenregular,
the fact that $\lambda_{s,i} \geq 1$ ensures that
$\Var[R_{s,\nneg}] \leq \eps$, and hence Property (3) holds.
The proof is complete.
\end{proof}

\fi

\subsection{Beyond the homogeneous case: handling multiple levels
of Wiener chaos} \label{sec:wholething}

\ifnum\confversion=1
In this subsection we sketch our most involved
decomposition procedure,
{\bf \blue{MultiRegularize-Many-Wieners}}, for decomposing
a $k(d+1)$-tuple consisting of $k$ elements from the
$j$-th Wiener chaos for each $j=0,\dots,d$.
We begin with an informal description of how the decomposition procedure works.
Let $p_1,\dots,p_k$ be $k$ degree-$d$ multilinear Gaussian
polynomials.  Each $p_i$ has a unique expansion
in terms of symmetric $q$-tensors $f_{i,q} \in {\cal H}^{\odot q}$ as
$
p_i = \sum_{q=0}^d p_{i,q},$ where
$p_{i,q}=I_q(f_{i,q}).
$
For $2 \leq q \leq \blue{d-1}$ let $OLD_q$ denote the set of polynomials
$\{I_q(f_{i,q})\}_{i=1,\dots,k}.$

The high-level idea of the decomposition is to ``work downward''
from higher to lower levels of the Wiener chaos in successive stages,
at each stage using {\bf MultiRegularize-One-Wiener} to simultaneously decompose
all of the polynomials at the current level.  By carefully choosing
the eigenregularity parameters at each stage we can ensure
that at the end of the decomposition we are left with
a collection of ``not too many'' polynomials (\blue{corresponding to} the $A_{i,j,\ell}$'s of Theorem
\ref{thm:main-decomp}) all of which are highly eigenregular\ignore{ (so
that we will be able to apply the CLT to this collection)}.

In a bit more detail,
in the first stage we simultaneously decompose the $k$ degree-$d$ polynomials
$I_{d}(f_{1,d}),$ $ \dots, I_d(f_{k,d})$
using the {\bf MultiRegularize-One-Wiener} algorithm with parameters
$1 = \eta_0 \gg \cdots \gg \eta_{K}$ and
$\eps$.
This generates $k$ polynomials in ${\cal W}^{d}$
that are each $\eta_{t+1}$-eigenregular, for some $1 \leq t
\leq \blue{K}$, where $K \leq O_{k,d,\eps}(1)$;
intuitively, these should be thought of as ``extremely eigenregular'' polynomials.
Let $REG$ denote this set of polynomials \blue{(they will not be used again in the decomposition).} It also generates,
for each $1 \leq q \leq d-1$, ``not too many'' (at most
$O_{k,d,\eps,\eta_{t}}(1)$)
new polynomials in ${\cal W}^{q}$; let $NEW_{q}$ denote this set of polynomials.
The key qualitative point is that the size of each
$NEW_q$ depends on $\eta_{t}$ while the
eigenregularity of the polynomials \blue{in $REG$} is $\eta_{t+1}$.
Thus if $\eta_{t+1}$ is much less than $\eta_t$, the number of newly
introduced polynomials compares favorably with the  \blue{eigenregularity bound $\eta_{t+1}.$}

We have thus ``dealt with the degree-$d$ part of the input'' since the only
remaining degree-$d$ polynomials are \blue{in $REG$ and are} extremely regular.  Next, we recursively apply
the above approach to handle the lower-degree part, including \blue{both the original
lower-degree components from $OLD_q$ and the new
lower-degree polynomials from $NEW_q$ that were introduced in dealing with the degree-$d$ part}.
The crux of the analysis is to argue that there is a suitable choice of parameters at
each stage that allows this procedure to be carried out ``all the way down,'' so that 
\blue{the total number of polynomials} that are ever produced in the
analysis is far smaller than $1/\eta$, where $\eta$ is the largest
eigenregularity of any of the final polynomials.  It turns out to be difficult
to argue this formally using the ``top-down'' view on the
decomposition procedure that we have adopted above.  Instead, in the detailed proof
proof which we give in the full version,
we take a ``bottom-up'' view of the decomposition procedure:  we first
show that it can be successfully carried out for low-degree polynomials,
and use this fact to show that it
can be successfully carried out for higher-degree polynomials.

\fi

\ifnum\confversion=0

In this subsection we describe and analyze our most involved
decomposition procedure,
{\bf \blue{MultiRegularize-Many-Wieners}}, for decomposing
a $k(d+1)$-tuple consisting of $k$ elements from the
$j$-th Wiener chaos for each $j=0,\dots,d$.
We will obtain Theorem \ref{thm:main-decomp} in the
following subsection using {\bf \blue{MultiRegularize-Many-Wieners}}.

\medskip

\noindent {\bf An informal ``top-down'' description.}
We begin with an informal description of how the decomposition procedure works.
Let $p_1,\dots,p_k$ be $k$ degree-$d$ multilinear Gaussian
polynomials.  Each $p_i$ has a unique expansion
in terms of symmetric $q$-tensors $f_{i,q} \in {\cal H}^{\odot q}$ as
$
p_i = \sum_{q=0}^d p_{i,q}, \quad \text{where} \quad
p_{i,q}=I_q(f_{i,q}).
$
For $2 \leq q \leq d$ let $OLD_q$ denote the set of polynomials
$\{I_q(f_{i,q})\}_{i=1,\dots,k}.$

The high-level idea of the decomposition is to ``work downward''
from higher to lower levels of the Wiener chaos in successive stages,
at each stage using {\bf MultiRegularize-One-Wiener} to simultaneously decompose
all of the polynomials at the current level.  By carefully choosing
the eigenregularity parameters at each stage we can ensure
that at the end of the decomposition we are left with
a collection of ``not too many'' polynomials (the $A_{i,j,\ell}$'s of Theorem
\ref{thm:main-decomp}) all of which are highly eigenregular\ignore{ (so
that we will be able to apply the CLT to this collection)}.

In a bit more detail,
in the first stage we simultaneously decompose the $k$ degree-$d$ polynomials
$I_{d}(f_{1,d}), \dots, I_d(f_{k,d})$
using the {\bf MultiRegularize-One-Wiener} algorithm with parameters
$1 = \eta^{(d)}_0 \gg \cdots \gg \eta^{(d)}_{K^{(d)}}$ and
$\eps^{(d)}$.
This generates:

\begin{itemize}

\item $k$ polynomials in ${\cal W}^{d}$
that are each $\eta^{(d)}_{t^{(d)}+1}$-eigenregular, for some $1 \leq t^{(d)}
\leq K^{(d)}-1$ where $K^{(d)} \leq O(k/\eps^{(d)} \cdot \log(1/\eps^{(d)}))$
(intuitively, these should be thought of as ``extremely eigenregular'' polynomials);
these are the $R_{s,\reg}$ polynomials.
Let $REG_{d}$ denote this set of polynomials.  It also generates

\item For each $1 \leq q \leq d-1$, ``not too many'' (at most
$k \cdot O((k/\eps^{(d)})\log(1/\eps^{(d)}))\cdot (1/(\eta^{(d)}_{t^{(d)}})^2)
\log(1/\eps^{(d)})$) new polynomials in ${\cal W}^{q}$;
these are the $P_{s,i,j}$ and $Q_{s,i,j}$ polynomials that lie
in ${\cal W}^{q}.$  Let $NEW_{q}$ denote this set of polynomials.

\end{itemize}

Let $ALL_{d-1}$ denote the union of $OLD_{d-1}$ and $NEW_{d-1}$.
Note that every element of $ALL_{d-1}$ belongs to ${\cal W}^{d-1}$,
and that the size of $|ALL_{d-1}|$ is upper bounded by
\[
k \cdot O((k/\eps^{(d)})\log(1/\eps^{(d)}))\cdot (1/(\eta^{(d)}_{t^{(d)}})^2)
\log(1/\eps^{(d)}).
\]

The key qualitative point is that after this first stage of
decomposition, the number of eigenregular polynomials that have been produced is
$|REG_d| \leq k$, and each such polynomial is
$\eta^{(d)}_{t^{(d)}+1}$-eigenregular, while the number of polynomials of
lower degree that remain to be dealt with (the pieces that came from the
original polynomial plus the elements of $NEW_2 \cup \cdots \cup NEW_{d-1}$)
is upper bounded in terms of $k,\eps^{(d)}$ and $\eta^{(d)}_{t^{(d)}}$.

In the second stage we simultaneously decompose all the elements of
$ALL_{d-1}$ by applying the {\bf MultiRegularize-One-Wiener} algorithm
to those $|ALL_{d-1}|$ polynomials, using input parameters
$1=\eta_0^{(d-1)} \geq \eta_1^{(d-1)} \geq \cdots \geq
\eta_{K^{(d-1)}}^{(d-1)}$ where $K^{(d-1)} = O(|ALL_{d-1}|/\eps^{(d-1)}
\cdot \log(1/\eps^{(d-1)}))$ and $\eps^{(d-1)}.$
This generates

\begin{itemize}

\item at most $|ALL_{d-1}|$ polynomials in ${\cal W}^{d-1}$
that are each $\eta^{(d-1)}_{t^{(d-1)}+1}$-eigenregular, for some $1 \leq t^{(d-1)}
\leq K^{(d-1)}-1$
(intuitively, these should be thought of as ``extremely eigenregular'' polynomials);
these are the $R_{s,\reg}$ polynomials.
Let $REG_{d-1}$ denote this set of polynomials.  It also generates

\item For each $1 \leq q \leq d-2$, ``not too many'' (at most
$|ALL_{d-1}| \cdot K^{(d-1)}  (1/(\eta^{(d-1)}_{t^{(d-1)}})^2)
\log(1/\eps^{(d-1)})$) new polynomials in ${\cal W}^{q}$;
these are the $P_{s,i,j}$ and $Q_{s,i,j}$ polynomials that lie
in ${\cal W}^{q}.$  Add these polynomials to the set $NEW_{q}$.

\end{itemize}

Similar to above, the key qualitative point to observe is that the
number of eigenregular
polynomials that have been produced is $|REG_d|+|REG_{d-1}|$, which is a
number depending only on $k,\eps^{(d)},
\eps^{(d-1)},$ and $1/\eta^{(d)}_{t^{(d)}}$,
while the eigenregularity of each such polynomial is at most
$
\max\{\eta^{(d)}_{t^{(d)}+1},
\eta^{(d-1)}_{t^{(d-1)}+1}\},
$
where the first expression inside the max
comes from the polynomials in $|REG_d|$ and the second from the polynomials
in $|REG_{d-1}|$.  By setting the $\eta$-parameters so that $\eta^{(d)}_{t^{(d)}+1}$
and $\eta^{(d-1)}_{t^{(d-1)}+1}$ are both much smaller than $\eta^{(d)}_{t^{(d)}}$, we can ensure
that the number of polynomials that are produced compares favorable with their eigenregularity.

Continuing in this fashion, the crux of the analysis is to argue that this
can be done ``all the way down,'' so that the total number of
polynomials that are ever produced in the
analysis is far smaller than $1/\eta$, where $\eta$ is the largest
eigenregularity of any of these polynomials.  However,
it is somewhat awkward to argue this using the ``top-down'' view on the
decomposition procedure that we have adopted so far.  Instead, in the formal
proof which we give below,
we take a ``bottom-up'' view of the decomposition procedure:  we first
show that it can be successfully
carried out for low-degree polynomials, and use this fact to show that it
can be successfully carried out
for higher-degree polynomials.

\fi

\ifnum\confversion=0

\subsubsection{The {\bf \blue{MultiRegularize-Many-Wieners}} procedure and its
analysis}

Now we present and analyze the actual {\bf \blue{MultiRegularize-Many-Wieners}} procedure.
Theorem \ref{thm:general} gives a performance bound on
the procedure.
Its proof will be by
induction on the degree:  we first establish the result for degree 2
and then use the fact that the theorem holds for degrees $2,\dots,d-1$
to prove the result for degree $d$.
\ignore{(We note that while Theorem \ref{thm:main-decomp} will be essentially a
special case of Theorem \ref{thm:general}, the more
general Theorem \ref{thm:general} is easier to prove using induction
because of the increased flexibility that is afforded by the
``auxiliary function'' $\beta$ in the theorem statement.)
}

\begin{theorem}
\label{thm:general}
Fix $d \geq 2$ and fix any non-increasing computable
function $\beta:  [1,\infty) \to (0,1)$ that satisfies $\beta(x) \leq 1/x$.
There is a procedure {\bf \blue{MultiRegularize-Many-Wieners}}$_{d,\beta}$
with the following properties.
The procedure takes as input the following:

\begin{itemize}

\item It is given $k$ lists of $d+1$ multilinear Gaussian polynomials; the
$s$-th list is $p_{s,0},\dots,p_{s,d}$ where $p_{s,q} \in {\cal W}^q$ and
\grade{
$\Var[p_{s,q}]=1$ for $1 \leq q \leq d.$}

\item It also takes as input
a parameter $\tau>0$.

\end{itemize}

The procedure runs in $\poly(n^d) \cdot O_{k,d,\tau}(1)$
time and outputs, for each input polynomial $p_{s,q}$,
a polynomial
$\Outer(p_{s,q})$
and a collection of polynomials that we denote
$\{\Inner(p_{s,q})_\ell\}_{\ell = 1,\dots, \num(p_{s,q})}$; here $\num(p_{s,q})$ is the number of arguments of the
polynomial $\Outer(p_{s,q}).$  (``$\Outer$'' stands for ``outer''
and ``$\Inner$'' stands for ``inner''.)

For $s=1,\dots,k$, $\blue{0} \leq q \leq d$ and $x \in \R^n$,
let
\begin{equation} \label{eq:tildep}
\tilde{p}_{s,q}(x)=
\Outer(p_{s,q})\left(
\Inner(p_{s,q})_1(x),\dots, \Inner(p_{s,q})_{\num(p_{s,q})}(x)\right)
\end{equation}

(Intuitively, each $\tilde{p}_{s,q}$ is a polynomial that has been
decomposed into constituent sub-polynomials
$\Inner(p_{s,q})_1,\dots, \Inner(p_{s,q})_{\num(p_s)_{q}}$;
$\tilde{p}_{s,q}$
is meant to be a good approximator for
$p_{s,q}$. The following conditions make this precise.)

The following conditions hold:

\begin{enumerate}

\item For each $s \in [k], \blue{0} \leq q \leq d$ the polynomial
$\tilde{p}_{s,q}(x)$ belongs to the $q$-th Wiener chaos ${\cal W}^q$. Additionally, each polynomial $\Inner(p_{s,q})_\ell$ \blue{with $q \geq 1$}
lies in ${\cal W}^j$ for some $1 \leq j \leq d$ and
has $\Var[\Inner(p_{s,q})_\ell] = 1.$

\item For each $s \in [k], \blue{0} \leq q \leq d$, we have
$
\Var[p_{s,q}-\tilde{p}_{s,q}] \leq \tau.
$

\item
Each polynomial $\Outer(p_{s,q})$ is a multilinear polynomial in its
$\num(p_{s,q})$ arguments.  Moreover,there exists $N = N_{\beta}(k,d,\tau)$ and $M=M_{\beta}(k,d,\tau)$ such that if $\mathrm{Coeff}(p_{s,q})$ denotes the sum of the absolute values of
the coefficients of $\Outer(p_{s,q})$, then $\sum_{s,q} \mathrm{Coeff}(p_{s,q}) \le M$ and
$\sum_{s,q} \num(p_{s,q}) \le N$.

\item Further, let $\mathrm{Num} = \sum_{s=1,\dots,k,q=0,\dots,d}
\num(p_{s,q})$  and $\mathrm{Coeff} = \sum_{s=1,\dots,k,q=0,\dots,d}
\mathrm{Coeff}(p_{s,q})$.
\ignore{be the total number of
polynomials $\Inner(p_{s,q})_\ell$ which are arguments to all the
$\Outer(p_{s,q})$
polynomials.  }Then,
each polynomial $\Inner(p_{s,q})_\ell$ is $\beta(\mathrm{Num}+\mathrm{Coeff})$-eigenregular.

\end{enumerate}

\end{theorem}

After proving Theorem \ref{thm:general}, in the next subsection
we will obtain Theorem \ref{thm:main-decomp} from it as a special case, 
by writing the degree-$d$ polynomial $p$ as $\sum_{q=0}^d
p_{q}$ where $q \in {\cal W}^q$ and applying
Theorem \ref{thm:general} to $(p_0,\dots,p_d).$

\paragraph{Base case:  Proof of Theorem \ref{thm:general} in the case $d=2$.}
Fix any
non-increasing function $\beta: [1,\infty) \to (0,1)$
that satisfies $\beta(x) \leq 1/x.$
The main step is to use the
{\bf MultiRegularize-One-Wiener} procedure
on the vector of (at most $k$) polynomials
$\{p_{s,2}\}_{s \in [k]}$;
we now specify precisely how to set the parameters of this procedure.
Fix $\eps = \tau$ and let $K=O(k/\eps \cdot \log(1/\eps))$
as in the statement of Theorem \ref{thm:regularize-many}.
We define the parameters $1=\eta_0 \geq \eta_1 \cdots \geq \eta_K > n_{K+1}=0$  as
follows:  for $\grade{t}=0,\dots,K-1$, we have

\begin{equation} \label{eq:eta-spec}
\grade{\eta_{t+1} :=\beta
\left(
C
{\frac {k^2} \eps} \cdot
{\frac 1 {\eta_t^2}} \cdot
(\log 1/\eps)^2 + C' \cdot k^3 \cdot \left( \frac{4}{\eta_t} \right)^{C' \cdot \frac{1}{\eta_t^2} \cdot \log(1/\epsilon)}
\right)}
\end{equation}
where $C, C'\ge 1$ are absolute constants defined below (see (\ref{eq:Nbound}) and (\ref{eq:coeffbound}); intuitively, the first term correspond to an upper bound on $\mathrm{Num}$ whereas the second term corresponds to an upper bound on $\mathrm{Coeff}$).
(Note the assumption that $\beta(x) \leq 1/x$ implies that indeed
$\eta_{t+1} \leq \eta_t$.)
When called with these parameters
on the vector of polynomials
$\grade{ \{p_{s,2}\}_{s \in [k]} } $,
by Theorem \ref{thm:regularize-many}
{\bf MultiRegularize-One-Wiener} outputs an index $t$
with $0 \leq t \leq K$ and a decomposition of each \grade{$p_{s,2}$}
as
\begin{equation} \label{eq:2}
\grade{p_{s,2}} =
\sum_{i=1}^t \sum_{j=1}^{M(i)} a_{s,i,j}P_{s,i,j} \cdot
Q_{s,i,j} + a_{s, \reg} \cdot R_{s,\reg} + R_{s,\nneg},
\end{equation}
(recall that $M(i) = {\frac 1 {\eta_i^2}} \log(1/\eps)$),
where for each $s \in [k]$,

\begin{enumerate}

\item $\Var[R_{s,\nneg}] \leq \eps = \tau$,

\item each $a_{s, \reg} \cdot R_{s,\reg}$ is $\eta_{t+1}$-eigenregular, $R_{s, \reg}$ has variance $1$, and

\item For each $s,i,j$, the polynomials $P_{s,i,j}$ and $Q_{s,i,j}$
are both in ${\cal W}^1$ and are defined over disjoint sets of
variables and have $\Var[P_{s,i,j}]=\Var[Q_{s,i,j}]=1$.  Moreover
$R_{s,\nneg}$ and $R_{s,\reg}$ both lie in ${\cal W}^2.$

\end{enumerate}

We now describe the polynomials $\Inner(p_{s,2})$ and $\Outer(p_{s,2})$,
whose existence is claimed in Theorem \ref{thm:general},
for each $s \in [k]$.
\ignore{For each $s \notin A$ (i.e. $s$ such that
$\alpha_s=0$) we simply take $\num(p_{s,2})=0$ and $\Outer(p_{s,2})=0$
so $\tilde{p}_{s,2}=0$ (note that this choice indeed
satisfies $\Var[p_{s,2}-\tilde{p}_{s,2}] \leq \tau$
as required by Condition (2)).}
For each $s \in [k]$,
the polynomials $\Inner(p_{s,2})_{\ell}$
include all of the polynomials
$P_{s,i,j},Q_{s,i,j}$ from (\ref{eq:2}). Furthermore, $R_{s, \reg}$ belongs to  $\Inner(p_{s,2})_{\ell}$ if and only if
$a_{s, \reg} \not =0$.
The polynomial
$\tilde{p}_{s,2}(x)$ is $p_{s,2} - R_{s,\nneg}$, and we have
\begin{equation} \label{eq:h2}
\tilde{p}_{s,2}(x) =
\Outer(p_{s,2})(\{\Inner(p_{s,2})_{\ell}\}) =
\sum_{i=1}^t \sum_{j=1}^{M(i)} a_{s,i,j} \cdot
P_{s,i,j} \cdot
Q_{s,i,j} +  a_{s, \reg} \cdot R_{s,\reg}.
\end{equation}Further, by the guarantee from Theorem~\ref{thm:regularize-many}, we get that for any $s,i,j$, $P_{s,i,j}$ and $Q_{s,i,j}$ are on disjoint sets
of variables.
The degree-1 and degree-0 portions are even simpler:  for each
$s \in [k]$ we have
\[
\tilde{p}_{s,1}(x) =
\Outer(p_{s,1})(p_{s,1})=p_{s,1} \quad \text{and} \quad
\tilde{p}_{s,0}(x) =
\Outer(p_{s,0})(p_{s,0}) = p_{s,0}.
\]
It is clear from this description that property (1) holds (for each
$s,q$, the polynomial $\tilde{p}_{s,q}$ indeed belongs
to ${\cal W}^{q}$).

Next, we show that Condition 2 holds. This is immediate for
$q=0,1$ since the polynomials $p_{s,0}$ and
$p_{s,1}$ are identical to $\tilde{p}_{s,0}$ and $\tilde{p}_{s,1}$
respectively.  For $q=2$ given any $s \in [k],$
\ignore{If $s \notin A$ then as noted earlier we have that
$\Var[p_{s,2}-\tilde{p}_{s,2}] \leq \tau$
as required by Condition (2), so suppose that $s \in A$. }
 we
have that $p_{s,2} - \tilde{p}_{s,2} =R_{s,\nneg}$, and hence
by the upper bound on $\Var[R_{s,\nneg}]$ (see Item (1) above)
we have that $\Var[p_{s,2} - \tilde{p}_{s,2}] =
\Var[R_{s,\nneg}] \leq \tau$
as required by Condition (2).

For Condition (3), the multilinearity of each $\Outer(p_{s,q})$ is easily
seen to hold as a consequence of
Theorem \ref{thm:regularize-many}.  To check the remaining part of Condition (3), note that for $q=0,1$ the polynomial
$\tilde{p}_{s,q}$ is simply the identity polynomial $x \mapsto x$.
For $q=2$, Equation (\ref{eq:h2}) immediately gives that \[
\num(p_{s,2}) \le 2t \cdot M(t) + 1 \leq  2K \cdot M(t) +1 \]
% \ignore{\cdot \max_{s,i,j} |c_{s,i,j}|
% \red{ \leq
% O\left({\frac k \tau} \log{\frac 1 \tau} \right) \cdot {\frac 1 {\eta_t^2}} \log(1/\tau)
% }}
Observe that since $\Outer(p_{s,0}) (\cdot)$ and $\Outer(p_{s,1})(\cdot)$ is simply the identity map
$x  \mapsto x$. Thus,
$$
\sum_{q=0}^2\num(p_{s,q}) \le 2 K \cdot M(t) +3.
$$
As $s \in [k]$,
\[
\mathrm{Num}= \sum_{s \in [k]}\sum_{q=0}^2\num(p_{s,q}) \le 2 K \cdot k \cdot M(t) +3k.
\]

Note that we can choose $C$ to be a sufficiently large constant (independent of $t$) so that
\begin{equation}\label{eq:Nbound}
 \mathrm{Num} \le 2 K \cdot k \cdot M(t) +3k  \le C\cdot \frac{ k^2}{ \epsilon} \log^2{\frac 1 \epsilon} \cdot {\frac 1 {\eta_t^2}}.
\end{equation}

We next upper bound the sum of the absolute values of the coefficients
appearing in $\Outer(p_{s,q})$. We begin by observing that $\Outer(p_{s,1})$ and $\Outer(p_{s,0})$ are just the
identity function and hence the absolute values of the coefficients is just $1$. For $\Outer(p_{s,2})$, note that
Item (4) of Theorem~\ref{thm:regularize-many}, gives that
$$
\sum_{i=1}^t \sum_{j=1}^{M(i)} a_{s,i,j}^2 \le \sum_{i=1}^t \left(\frac{4}{\eta_i}\right)^{4(M(i)-1)} \le t \cdot \left(\frac{4}{\eta_t}\right)^{4(M(t)-1)}. $$
Thus, summing over all $s \in [k]$, we get
$$
\sum_{s \in [k]}
\sum_{i=1}^t \sum_{j=1}^{M(i)} a_{s,i,j}^2  \le  k \cdot  t \cdot \left(\frac{4}{\eta_t}\right)^{4(M(t)-1)}.$$
Recalling that $\sum_{s \in [k]} \sum_{i=1}^t M(i) \le 2 K \cdot k \cdot M(t) +3k
$ and applying Cauchy-Schwarz, we get
$$
\sum_{s \in [k]} \sum_{i=1}^t |a_{s,i,j} | \le  \left(t \cdot \left(\frac{4}{\eta_t}\right)^{4(M(t)-1)}\right) \cdot \left(2 K \cdot k \cdot M(t) +3k  \right)
$$
Thus, we can choose a sufficiently large constant $C'$ (independent of $t$) such that
\begin{equation}\label{eq:coeffbound}
\mathrm{Coeff} = \sum_{s \in [k]} \sum_{i=1}^t |a_{s,i,j} | \le C' \cdot k^3 \cdot \left( \frac{4}{\eta_t} \right)^{C' \cdot \frac{1}{\eta_t^2} \cdot \log(1/\epsilon)}
\end{equation}
Using (\ref{eq:coeffbound}) and (\ref{eq:Nbound}), the recursive definition of $\eta_{t}$ from (\ref{eq:eta-spec}) and $t \le K$, there exists $N = N_{\beta}(k,d,\tau )$ and $M=M_{\beta}(k,d,\tau)$ such that $\mathrm{Num} \le N$ and $\mathrm{Coeff} \le M$.

It remains only to bound the eigenregularity of the $\Inner(p_{s,q})_{\ell}$
polynomials.  Each such polynomial is either a degree-0 polynomial
$p_{s,0}$, a degree-1 polynomial $p_{s,1}$ or $P_{s,i,j}$
or $Q_{s,i,j}$ (from (\ref{eq:h2})), or a degree-2 polynomial $R_{s,\reg}$ for which $a_{s,\reg} \not =0$.
Since degree-0 and degree-1 polynomials are $0$-eigenregular, it remains only
to bound the eigenregularity of each $R_{s,\reg}$.  Since each $R_{s,\reg}$
is $\eta_{t+1}$-eigenregular,
our choice of the sequence $1=\eta_0 \geq \cdots \geq \eta_K$ (see
(\ref{eq:eta-spec})) implies that each $R_{s,\reg}$ is indeed $\beta(\textrm{Num}+\textrm{Coeff})$-eigenregular
as required by Condition (4).  This concludes the proof of the
degree-2 base case of the induction for Theorem \ref{thm:general}.
\qed

\paragraph{Inductive step:  Proof of \ref{thm:general}.}
Fix a value $d \geq 3$; with the base case in hand
from above, we may assume that Theorem \ref{thm:general}
holds for degrees $2,\dots,d-1.$
%Fix any
%non-increasing function $\beta: [1,\infty) \to (0,1)$
%that satisfies $\beta(x) \leq 1/x.$
\ignore{Similar to the degree-2 base case,
for each $s \in [k]$ we define a multiplicative re-scaling constant
$\alpha_s$ as follows:  if $\Var[p_{s,d}] < \tau$ then we set
$\alpha_s=0$ and if $\Var[p_{s,d}] \geq \tau$ then we set
$\alpha_s = 1/\sqrt{\Var[p_{s,d}]} \in \red{[1,1/\sqrt{\tau}]}$.
Let $A \subseteq [k]$ be the set of those values $s$ such that
$\alpha_s>0$, so each polynomial in the vector
$(\alpha_s p_{s,d})_{s \in A}$ has $\Var[\alpha_s p_{s,d}]=1$
(as required by Theorem \ref{thm:regularize-many}).}
Similar to the degree-2 base case, the first main step of the algorithm is to
use the {\bf MultiRegularize-One-Wiener} procedure on the vector of
polynomials $(p_{s,d})_{s \in [k]}$; we now specify
how to set
the parameters of {\bf MultiRegularize-One-Wiener}.
Fix $\eps = \tau/8$ and let $K=O(k/\eps \cdot \log(1/\eps))$
as in the statement of Theorem \ref{thm:regularize-many}.
We define the parameters $1=\eta_0 \geq \eta_1 \cdots \geq \eta_K > n_{K+1}=0$  as
follows:  for $t=0,\dots,K-1$, we have

\begin{equation} \label{eq:eta-spec-d}
\eta_{t+1} :=
\beta
\left(
N_{\beta^{\ast}_{t}}\left(L_t,d-1,\frac{\tau}{16 \cdot L_t  \cdot L_t^{L_t}}\right)^2 +M_{\beta^{\ast}_t} \left(L_t, d-1,\frac{\tau}{16 \cdot  L_t\cdot L_t^{L_t}}\right)^2  \cdot L_{t}^{L_t} +2k + L_t^{L_t}\right),
\end{equation} $$ \text{where} \quad
L_t \eqdef
C' \cdot
{\frac {k^2 d} \eps} \cdot
{\frac 1 {\eta_t^2}} \cdot
(\log 1/\eps)^2 ,
$$
and $C'>0$ is an absolute constant defined below and $\beta^{\ast}_t$ is defined in (\ref{eq:betaonedef}).  The reader can verify that the sequence $\{\eta_t \}$ is defined consistently.  As before, intuitively, the first term in
the argument to $\beta$ corresponds to an upper bound on $\mathrm{Num}$ whereas the second term corresponds
to an upper bound on $\mathrm{Coeff}$.
%and the $\gamma_s$ function is defined below (see (\ref{eq:gamma})).
%(Note the assumption that $\beta(x) \leq 1/x$ implies that indeed
%$\eta_{s+1} \leq \eta_s$.
(Note that from the recursive definition (\ref{eq:eta-spec-d}),
for all $t=1,\dots,K$ we have that $\eta_t = \kappa_t(k,\tau,d)$
for some function $\kappa_t$; this will be useful later.)
When called with these parameters on the vector of
polynomials $(p_{s,d})_{s \in [k]}$, by Theorem \ref{thm:regularize-many}
{\bf MultiRegularize-One-Wiener} outputs
an index $t$
with $0 \leq t \leq K$ and a decomposition of each $p_{s,d}$ as
\begin{equation} \label{eq:d}
p_{s,d} =
\sum_{i=1}^t \sum_{j=1}^{M(i)} a_{s,i,j}P_{s,i,j} \cdot
Q_{s,i,j} + a_{s,\reg} \cdot R_{s,\reg} + R_{s,\nneg},
\end{equation}
(recall that $M(i) = {\frac 1 {\eta_i^2}} \log(1/\eps)$),
where for each $s \in [k]$,

\begin{itemize}

\item $\Var[R_{s,\nneg}] \leq \eps = \tau/8$,

\item each $a_{s, \reg} \cdot R_{s,\reg}$ is $\eta_{t+1}$-eigenregular, $R_{s,\reg}$ has variance $1$, and

\item For each $s,i,j$, the polynomial $P_{s,i,j}$ belongs to
${\cal W}^{q_{s,i,j,1}}$ and $Q_{s,i,j}$ belongs to
${\cal W}^{q_{s,i,j,2}}$ where $0 < q_{s,i,j,1},q_{s,i,j,2}$,
$q_{s,i,j,1}+q_{s,i,j,2}=d$, and
$P_{s,i,j},Q_{s,i,j}$ are defined over disjoint sets of
variables and have $\Var[P_{s,i,j}]=\Var[Q_{s,i,j}]=1$.  Moreover
$R_{s,\nneg}$ and $R_{s,\reg}$ both lie in ${\cal W}^d.$

\end{itemize}

Define the function $\beta^{\ast}_{t}: [1,\infty) \to (0,1),$
$\beta^{\ast}_{t}(x) \leq 1/x$, as
\begin{equation} \label{eq:betaonedef}
\beta^{\ast}_{t}(x) := \beta(x^2 \cdot k \cdot L_t^{L_t} + k +L_t^{L_t}).
\end{equation}

The second main step of the algorithm is to run the procedure
{\bf \blue{MultiRegularize-Many-Wieners}}$_{d-1,\beta^{\ast}_{t}}$ with its inputs
set in the following way:

\begin{itemize}

\item
There are $k$ lists of $d$ multilinear polynomials
$p_{s,0},\dots,p_{s,d-1}$ as $s$ ranges from $1$ to $k$.
Additionally,
for each of the $P_{s,i,j}$ and $Q_{s,i,j}$ polynomials that are obtained
from (\ref{eq:d}), there is a list of $d$ multilinear polynomials
in ${\cal W}^0,\dots,{\cal W}^{d-1}$.  In each such list
all the polynomials are 0 except for the $P_{s,i,j}$ or $Q_{s,i,j}$
polynomial.
(Note that each of these polynomials indeed belongs to a single
specific layer ${\cal W}^q$
of the Wiener chaos for some $1 \leq q \leq d-1$,
as required by {\bf \blue{MultiRegularize-Many-Wieners}}$_{d-1,\beta^{\ast}_{t}}$;
note further that since the variance of each $P_{s,i,j}$
and $Q_{s,i,j}$ polynomial is 1, the variance condition in the first
bullet of Theorem \ref{thm:general} is satisfied.
Finally, we emphasize that the $R_{s, \reg}$ polynomial, which
lies in ${\cal W}^d$, is not used here.)
Thus, note that we can choose $C'$ to be a sufficiently large constant so that the total number of input polynomials that are given to the procedure
is at most
\begin{equation} \label{eq:numpolys}
L_t \eqdef
C' \cdot
{\frac {k^2 d} \eps} \cdot
{\frac 1 {\eta_t^2}} \cdot
(\log 1/\eps)^2.
\end{equation}
 Note that $L_t = \kappa_1(k,\tau,d)$ for
some function $\kappa_1$; this will be useful later.
Another thing that will be useful is the following. Define \begin{equation}\label{eq:defad} A_d = ( \sum_{s \in [k]} \sum_{i=1}^t  \sum_{j=1}^{M(i)} |a_{s,i,j}|)^2 . \end{equation} Note that mimicking the calculations preceding (\ref{eq:coeffbound}), it can be shown that $C'$ can be chosen so that  $A_d \le L_t^{ L_t}$.

\item The ``$\tau$'' parameter of the procedure is set to $\frac{\tau}{16 \cdot L_t \cdot L_t^{L_t}}$.
\end{itemize}

By the inductive hypothesis, when invoked this way, the procedure
{\bf \blue{MultiRegularize-Many-Wieners}}$_{d-1,\beta^{\ast}_{t}}$
returns the following:

\begin{itemize}

\item  For each $P_{s,i,j}$ and $Q_{s,i,j}$ from
(\ref{eq:d}),
outer polynomials
$\Outer(P_{s,i,j})$
and
$ \Outer(Q_{s,i,j})$,
and likewise for each $p_{s,q}$ ($1 \leq s \leq k,
0 \leq q \leq d-1$) an outer polynomial
$\Outer(p_{s,q})^{(d-1)}$;
the ``$(d-1)$'' superscript here is to emphasize that this polynomial
is obtained from the call to {\bf \blue{MultiRegularize-Many-Wieners}}$_{d-1,\beta^{\ast}_{t}}$.

\item For each $P_{s,i,j}$ and $Q_{s,i,j}$ from (\ref{eq:d}), a
collection of ``inner polynomials''
\begin{equation} \label{eq:12}
\{\Inner(P_{s,i,j})_{\ell}\}_{\ell=1,\dots,\num(P_{s,i,j})} \quad
\text{and} \quad
\{\Inner(Q_{s,i,j})_{\ell}\}_{\ell=1,\dots,\num(Q_{s,i,j})}
\end{equation}
and likewise for each $p_{s,q}$ ($1 \leq s \leq k,
0 \leq q \leq d-1$) a collection of inner polynomials
\begin{equation} \label{eq:11}
\{\Inner(p_{s,q})^{(d-1)}_{\ell}\}_{
\ell=1,\dots,\num(p_{s,q})};
\end{equation}
similar to above, the ``$(d-1)$'' superscript here is to emphasize
that these polynomials
are obtained from a call to {\bf \blue{MultiRegularize-Many-Wieners}}$_{d-1,\beta^{\ast}_{t}}$.

\end{itemize}
For each $p_{s,q}$ with $1 \leq s \leq k,0 \leq q \leq d-1$,
let us write
$\tilde{p}_{s,q}^{(d-1)}$ to denote
$\Outer(p_{s,q})^{(d-1)}(\{\Inner(p_{s,q})^{(d-1)}_{\ell}\}).$

The pieces are now in place for us to describe the polynomials
$\Outer(p_{s,q})$ and
$\{\Inner(p_{s,q})_{\ell}\}_{\ell=1,\dots,\num(p_{s,q})}$
whose existence is asserted by Theorem \ref{thm:general}.

Fix any $s \in [k]$.
We begin with the easy case of $q < d$; so fix any $0 \leq q < d.$
The polynomial $\Outer(p_{s,q})$ is simply
$\Outer(p_{s,q})^{(d-1)}$, and the polynomials
$\{\Inner(p_{s,q})_{\ell}\}$ are simply the polynomials
$\{\Inner(p_{s,q})^{(d-1)}_{\ell}\}$.

Now we turn to the more involved case of $q=d$.
Fix any $s \in [k].$
Recall the decomposition of $p_{s,d}$ given by (\ref{eq:d}).
Fix any $i \in [t], j \in [M(i)]$ and consider first
the polynomial $P_{s,i,j}$ from (\ref{eq:d}).  By the inductive
hypothesis, the call to {\bf \blue{MultiRegularize-Many-Wieners}}$_{d-1,\beta^{\ast}_{t}}$
yields a polynomial $\tilde{P}_{s,i,j}$ which
is defined via (\ref{eq:tildep}) in terms of the
$\Outer(P_{s,i,j})$ and
$\Inner(P_{s,i,j})_{\ell}$ polynomials, namely
\begin{equation} \label{eq:tildeP}
\tilde{P}_{s,i,j} = \Outer(P_{s,i,j})(
\Inner(P_{s,i,j})_1,\dots,\Inner(P_{s,i,j})_{\num(P_{s,i,j})}).
\end{equation}

Similarly, considering the
polynomial $Q_{s,i,j}$ from (\ref{eq:d}),
the call to {\bf \blue{MultiRegularize-Many-Wieners}}$_{d-1,\beta^{\ast}_{t}}$ also
yields a polynomial $\tilde{Q}_{s,i,j}$ which
is defined via (\ref{eq:tildep}) in terms of the
$\Outer(Q_{s,i,j})$ and
$\Inner(Q_{s,i,j})_{\ell}$ polynomials, namely

\begin{equation} \label{eq:tildeQ}
\tilde{Q}_{s,i,j} = \Outer(Q_{s,i,j})(
\Inner(Q_{s,i,j})_{1},\dots,\Inner(Q_{s,i,j})_{\num(Q_{s,i,j})}).
\end{equation}

The polynomials $\{\Inner(p_{s,d})_{\ell}\}$ are the elements of
\begin{eqnarray}
\left(
\bigcup_{i=1}^t \bigcup_{j=1}^{M(i)}
\left(\{\Inner(P_{s,i,j})_{\ell'}\}_{\ell'=
1,\dots,\num(P_{s,i,j})}
\cup
\{\Inner(Q_{s,i,j})_{\ell''}\}_{\ell''=
1,\dots,\num(Q_{s,i,j})}
\right)
\right) \cup \{R_{s,\reg} : a_{s, \reg} \not =0\}
\label{eq:inner}
\end{eqnarray}
and the polynomial $\Outer(p_{s,d})$ is given by
\begin{eqnarray} \label{eq:outer}
\Outer(p_{s,d})(\{\Inner(p_{s,d})_{\ell}\}) =
\sum_{i=1}^t \sum_{j=1}^{M(i)}
a_{s,i,j} \tilde{P}_{s,i,j} \cdot
\tilde{Q}_{s,i,j}
+ a_{s, \reg} \cdot R_{s,\reg}.
\end{eqnarray}
Recalling (\ref{eq:tildep}), (\ref{eq:tildeP}) and
(\ref{eq:tildeQ}), we may write this more explicitly as
\begin{eqnarray} \label{eq:30}
\tilde{p}_{s,d} = \sum_{i=1}^t \sum_{j=1}^{M(i)}
 a_{s,i,j} \cdot
\left(
\Outer(P_{s,i,j})(\{\Inner(P_{s,i,j})_{\ell}\})
\right)
\cdot
\left(
\Outer(Q_{s,i,j})(\{\Inner(Q_{s,i,j})_{\ell'}\})
\right)
+ a_{s, \reg} \cdot R_{s,\reg}.
\end{eqnarray}

\ignore{
This completes the description of each $\tilde{p}_{s,q}$ for $0 \leq q
\leq d$, and recalling that $\tilde{p}_s = \sum_{q=0}^d \tilde{p}_{s,q}$,
it completes the description of $\tilde{p}_s.$

So the overall polynomial $\tilde{p}_s$ whose existence is asserted
by (\ref{eq:tildep}) is
\begin{equation} \label{eq:20}
\tilde{p}_s =
\Outer(p_s)_d(\{\Inner(p_s)_{d,\ell}\}) +
\sum_{q=0}^{d-1}
\Outer(p_s)_q(\{\Inner(Z_s)_{q,\ell}\}).
\end{equation}
}

This concludes the specification of the
$\Outer(p_{s,q})$, $\{\Inner(p_{s,q})_{\ell}\}$ and $\tilde{p}_{s,q}$
polynomials;
it remains to show that these polynomials satisfy all of the claimed
Conditions (1)-(4).

First we consider Condition (1), that for each $s,q$ the polynomial $\tilde{p}_{s,q}$
belongs to the $q$-th Wiener chaos ${\cal W}^q$ and each polynomial $\Inner(p_{s,q})_\ell$ has $\Var[\Inner(p_{s,q})_\ell] = 1$.  For $q < d$ this follows from
the inductive hypothesis.  For $q=d$,
consider (\ref{eq:d}) and (\ref{eq:outer}).
Part (2) of Theorem \ref{thm:regularize-many}
ensures that $R_{s,\reg}$ and each $P_{s,i,j} \cdot Q_{s,i,j}$ all
lie in ${\cal W}^d$ and that $P_{s,i,j}$ and $Q_{s,i,j}$ are on disjoint sets of variables. Using the inductive
hypothesis applied to each $P_{s,i,j}$ and
$Q_{s,i,j}$ \grade{and the fact $P_{s,i,j}$ and $Q_{s,i,j}$ are on disjoint sets of variables}, each product
$\tilde{P}_{s,i,j} \cdot \tilde{Q}_{s,i,j}$ also must lie
in ${\cal W}^d.$ Since $R_{s,\reg}$ lies in  ${\cal W}^d$ and ${\cal W}^d$ is a subspace, we get that $\tilde{p}_{s,d} \in {\cal W}^d$ for all $s \in [k]$.
To see that each polynomial in $\{\Inner(P_{s,i,j})_{\ell}\}$ and $\{\Inner(Q_{s,i,j})_{\ell'}\}$ has variance $1$, we use the induction hypothesis. Also, note that
by Theorem~\ref{thm:regularize-many}, $\Var (R_{s,\reg}) =1$. This shows that all elements of  $\{\Inner(p_{s,q})_{\ell}\}$ have variance $1$ as claimed.

Next let us consider Condition (2), i.e. we must upper bound each
$\Var[\tilde{p}_{\grade{s,q}}-p_{s,q}].$
Fix any $s \in [k]$ and first consider some $1 \leq q \leq d-1$.  In this
case we have that
\begin{eqnarray*}
\tilde{p}_{s,q} &=&
\Outer(p_{s,q})\left(
\Inner(p_{s,q})_1(x),\dots, \Inner(p_{s,q})_{\num(p_{s,q})}(x)\right)\\
&=&
\Outer(p_{s,q})^{(d-1)}\left(
\Inner(p_{s,q})^{(d-1)}_1(x),\dots, \Inner(p_{s,q})^{(d-1)}_{
\num(p_{s,q})}(x)\right)\\
&=&
\tilde{p}_{s,q}^{(d-1)}
\end{eqnarray*}
so we get the desired bound on $\Var[\tilde{p}_{s,q}-p_{s,q}]$
from the inductive hypothesis, recalling
that the ``$\tau$'' parameter that was provided to
{\bf \blue{MultiRegularize-Many-Wieners}}$_{d-1,\beta^{\ast}_{t}}$ was set to $\grade{\tau/(16 \cdot  L_t \cdot L_t^{L_t} )
\leq \tau.}$

Next, consider the case $q=d$. In this case,
$$
p_{s,d} - \tilde{p}_{s,d} = R_{s,\nneg} + \sum_{i=1}^t \sum_{j=1}^{M(i)} a_{s,i,j} (\tilde{P}_{s,i,j} \cdot \tilde{Q}_{s,i,j}- P_{s,i,j} \cdot Q_{s,i,j})
$$
We will record the following fact which will be used a couple of times.
\begin{fact}\label{fact:variance-ineq}
$\Var[X_1 + \ldots + X_m] \le \sum_{i=1}^m m \cdot \Var[X_i]$.
\end{fact}
First, applying Fact~\ref{fact:variance-ineq}, we can say that
$$
\Var[p_{s,d} - \tilde{p}_{s,d} ] \le 2 \Var [R_{s,\nneg} ] + 2\Var\left[\sum_{i=1}^t \sum_{j=1}^{M(i)} a_{s,i,j} (\tilde{P}_{s,i,j} \cdot \tilde{Q}_{s,i,j}- P_{s,i,j} \cdot Q_{s,i,j})\right].
$$
We recall that $\Var[R_{s,\nneg}] \le \tau/8$. To bound the second term,  define $$\kappa_{\max} = \max_{i,j} \Var(\tilde{P}_{s,i,j} \cdot \tilde{Q}_{s,i,j} -{P}_{s,i,j} \cdot {Q}_{s,i,j}).$$
\noindent
Then, we have that
\begin{eqnarray*}
&&\Var\left[\sum_{i=1}^t \sum_{j=1}^{M(i)} a_{s,i,j} (\tilde{P}_{s,i,j} \cdot \tilde{Q}_{s,i,j}- P_{s,i,j} \cdot Q_{s,i,j})\right]\\
&\leq & \left(\sum_{i=1}^t M(i)\right) \cdot \left(\sum_{i=1}^t \sum_{j=1}^{M(i)} a_{s,i,j}^2 \Var(\tilde{P}_{s,i,j} \cdot \tilde{Q}_{s,i,j} -{P}_{s,i,j} \cdot {Q}_{s,i,j}) \right) \\
&\le& \left(\sum_{i=1}^t M(i) \right)\cdot \left(\sum_{i=1}^t \sum_{j=1}^{M(i)} a_{s,i,j}^2\right) \cdot  \kappa_{\max} \\
&\le&  \left(\sum_{i=1}^t M(i)\right) \cdot \left(\sum_{i=1}^t \sum_{j=1}^{M(i)} |a_{s,i,j}|\right)^2 \cdot  \kappa_{\max}\\
&\le& L_t \cdot A_d \cdot \kappa_{\max} \le L_t \cdot L_t^{L_t} \cdot \kappa_{\max}
\end{eqnarray*}
In the above, the first inequality uses Fact~\ref{fact:variance-ineq}, the fourth inequality uses  (\ref{eq:numpolys}) and (\ref{eq:defad}) and the fifth inequality uses the bound on $A_d$.  It remains to bound $\kappa_{\max}$. However, note that for any $1 \le i \le t$ and $1 \le j \le M(i)$, we have that
\begin{eqnarray*}
 \Var(\tilde{P}_{s,i,j} \cdot \tilde{Q}_{s,i,j} -{P}_{s,i,j} \cdot {Q}_{s,i,j}) &\le& 2 \cdot \Var(Q_{s,i,j} \cdot  (\tilde{P}_{s,i,j}  -{P}_{s,i,j}) ) + 2\cdot \Var(\tilde{P}_{s,i,j} \cdot (\tilde{Q}_{s,i,j} - Q_{s,i,j})) \\
 &\le& 2 \Var(Q_{s,i,j}) \cdot \Var(\tilde{P}_{s,i,j}  -{P}_{s,i,j}) + 2 \cdot \Var(\tilde{P}_{s,i,j}) \cdot \Var(\tilde{Q}_{s,i,j} - Q_{s,i,j})) \\
 &\le& 2 \Var(\tilde{P}_{s,i,j}  -{P}_{s,i,j}) + 4 \Var(\tilde{Q}_{s,i,j} - Q_{s,i,j})) \le \frac{ 3 \cdot \tau}{8 \cdot L_t \cdot L_t^{L_t}}.
 \end{eqnarray*}
 In the above, the first inequality uses Fact~\ref{fact:variance-ineq} and  the second  inequality uses that $P_{s,i,j}$ and $Q_{s,i,j}$ are on disjoint sets of variables. The third inequality uses that $\Var(Q_{s,i,j}) =1$ and that $\Var(\tilde{P}_{s,i,j}) \le 2$ and the fourth one follows from the choice of the ``$\tau$" parameter in the recursive procedure.
Hence,   we get that
$$
\Var\left[\sum_{i=1}^t \sum_{j=1}^{M(i)} a_{s,i,j} (\tilde{P}_{s,i,j} \cdot \tilde{Q}_{s,i,j}- P_{s,i,j} \cdot Q_{s,i,j})\right] \le \frac{3\tau}{8}.
$$
As a consequence, we get $$\Var (p_{s,d} - \tilde{p}_{s,d}) \le 2 \Var (R_{s,\nneg} ) + 2\Var\left(\sum_{i=1}^t \sum_{j=1}^{M(i)} a_{s,i,j} (\tilde{P}_{s,i,j} \cdot \tilde{Q}_{s,i,j}- P_{s,i,j} \cdot Q_{s,i,j})\right) \le \tau$$.
% Next consider the case $q=d$.  If $s \notin A$ then it is clear from above
% that as desired we have $\Var[p_{s,d} - \tilde{p}_{s,d}] \leq
% \tau$, so consider any $s \in A.$
% \blue{Informal argument:
% The difference between $\tilde{p}_{s,d}$
% and $p_{s,d}$ comes from (a) the fact that we
% dropped $R_{s,\nneg}/\alpha_s$ in (\ref{eq:outer}), and (b) the fact in
% (\ref{eq:outer}), each
% $\tilde{P}_{s,i,j}$ and $\tilde{Q}_{s,i,j}$ is just an approximator
% to the real $P_{s,i,j}$ or $Q_{s,i,j}$.
% By the upper bound on $\Var[R_{s,\nneg}]$, (a) above incurs a cost of
% at most $\Var[R_{s,\nneg}]/\alpha_s^2 \leq \Var[R_{s,\nneg}] \leq \tau$.
% For (b),
% inductively we know that each approximator $\tilde{P}_{s,i,j}$
% to $P_{s,i,j}$ has $\Var[\tilde{P}_{s,i,j} - P_{s,i,j}] \leq
% \red{FILL-IN}$ (recall that this is what the ``$\tau$''
% parameter was set to in the call to
% {\bf \blue{MultiRegularize-Many-Wieners}}$_{d-1,\beta_1}$),
% and likewise for each approximator
% $\tilde{Q}_{s,i,j}$ to $Q_{s,i,j}.$  \red{Use the fact that
% there are at most $2tM(t)$ many of these
% $\tilde{P}_{s,i,j}$ and $\tilde{Q}_{s,i,j}$ approximators, together
% with the fact that each $c_{s,i,j}/\alpha_s$ is at most 1, to argue
% that the ``loss'' in the variance here is at most something like
% $2tM(t) \cdot \red{FILL-IN},$ so for a suitable choice of
% \red{$FILL-IN$} the overall polynomial $\tilde{p}_{s,d}-p_{s,d}$ will have
% variance at most $\tau.$}
% }

For Condition (3), the multilinearity of each $\Outer(p_{s,q})$
follows easily from the inductive hypothesis and from
(\ref{eq:outer}), using the fact that each $R_{s,\reg}$
is multilinear and that for each $(s,i,j)$ triple the two
multilinear polynomials $P_{s,i,j}$ and $Q_{s,i,j}$ are defined over disjoint
sets of variables.
\ignore{For the sum of absolute values of
coefficients of $\Outer(p_{s,q})$, we first observe that
for $q \leq d-1$ the bound follows from the inductive hypothesis,
which gives that the sum of absolute values of coefficients
of the outer polynomials returned in the call to
{\bf \blue{MultiRegularize-Many-Wieners}}$_{d-1,\beta^{\ast}_{t}}$ is at most
$S_{\beta_1}(d-1,L_t,\red{FILL-IN}).$
For $q=d$,
we first observe that as above, each multilinear polynomial
$\Outer(P_{s,i,j})$ has sum of absolute values of all coefficients
of its arguments $\{\Inner(P_{s,i,j})_{\ell}\}$ bounded by at most
$S_{\beta_1}(d-1,L_t,\red{FILL-IN})$, and likewise for each multilinear
polynomial
$\Outer(Q_{s,i,j})$.
Hence using (\ref{eq:30}) and observing that
each $a_{s,i,j}/\alpha_s$ has magnitude at most 1,  we see that the
sum of absolute values of all coefficients
of $\Outer(p_{s,d})$ (recall that this polynomial's arguments are
the polynomials specified in (\ref{eq:inner})
is at most
\[
tM(t) S_{\beta_1}(d-1,L_t,\red{FILL-IN})^2 + 1 =
O(
tM(t) S_{\beta_1}(d-1,L_t,\red{FILL-IN})^2) := S_{\beta}(d,k,\tau).
\]
}

To finish establishing condition (3), we now bound the overall number of inner polynomials
produced in these decompositions.
We start by bounding the total number of inner polynomials
that {\bf \blue{MultiRegularize-Many-Wieners}}$_{d,\beta}$ produces.
Recall first that the number of polynomials that are input
to the call to {\bf \blue{MultiRegularize-Many-Wieners}}$_{d-1,\beta^{\ast}_t}$ is $L_t=
\kappa_t(k,\tau,d)$ (see \ref{eq:numpolys}).
From the specification of the $\{\Inner(p_{s,q})_\ell\}$ polynomials
given earlier, all but precisely $k$ of these inner polynomials
returned by {\bf \blue{MultiRegularize-Many-Wieners}}$_{d,\beta}$
are simply the inner polynomials that are returned from
the call to {\bf \blue{MultiRegularize-Many-Wieners}}$_{d-1,\beta^{\ast}_t}$,
and the other $k$ polynomials are $R_{1,\reg},\dots,R_{k,\reg}$.
Since the {\bf \blue{MultiRegularize-Many-Wieners}}$_{d-1,\beta^{\ast}_t}$,
procedure is called on $L_t$ many polynomials,
by the inductive hypothesis the total number of inner polynomials
returned by this procedure call is some value
\begin{equation} \label{eq:gamma}
N_{\beta^{\ast}_t}\left(L_t,d-1,\frac{\tau}{16\cdot L_t \cdot L_t^{L_t}}\right) =
O_{k,\tau,d}(1).
\end{equation}
Consequently
the total number of all inner polynomials
$\{\Inner(p_{s,q})_\ell\}$ returned by {\bf \blue{MultiRegularize-Many-Wieners}}$_{d,\beta}$
is bounded by $N_{\beta^{\ast}_t}\left(L_t,d-1,\frac{\tau}{16\cdot L_t \cdot L_t^{L_t}}\right) + k$. Noting that $L_t \le L_K$, we have that the total number of inner polynomials
returned by {\bf \blue{MultiRegularize-Many-Wieners}}$_{d,\beta}$ is bounded by 
\[
N_{\beta^{\ast}_K}\left(L_K,d-1,\frac{\tau}{16\cdot L_K \cdot L_K^{L_K}}\right) + k,\] 
which we define to be
$N_{\beta}(k,d,\tau)$.

The next task is to upper bound  $\sum_{s,q} \mathrm{Coeff}(p_{s,q})$. The main step is to bound the contribution from $q=d$.  Recalling (\ref{eq:30}), we have
$$
\tilde{p}_{s,d} = \sum_{i=1}^t \sum_{j=1}^{M(i)}
 a_{s,i,j} \cdot
\left(
\Outer(P_{s,i,j})(\{\Inner(P_{s,i,j})_{\ell}\})
\right)
\cdot
\left(
\Outer(Q_{s,i,j})(\{\Inner(Q_{s,i,j})_{\ell'}\})
\right)
+ a_{s, \reg} \cdot R_{s,\reg}
$$
and thus,
\begin{eqnarray*}
\mathrm{Coeff}(p_{s,d}) \le 1 +  \sum_{i=1}^t \sum_{j=1}^{M(i)}  |a_{s,i,j} | \cdot \mathrm{Coeff}(P_{s,i,j}) \cdot \mathrm{Coeff}(Q_{s,i,j}) .
\end{eqnarray*}
%Note that the first step of {\bf \blue{MultiRegularize-Many-Wieners}}$_{d,\beta}$ expresses
%$$
%\grade{p_{s,d}} =
%\sum_{i=1}^t \sum_{j=1}^{M(i)} a_{s,i,j}P_{s,i,j} \cdot
%Q_{s,i,j} + a_{s, \reg} \cdot R_{s,\reg} + R_{s,\nneg}.
%$$
 Note that $\sum_{s \in [k]} \sum_{i=1}^t \sum_{j=1}^{M(i)} |a_{s,i,j}| \le L_t^{L_t}$ and $|a_{s,\reg}| \le 1$. Further, by induction hypothesis, $$\sum_{s \in [k]} \sum_{i=1}^t \sum_{j=1}^{M(i)} \left( \mathrm{Coeff}(P_{s,i,j}) + \mathrm{Coeff}(Q_{s,i,j}) \right) + \sum_{s \in [k]} \sum_{q  =1}^{d-1} \mathrm{Coeff}(p_{s,q}) \le M_{\beta^{\ast}_t} \left(L_t, d-1, \frac{\tau}{16 \cdot L_t \cdot L_t^{L_t}} \right).   $$
%Observing that
%$$
%\sum_{s \in [k]} \mathrm{Coeff}(\Outer(p_{s,d})) \le k + \sum_{s\in [k]} \sum_{i=1}^t \sum_{j=1}^{M(i)} |a_{s,i,j}| \cdot  \mathrm{Coeff}(\Outer(P_{s,i,j})) \cdot \mathrm{Coeff}(\Outer(Q_{s,i,j})),
%$$
%(the extra $k$ is for  the contribution from $a_{s,\reg}$) and that  $\sum_{s\in [k]} \sum_{i=1}^t \sum_{j=1}^{M(i)} |a_{s,i,j}|  \le L_t^{L_t}$, we get
Hence,
$$
\sum_{s \in [k]}  \sum_{q=1}^d \mathrm{Coeff}(\Outer(p_{s,q})) \le k+L_t^{L_t} \cdot \left( M_{\beta^{\ast}_t} \left(L_t, d-1,\frac{\tau}{16 \cdot L_t \cdot L_t^{L_t}} \right) \right)^2.
$$
By definition,
$$
L_t^{L_t} \cdot \left( M_{\beta^{\ast}_t}  \left(L_t, d-1,\frac{\tau}{16 \cdot L_t \cdot L_t^{L_t}} \right) \right)^2 \le L_K^{L_K} \cdot \left( M_{\beta^{\ast}_K}  \left(L_K, d-1,\frac{\tau}{16 \cdot L_K \cdot L_K^{L_K}} \right) \right)^2.
$$
As $K = O(k/\epsilon) \log (1/\epsilon)$, we can define the quantity on the right hand side to be $M_{\beta}(k,d,\tau)$ which verifies condition (3).

Our last task is to bound the eigenregularity of the $\Inner(p_{s,q})_\ell$
polynomials.
As noted above these polynomials are of two types:  the $k$
polynomials $R_{1,\reg},\dots,R_{k,\reg}$ and the inner
polynomials $\{\Inner(p_{s,q})^{(d-1)}_\ell\}$ that were returned from
the call to {\bf \blue{MultiRegularize-Many-Wieners}}$_{d-1,\beta^{\ast}_t}$ on its
$L_t$ input polynomials.

%The key observations are: \newline 1. The $k$ polynomials $R_{s,\reg}$ are each
%$\eta_{t+1}$-regular. \newline 2. By inductive hypothesis
%each of the $\{\Inner(p_{s,q})^{(d-1)}_\ell\}$ polynomials
%is $\beta^{\ast}_t(\gamma_t(L_t,\tau,d))$-regular.

We first tackle the $\{\Inner(p_{s,q})^{(d-1)}_\ell\}$ polynomials. Define $\mathrm{Coeff}_{\mathop{rec}}$ as $$\mathrm{Coeff}_{\mathop{rec}} = \sum_{s \in [k]} \sum_{i=1}^t \sum_{j=1}^{M(i)} \left( \mathrm{Coeff}(P_{s,i,j}) + \mathrm{Coeff}(Q_{s,i,j}) \right) + \sum_{s \in [k]} \sum_{q  =1}^{d-1} \mathrm{Coeff}(p_{s,q}),$$
the sum of the absolute values of the coefficients of the outer polynomials returned by the recursive call.
Likewise, define $\mathrm{Num}_{\mathop{rec}}$ as $$
\mathrm{Num}_{\mathop{rec}} =  \left| \left\{ \bigcup_{s \in [k], q \in [1, \ldots, d-1]} \Inner(p_{s,q}) \cup
\bigcup_{i=1}^t \bigcup_{j=1}^{M(i)}
\left(\{\Inner(P_{s,i,j})_{\ell'}\}_{\ell'=
1,\dots,\num(P_{s,i,j})}
\cup
\{\Inner(Q_{s,i,j})_{\ell''}\}_{\ell''=
1,\dots,\num(Q_{s,i,j})}
\right)
 \right\}  \right|,
$$
the total number of inner polynomials returned by the recursive call.
Then, note that  $\mathrm{Coeff} \le \mathrm{Coeff}_{\mathop{rec}}^2 \cdot L_t^{L_t} + k  $ and $\mathrm{Num} = \mathrm{Num}_{\mathop{rec}}+k$.  By inductive hypothesis,
the polynomials  $\{\Inner(p_{s,q})^{(d-1)}_\ell\}$ are $\beta^{\ast}_t(\mathrm{Num}_{\mathop{rec}} + \mathrm{Coeff}_{\mathop{rec}})$-eigenregular. However, note that
\begin{eqnarray*}
\beta^{\ast}_t (\mathrm{Num}_{\mathop{rec}} + \mathrm{Coeff}_{\mathop{rec}}) &=& \beta((\mathrm{Num}_{\mathop{rec}} + \mathrm{Coeff}_{\mathop{rec}})^2 \cdot k \cdot L_t^{L_t} + k + L_t^{L_t})  \\
&\le& \beta(\mathrm{Coeff}_{\mathop{rec}}^2 \cdot L_t^{L_t} +k + \mathrm{Num}_{\mathop{rec}} +k) \le \beta(\mathrm{Num} + \mathrm{Coeff}).
\end{eqnarray*}
Here the inequality uses the fact that $\beta^{\ast}_t$ is a non-increasing function. This shows that the polynomials $\{\Inner(p_{s,q})^{(d-1)}_\ell\}$ polynomials are $\beta(\mathrm{Num} + \mathrm{Coeff})$-eigenregular.

Next, note that $R_{s, \reg}$ is $\eta_{t+1}$-eigenregular. We have
\begin{eqnarray*}
\eta_{t+1} &=& \beta
\left(
N_{\beta^{\ast}_{t}}\left(L_t,d-1,\frac{\tau}{16 \cdot L_t  \cdot L_t^{L_t}}\right)^2+M_{\beta^{\ast}_t} \left(L_t, d-1,\frac{\tau}{16 \cdot  L_t\cdot L_t^{L_t}}\right)^2 \cdot L_{t}^{L_t} +2k + L_t^{L_t}\right) \\
&\le& \beta \left(
\mathrm{Num}_{\mathop{rec}}^2  +\mathrm{Coeff}_{\mathop{rec}}^2  \cdot L_{t}^{L_t} +2k + L_t^{L_t}\right)
 \le \beta\left(\mathrm{Num} + \mathrm{Coeff} \right).
\end{eqnarray*}
This implies that $R_{s,\reg}$ is $\beta(\mathrm{Num} + \mathrm{Coeff})$-eigenregular for all $s\in [k]$ verifying Condition (4).
%Thus we see that in order for Condition 4(ii) to hold, it suffices to
%have that
%\begin{itemize}
%
%\item $\eta_{t+1} \leq \beta(\gamma_t(L_t,\tau,d) + k)$ (this ensures that
%$R_{s,\reg}$ is $\beta(N)$-regular), and
%
%\item $\beta^{\ast}_t(\gamma_t(L_t,\tau,d))\leq \beta(\gamma_t(L_t,\tau,d)+k)$
%(this ensures that every $\Inner(p_{s,q})_\ell$ polynomial
%other than $R_{s,\reg}$ is also $\beta(N)$-regular).
%
%\end{itemize}
%
%The first bullet above holds by the definition
%of the $\eta_{s+1}$ sequence (see (\ref{eq:eta-spec-d}).
%The second bullet above holds by the definition of $\beta^{\ast}_t$
%(see (\ref{eq:betaonedef}).
This concludes the proof of the inductive degree-$d$ case,
and thus concludes the proof of Theorem \ref{thm:general}.
\qed

\subsection{Proof of Theorem \ref{thm:main-decomp}}
\label{sec:main-decomp}

With Theorem \ref{thm:main-decomp} in hand it is straightforward to prove
Theorem \ref{thm:general}.  As in the statement of Theorem \ref{thm:general},
let us write $p(x)$ as
$\sum_{q=0}^d c_q p_q(x)$ where $p_q \in {\cal W}^q$ for all $q$ and
$\Var[p_q]=1$ for $1 \leq q \leq d$.  Since $\Var[p]=1,$ by Fact \ref{fact:var-ito} we have
that $\sum_{q=1}^d c_q^2 = 1.$
The procedure {\bf \blue{Regularize-Poly}}$_\beta$ calls
{\bf \blue{MultiRegularize-Many-Wieners}}$_{d,\beta}$ on the $d+1$ polynomials
$p_0, p_1,\dots,p_d$ with the ``$\tau$'' parameter  of {\bf \blue{MultiRegularize-Many-Wieners}}$_{d,\beta}$ set
to be $\frac{1}{d} \cdot (\tau/d)^{3d}$.
This call to {\bf \blue{MultiRegularize-Many-Wieners}}$_{d,\beta}$ returns
polynomials $\Outer(p_0),\dots,\Outer(p_d)$ (these are the polynomials
$h_0,\dots,h_d$) and
$\{\Inner(p_{0})_\ell\}_{\ell=1,\dots,m_0}$, $\dots,$
$\{\Inner(p_{d})_\ell\}_{\ell=1,\dots,m_d}$
(these are the polynomials
$\{A_{q,\ell}\}_{q=0,\dots,d,\ell=1,\dots,m_q}$).

Condition (1) of Theorem \ref{thm:main-decomp} follows directly from Condition (1)
of Theorem \ref{thm:general}.  Condition (2) of Theorem \ref{thm:general} implies
that each $q \in \{0,\dots,d\}$ has $\Var[p_q -\Outer(p_q)(\{\Inner(p_q)\}_{\ell})] \leq (1/d) \cdot (\tau/d)^{3d}$.
Observing that $p_q -\Outer(p_q)(\{\Inner(p_q)\}_{\ell})$ lies in ${\cal W}_q$, using Fact~\ref{fact:var-ito}, we have $\Var[p - \tilde{p}] \le (\tau/d)^{3d}.$
Using Lemma~\ref{lem:small-var-diff-kol-close}, we get that
\[
\left|
\Pr_{x \sim N(0,1)^n}[p(x) \geq 0] -
\Pr_{x \sim N(0,1)^n}[\tilde{p}(x) \geq 0]
\right|
\leq O(\tau).
\]
This concludes the verification of Condition (2). Conditions (3) and (4) in Theorem~\ref{thm:main-decomp} follow from the respective conditions in Theorem~\ref{thm:general}.

\fi

% %%%%%%%%%%%%%%%%%
% end of decomp.tex
% %%%%%%%%%%%%%%%%%

%% file: combine.tex
%\newpage

\ifnum\confversion=1
\section{Proof of Theorem \ref{thm:degd-main-gauss}}
\label{sec:combine}
\fi

\ifnum\confversion=0
\section{Proof of Theorem \ref{thm:degd-main-gauss}:    Deterministic approximate counting
for degree-$d$ PTFs over $N(0,1)^n$}
\label{sec:combine}
\fi

In this section we combine the tools developed in the previous sections to
prove Theorem \ref{thm:degd-main-gauss}.
We do this in two main steps.    First we use the
CLT from Section \ref{sec:CLT} and the decomposition procedure from
Section \ref{sec:decomp} to reduce the original problem (of $\eps$-approximately
counting satisfying assignments of a degree-$d$ PTF under $N(0,1)^n$)
to the problem of $\eps$-approximating an expectation $\E_{G \sim N(0^r,\Sigma)}[\tilde{g}_c(G)]$, where $N(0^r,\Sigma)$ is a mean-0
$r$-dimensional Gaussian with covariance matrix $\Sigma$,
and $\tilde{g}_c: \R^r \to [0,1]$ is a particular
explicitly specified function.  The key points here are that the value of $r$,
the description length (bit complexity) of $\tilde{g}_c$, and the bit complexity of each entry of the covariance matrix $\Sigma$
are all $O_{d,\eps}(1)$ (completely independent of $n$).  Next, we
describe how an $O_{d,\eps}(1)$-time deterministic algorithm can $\eps$-approximate the
desired expectation $\E_{G \sim N(0^r,\Sigma)}[\tilde{g}_c(G)]$.
\ifnum\confversion=1 Theorem \ref{thm:degd-main-gauss} follows directly from Theorems \ref{thm:finite-dim}
and \ref{thm:alg} which we state below.

\begin{theorem} \label{thm:finite-dim}
There is an $O_{d,\eps}(1) \cdot \poly(n^d)$-time deterministic
algorithm with the following performance guarantee:  Given as input a degree-$d$
real polynomial $p(x_1,\dots,x_n)$ and a parameter $\eps > 0$, it outputs
an integer $r$, a matrix of covariances
$\Sigma \in \R^{r \times r}$ (whose diagonal entries are all 1), and a description of a function $\tilde{g}_c: \R^r \to [0,1]$, such that
$ \left|
\Pr_{x \sim N(0,1)^n}[
p(x) \geq 0] -
\E_{G \sim N(0^r,\Sigma)}[\tilde{g}_c(G)] \right| \leq O(\eps).
$
Moreover $r$ is $O_{d,\eps}(1)$, the description length of
$\tilde{g}_c$ is $O_{d,\eps}(1)$ bits, and each entry of $\Sigma$ is a
rational number whose numerator and denominator are both integers of
magnitude $O_{d,\eps}(1).$
\end{theorem}

\begin{theorem} \label{thm:alg}
There is a deterministic $O_{d,\eps}(1)$-time algorithm
which, given as input the output $r, \Sigma, \tilde{g}_c$ of
Theorem \ref{thm:finite-dim} and  $\eps > 0$,
outputs a value $\nu$ such that
$\left| \nu - \mathbf{E}_{(G_1,\dots,G_r)  \sim N(0^r,\Sigma)} [\tilde{g}_c(G_1, \ldots, G_r)]\right| \le \epsilon.
$
\end{theorem}

\fi

\ifnum\confversion=0

Theorem \ref{thm:degd-main-gauss} follows directly from Theorems \ref{thm:finite-dim}
and \ref{thm:alg}, the main results of Sections \ref{sec:finite-dim} and \ref{sec:alg} respectively.

Our approach is based on mollification (so that we can apply our CLT);
we will need the following definitions from \cite{DKNfocs10}.   We first define the bump function $b: \R^r  \to \R$ as follows:
$$
b(x) = \begin{cases}\sqrt{C_r} ( 1- \Vert x \Vert_2^2 )&\mbox{if } \Vert x  \Vert_2 \le 1 \\
0 & \mbox{if }  \Vert x  \Vert_2 > 1, \end{cases}
$$
where the constant $C_r$ is chosen so that $\int_{x \in \mathbb{R}^r} b^2(x) dx = 1$. We let $\widehat{b}$ denote the Fourier transform of $b$, so
$$
\widehat{b}(x) = \frac{1}{(2\pi)^{r/2}} \cdot \int_{y \in \R^r} b(y) \cdot \exp(-\ignore{2\pi }i \langle x , y \rangle) dy.
$$
For $c > 0$ we define $B_c : \R^r \to \R$ as
$$
B_c(x) = c^r \cdot \widehat{b}(c\cdot x)^2.
$$
Using Parseval's identity, it is easy to see that $B_c(x)$ is non-negative and $\int_x B_c(x)=1$, so $B_c$ is a density function.
We observe that $\Vert \widehat{b} \Vert_{\infty}$ is upper bounded by $O_r(1)$ and hence $\Vert B_c \Vert_{\infty} = O_{c,r}(1)$. Finally, for $g : \R^r \rightarrow [0,1]$ and $c>0$,
$\tilde{g}_c : \R^r \to  [0,1]$ is defined
as
\begin{equation}\label{eq:desc}
\tilde{g}_c(x) = \int_{y \in \R^r} g(x-y)  \cdot B_c(y) dy.
\end{equation}

\subsection{Reducing to an $O_{d,\eps}(1)$-dimensional problem} \label{sec:finite-dim}

In this subsection we prove the following.  (The function
$\sign_{0,1}(z)$ below outputs 1 if $z \geq 0$ and outputs 0
otherwise.)

\begin{theorem} \label{thm:finite-dim}
There is an $O_{d,\eps}(1) \cdot \poly(n^d)$-time deterministic
algorithm with the following performance guarantee:  Given as input a degree-$d$
real polynomial $p(x_1,\dots,x_n)$ and a parameter $\eps > 0$, it outputs
an integer $r$, a matrix of covariances
$\Sigma \in \R^{r \times r}$ (whose diagonal entries are all 1), and a description of a function $\tilde{g}_c: \R^r \to [0,1]$, such that
\begin{equation} \label{eq:z}
\left|
\Pr_{x \sim N(0,1)^n}[
p(x) \geq 0] -
\E_{G \sim N(0^r,\Sigma)}[\tilde{g}_c(G)] \right| \leq O(\eps).
\end{equation}
Moreover, $\tilde{g}_c$ is of the form given in (\ref{eq:desc}), where $g  = \sign_{0,1}(\phi)$ and $\phi : \R^r \rightarrow \R$ is a degree-$d$ polynomial whose coefficients are rational numbers
with numerator and denominator that are each integers of magnitude $O_{d,\eps}(1)$. Also, $r$ is $O_{d,\eps}(1)$, and hence the description length of
$\tilde{g}_c$ is $O_{d,\eps}(1)$ bits.
Finally, each entry of $\Sigma$ is a
rational number whose numerator and denominator are both integers of
magnitude $O_{d,\eps}(1).$
\end{theorem}

 The following lemma will be useful for us.
 \begin{lemma} \label{lem:round-matrix}
 Let $a: \mathbb{R}^r \rightarrow [0,1]$ be  a $c$-Lipschitz function. Let $\Sigma, \Sigma' \in \R^{r \times r}$
 be two psd matrices such that $\Vert \Sigma-\Sigma' \Vert_2 \le \delta$. Then
 $\left | \E_{G \sim N(0^r,\Sigma)} [a(G)] - \E_{G' \sim N(0^r,\Sigma')}[a(G')] \right|  \le c r(\delta + 3 \sqrt{\delta \|\Sigma\|_2}).$
 \end{lemma}
 \begin{proof}
 Let $Z \sim {N}(0^r,\Sigma)$ and $Z' \sim {N}(0^r,\Sigma')$. It is shown in  \cite{DowLan82} that we have
 $$
 d_{W,2} (Z,Z')^2 = \Tr(\Sigma + \Sigma' - 2 (\Sigma^{1/2} \Sigma' \Sigma^{1/2})^{1/2}),
 $$
 where $d_{W,2}(\cdot, \cdot)$ denotes the Wasserstein distance between two distributions in the $\ell_2$ metric.
 Further, it is known \cite{Bhatia} that if $A$ and $B$ are psd matrices, then
 $$
 \Vert A^{1/2} - B^{1/2} \Vert_2 \le \sqrt{\Vert A- B \Vert_2}.
 $$
 Observe that both $\Sigma +\Sigma'$ and $4\Sigma^{1/2} \Sigma' \Sigma^{1/2}$ are psd. As a consequence,
 $$
 \Vert \Sigma + \Sigma' - 2 (\Sigma^{1/2} \Sigma' \Sigma^{1/2})^{1/2} \Vert_2 \le \sqrt{\Vert (\Sigma+\Sigma')^2 -4\Sigma^{1/2} \Sigma' \Sigma^{1/2} \Vert_2 }.
 $$
 However, note that if we let $\Delta = \Sigma'-\Sigma$ so that $\Vert \Delta \Vert_2 \le \delta$, then it is easy to see that
 $$
 \sqrt{\Vert (\Sigma+\Sigma')^2 -4\Sigma^{1/2} \Sigma' \Sigma^{1/2} \Vert_2 } \le \delta + 3 \sqrt{\delta \|\Sigma\|_2}.
 $$
Hence we have $\Tr(\Sigma + \Sigma' - 2 (\Sigma^{1/2} \Sigma' \Sigma^{1/2})^{1/2}) \leq r(\delta + 3 \sqrt{\delta \|\Sigma\|_2})^2,$ so
$d_{W,2} (Z,Z') \leq \sqrt{r}(\delta + 3 \sqrt{\delta \|\Sigma\|_2})$.
Using Cauchy-Schwarz, we get
$$
d_{W,1} (Z,Z') \le r(\delta + 3 \sqrt{\delta \|\Sigma\|_2}).
$$
Using the fact that $a$ is $c$-Lipschitz and recalling the coupling interpretation of the Wasserstein distance, we get the stated result.
 \end{proof}

Now we turn to the proof of Theorem \ref{thm:finite-dim}.
We first recall that by Theorem \ref{thm:multilinearize} we may assume
without loss of generality that
the polynomial $p(x_1,\dots,x_n)$ is multilinear.  By rescaling we may
further assume that $\Var[p]=1.$

The algorithm begins by running the procedure {\bf Decompose-Poly}$_\beta$ on $p$
with its ``$\tau$'' parameter set to $\eps$ and
the $\beta$ function set  to be
\begin{equation}\label{eq:def-beta}
\beta(x) = \left(\frac{\eps}{C \cdot d \cdot x}\right)^{Cd^2}
\end{equation}
where $C$ is a sufficiently large absolute constant (we discuss the choice of $C$ below; see (\ref{eq:eta-choice})).
  By
Theorem \ref{thm:main-decomp}, {\bf Decompose-Poly} outputs a polynomial
\[
\tilde{p}(x) = \sum_{j=0}^d h_j(A_{j,1},\dots,A_{j,m_j}(x))
\]
that satisfies $|\Pr_{x \sim N(0,1)^n}[p(x) \geq 0] - \Pr_{x
\sim N(0,1)^n}[\tilde{p}(x) \geq 0] |  \leq \eps.$
Furthermore, for each $j \in [1, \ldots, d]$ and $k \in [1 , \ldots, m_j]$
we have that
$\Var (A_{j,k}) =1$,
$A_{j,k} \in {\cal W}^q$ for some $q \in [1,j]$,
$\sum_{j=1}^d m_j \le r=r(d,\tau)$, and the sum of
squared coefficients of $h_j$ is bounded by $S=S(d,\tau)$.
We moreover have that $\Var[\tilde{p}] \in [1/2, 3/2]$, that
each $h_{j}$ is a multilinear polynomial of degree at most $d$,
\red{and that each $h_j(A_{j,1},\dots,A_{j,m_j}(x))$
lies in the $j$-th Wiener chaos (and hence has
expectation zero).} Furthermore,
if $A_{j,i_1}$ and $A_{j,i_2}$ appear in a monomial of $h_j$
together, then they contain disjoint sets of variables.

Since
each $h_j$ has at most $r^d$ coefficients, if we round each
coefficient of each $h_j$ to the nearest integer multiple of
$\sqrt{\red{(\eps/d)^{3d}} / (d r^d)}$ and call the resulting polynomial
$h_{new,j}$, and subsequently define
$$
\tilde{p}_{new}(x) \eqdef \E[p] + \sum_{j=1}^d h_{new,j}(A_{j,1}(x),
\dots,A_{j,m_j}(x)),
$$
then using Fact \ref{fact:variance-ineq} we have that
\[
\Var[\tilde{p}_{new}(x) - \tilde{p}(x)] \leq (d r^d) \cdot (\red{(\eps/d)^{3d}}/(d r^d))
= \red{(\eps/d)^{3d}}.
\]
Since $\tilde{p}$ and $\tilde{p}_{new}$ have the same mean we may
apply Lemma \ref{lem:small-var-diff-kol-close} and we get that
$|\Pr_{x \sim N(0,1)^n}[\tilde{p}_{new}(x) \geq 0] - \Pr_{x
\sim N(0,1)^n}[\tilde{p}(x) \geq 0] |  \leq O(\eps).$
From now on, we will work with the polynomial $\tilde{p}_{new}$.

At a high level, the plan is as follows:
$\Pr_{x \sim N(0,1)^n} [\tilde{p}_{new}(x) \ge 0]$ is equal to

\begin{equation} \label{eq:E-sign}
\mathbf{E}_{x \sim N(0,1)^n} \left[\sign_{0,1} \left(\E[\tilde{p}] + \sum_{j=1}^d  h_{new,j}(A_{j,1}(x),\dots,A_{j,m_{j}}(x))\right)\right].
\end{equation}
Since the $\sign_{0,1}$ function is discontinuous, we cannot apply our CLT from Section \ref{sec:CLT} to (\ref{eq:E-sign}) directly (recall that the CLT only applies to functions with bounded
second derivatives).  To get around this we define a ``mollified'' version (a version
that is continuous and has bounded second derivative) of the relevant function  and show that the expectation (\ref{eq:E-sign}) is close to the corresponding expectation of the mollified function.  This allows us to apply the CLT to the mollified function.
{A final step is that we will round the covariance matrix $\Sigma'$ obtained from the CLT to $\Sigma$ so that all its entries have bounded
bit complexity; Lemma \ref{lem:round-matrix} will ensure that we can do this and only incur an acceptable increase in the error.}

We now enter into the details.
Let $\phi : \mathbb{R}^{r} \rightarrow \R$ be defined as
$$
\phi(x_{1,1}, \ldots, x_{1,m_1}, \ldots, x_{d,1}, \ldots, x_{d, m_d}) =\E[p] + \sum_{j=1}^d  h_{new,j}(x_{j,1},\dots,x_{j,m_{j}}),
$$
and let $g : \mathbb{R}^r \rightarrow \{0,1\}$ be defined as $g(x) = \sign_{0,1}(\phi(x))$. In the subsequent discussion, for any function $F : \mathbb{R}^r \rightarrow \mathbb{R}$, we write $F^{(k)}:\mathbb{R}^r \rightarrow \mathbb{R}^m$ (where $m = \binom{r+k-1}{k}$) to denote the function whose coordinates are all  $k$-th order partial derivatives of $F$, so $\|F^{(k)}\|_\infty$
denotes the supremum of all values achieved by any $k$-th order partial derivative at any
input in $\R^r.$
Intuitively, our goal is to construct a mollification $\tilde{g} : \mathbb{R}^r \rightarrow \mathbb{R}$ for $g$ such that $\|\tilde{g}^{(2)}\|_\infty < \infty$ and $\tilde{g}$ is a ``good approximation'' to $g$
at most points in $\mathbb{R}^r$.
There are many different mollification constructions that could potentially be used; we shall use the following theorem from \cite{DKNfocs10}:

\begin{theorem}\label{DKNfocs10} \cite{DKNfocs10}
For any region $R \subseteq \mathbb{R}^r$ and any $c>0$,  the mollification
 $\widetilde{I}_{R,c}: \mathbb{R}^r \rightarrow [0,1]$ of
the $\{0,1\}$-valued function $I_R(x) \eqdef \mathbf{1}_{x \in R}$
has the property that for every point $x \in \mathbb{R}^r$,
$$
\left| I_R(x) - \widetilde{I}_{R,c}(x) \right| \le \min \left\{1, O\left( \frac{r^2}{c^2 \cdot \dist(x, \partial R)^2} \right) \right\},
$$
where $\dist(x,\partial R)$ denotes the Euclidean distance between $x$ and the boundary of $R$.
Moreover, this mollification satisfies
$\Vert \widetilde{I}_{R,c}^{(1)} \Vert_{\infty} \le 2c$ and
$\Vert \widetilde{I}_{R,c}^{(2)} \Vert_{\infty} \le 4c^2$.
\end{theorem}

Applying Theorem~\ref{DKNfocs10} to the region $R \eqdef \{x \in \mathbb{R}^r : g(x) = 1\}$, we get that for any $c>0$ there is a mollified function
$\tilde{g}_c : \mathbb{R}^r \rightarrow [0,1]$ such that  $\Vert \tilde{g}_c^{(2)} \Vert_{\infty} \le 4c^2$,
 $\Vert \tilde{g}_c^{(1)} \Vert_{\infty} \le 2c$, and
\begin{equation}\label{eqn:moll-error}
 \left| g(x) - \tilde{g}_c(x) \right| \le \min \left\{1, O\left( \frac{r^2}{c^2 \cdot \dist(x, \partial R)^2} \right) \right\}.
\end{equation}

The following lemma ensures that for a suitable
 choice of $c$, the mollification $\tilde{g}_c$ is indeed a useful proxy for $g$ for our purposes:
\begin{lemma}\label{lemma:error-bound}
For $c$ as specified in (\ref{eq:choice-of-c}) below, we have that the function $\tilde{g}_c : \mathbb{R}^r \rightarrow [0,1]$ described above satisfies
\begin{equation} \label{eq:carrot}
\E_{x \sim N(0,1)^n} \left[\left|
g(A_{1,1}(x), \dots, A_{d,m_d}(x)) - \tilde{g}_c(A_{1,1}(x), \dots, A_{d,m_d}(x))
\right| \right] \leq O(\eps).
\end{equation}
\end{lemma}

\begin{proof}
We will use the following claim:
\begin{claim}\label{claim:a}
Let $x, y \in \R^{r}$ and $\Vert x \Vert_{\infty}
\le B$. If $\Vert x - y \Vert_2 \le \delta \leq B$, then $|\phi(x) - \phi(y)|
\le d (2B)^d r^{d/2} \cdot  \sqrt{S} \cdot \delta.
\ignore{ d \cdot (2B)^d \cdot r^d \cdot S \cdot \delta}$
\end{claim}

\begin{proof}
Recall that $\phi(x)$ is a multilinear
degree-$d$ polynomial in $r$ variables for which the sum of squares of
coefficients is at most $S$. \ignore{Since the number of monomials in $\phi(x)$ is at most $r^d$,
we have that the sum of the absolute values of all coefficients of $\phi(x)$ is at most
$r^d \cdot S$.}  Let us write $\phi$ as
$$
\phi(x) = \sum_{\mathcal{A} \in \binom{[r]}{ \leq d}} c_\mathcal{A} x_\mathcal{A}
$$
where $x_\mathcal{A}$ represents the monomial $\prod_{i \in {\cal A}} x_i$ corresponding to the set $\mathcal{A}$, so we have $\sum_{\mathcal{A}} (c_{\mathcal{A}})^2 \le S$.
For any fixed $\mathcal{A} \in \binom{[r]}{\leq d}$, a simple application of the triangle
inequality across the (at most) $d$ elements of ${\cal A}$ gives that
$$
|x_{\mathcal{A}} - y_{\mathcal{A}}| \le d \cdot (2B)^d \cdot \delta.
$$
Since the number of monomials in $\phi(x)$ is at most $r^d$, using Cauchy-Schwarz we get
that
\[
|\phi(x)-\phi(y)| =
\left| \sum_{{\cal A}} c_{\cal A} (x_{\cal A}-y_{\cal A})\right|
\leq
\sqrt{\sum_{{\cal A}} (c_{\cal A})^2} \cdot \sqrt{\sum_{\cal A} (x_{\cal A}-y_{\cal A})^2}
\leq \sqrt{S} \cdot r^{d/2} \cdot d \cdot (2B)^d \cdot \delta,
\]
the claimed bound.
\end{proof}

By Claim  \ref{claim:a}, we have that if
$x \in \R^r$ has $\|x\|_\infty \leq B$ and
$|\phi(x)| > d (2B)^d r^{d/2} \cdot  \sqrt{S} \cdot \delta,$ where $\delta \leq B$,
then $\|\dist(x,\partial R)\|_2 > \delta$.
By (\ref{eqn:moll-error}), if $\dist(x,\partial R)  > \delta$ then
$|g(x)-\tilde{g}_c(x)| \leq
O({\frac {r^2}{c^2 \delta^2}}).$
Hence provided that we take $\delta \leq B$, we may upper bound
$\E_{x \sim N(0,1)^n} \left[\left|
g(A_{1,1}(x), \dots, A_{d,m_d}(x)) - \tilde{g}_c(A_{1,1}(x),
\dots, A_{d,m_d}(x)) \right| \right]$ (the LHS of
(\ref{eq:carrot}) by
\begin{eqnarray}
&& \Pr_{x \sim N(0,1)^n} \left[\max_{i \in [1,\ldots, d], j \in
 [1,\ldots, m(i)]}|A_{i,j}(x)| >B\right] \notag \\ &&+
\Pr_{x \sim N(0,1)^n}[|\phi(A_{1,1}(x), \dots, A_{d,m_d}(x))| \le d (2B)^d r^{d/2} \cdot  \sqrt{S} \cdot \delta] +
O\left({\frac {r^2}{c^2 \delta^2}}\right). \label{eq:salt}
\end{eqnarray}

To bound the second summand above, we recall that $\phi(A_{1,1}(x), \dots, A_{d,m_d}(x))$ is a multilinear degree-$d$
polynomial whose variance is at least $1/2$.
By the anti-concentration bound Theorem \ref{thm:cw}, we get that
\ignore{\[
\Pr_{x \sim N(0,1)^n}[|\phi(A_{1,1}(x), \dots, A_{d,m_d}(x))| \leq d \cdot (2B)^d r^{d/2} \cdot  \sqrt{S} \cdot \delta] \leq
O\left(d  \cdot B\cdot \sqrt{r} \cdot \delta^{1/d}\right).
\]}
\[
\Pr_{x \sim N(0,1)^n}[|\phi(A_{1,1}(x), \dots, A_{d,m_d}(x))| \leq d \cdot (2B)^d r^{d/2} \cdot  \sqrt{S} \cdot \delta] \leq
 O\left(d  \cdot B\cdot \sqrt{r} \cdot S^{1/2d} \cdot \delta^{1/d}\right).
\]
To bound the first summand,
we observe that since each $A_{i,j}$ is a mean-0 variance-1 degree-$d$ polynomial,
by the degree-$d$ Chernoff bound (Theorem~\ref{thm:dcb}) and a union bound over the $r$ polynomials
$A_{i,j}$, for any $B > e^d$ we have that
\begin{equation}\label{eq:dcb-union}
\Pr_{x \sim N(0,1)^n} \left[\max_{i \in [1,\ldots, d], j \in [1,\ldots, m(i)]}|A_{i,j}(x)| >B\right] \le r \cdot d \cdot e^{-\Omega(B^{2/d})}.
\end{equation}
Thus we get that (\ref{eq:salt}) is at most
\begin{equation} \label{eq:onion}
r \cdot d \cdot e^{-\Omega(B^{2/d})}
+ O\left(d \cdot B\cdot \sqrt{r} \cdot S^{1/2d} \cdot \delta^{1/d}\right)
+ O\left({\frac {r^2}{c^2 \delta^2}}\right).
\end{equation}
Choosing
\begin{equation} \label{eq:choice-of-c}
B = \left(\Omega(1) \cdot \ln {\frac {rd} \eps}\right)^{d/2}, \quad
\delta = \left({\frac {\eps}{d\cdot B \cdot \sqrt{r} \cdot S^{1/2d}}}\right)^d,
\quad \text{and} \quad
c = {\frac r {\delta \sqrt{\eps}}}
\end{equation}
(note that these choices satisfy the requirements that $B \geq e^d$
and $\delta \leq B$), we get that each of the three summands constituting
(\ref{eq:onion}) is $O(\eps)$, and Lemma \ref{lemma:error-bound} is proved.\end{proof}

\ignore{
%
%\red{
%
%\bigskip \bigskip \bigskip
%
%Combining  Claim~\ref{claim:a} with (\ref{eqn:moll-error}), we get conclude that
%\begin{align*}
%&\mathbf{E}_{x \sim \mathcal{N}^n(0,1)} \left|[g(A_{0,1}(x), A_{0,m_0}(x), \ldots, A_{d,1}(x), \ldots, %A_{d,m_d}(x))] - [\tilde{g}(A_{0,1}(x), A_{0,m_0}(x), \ldots, A_{d,1}(x), \ldots, A_{d,m_d}(x))]
%\right| \\ &\le \Pr_{x \sim \mathcal{N}^n(0,1)} \left[\max_{i \in [0,\ldots, d], j \in [1,\ldots, %m(i)]}|A_{i,j}(x)| >B\right] +
%\Pr[|h(x)| \le \delta] +
%\frac{N^{2d+2} \cdot d^2 \cdot S^2 \cdot (2B)^{2d} }{c^2 \cdot \delta^2  } \end{align*}
%Applying (\ref{eq:dcb-union}) and Theorem~\ref{thm:cw}, we get that we can upper bound the right hand %side by
%$$
%N \cdot d \cdot e^{-B^{2/d}} + d \delta^{1/d} +  \frac{N^{2d+2} \cdot d^2 \cdot S^2 \cdot (2B)^{2d} %}{c^2 \cdot \delta^2  }
%$$
%Plugging $$\delta = \frac{N^d \cdot \sqrt{d} \cdot S \cdot (2B)^d}{c^{1-1/d}} \quad B = \left(- \log %\left( \frac{d^2\cdot N^2 \cdot S^{1/d}}{c^{1/d}}\right)\right)^{d/2}$$
%we get that the right hand side can be upper bounded by
%$$
%O\left(d^2 \cdot N^2 \cdot \left(\frac{S}{c} \right)^{1/2d} \right)
%$$
%}
%
%As an immediate consequence of
%Claim \ref{clm:error-bound}, we get that
%\begin{equation} \label{eq:bread}
%\left|\E_{x \sim N(0,1)^n}
%[g(A_{1,1}(x), \dots, A_{d,m_d}(x))] -
%\E_{x \sim N(0,1)^n}[\tilde{g}_c(A_{1,1}(x), \dots, A_{d,m_d}(x))]
%\right|
%\leq O\left(d^2 \cdot r^2 \cdot \left(\frac{S}{c} \right)^{1/2d} \right).
%\end{equation}
}

Using Condition (4) of Theorem~\ref{thm:main-decomp}, we have that each $A_{j,k}$ is $\eta$-eigenregular where $\eta \le \beta(r+S)$. Now since $\|\tilde{g}_c^{(2)}\|_\infty \leq 4c^2$ and each $A_{j,k}(x)$ is a
mean-0,  degree-$d$ Gaussian polynomial with $\Var[A_{j,k}(x)]=1$,
we may apply our CLT, Theorem \ref{thm:mainclt}, and we get that
\begin{equation} \label{eq:butter}
\left|\E_{x \sim N(0,1)^n}[\tilde{g}_c(A_{1,1}(x), \dots, A_{d,m_d}(x))]
- \E_{G' \sim N(0^r,\Sigma')}[\tilde{g}_c(G')] \right| \leq
2^{O(d \log d)} \cdot r^2 \cdot \sqrt{\beta(r+s)} \cdot 4c^2,
\end{equation}
where $\Sigma' \in \R^{r \times r}$ is the covariance matrix corresponding to the
$A_{j,k}$'s (note that the variance bound on each $A_{j,k}$ ensures that the
diagonal entries are indeed all 1 as claimed).
It is easy to see that there exists a choice of $C$ in our definition of the function $\beta$ (see (\ref{eq:def-beta})) which has the property
{
\begin{equation} \label{eq:eta-choice}
\beta(r+s)  \le  {\frac {\eps^2}
{2^{O(d \log d)} r^4 c^4}}.
\end{equation}
}
As a result,
the triangle inequality applied to (\ref{eq:butter}) and Lemma \ref{lemma:error-bound}
gives that
\[
\left|
\Pr_{x \sim N(0,1)^n} [\tilde{p}(x) \ge 0] - \E_{G' \sim N(0^r,\Sigma')}[\tilde{g}_c(G')] \right| \leq
O(\eps).
\]
{We are almost done; it remains only to pass from $\Sigma'$ to $\Sigma,$ which we do using
Lemma \ref{lem:round-matrix}.  By
Theorem~\ref{DKNfocs10} we have that $\|\tilde{g}_{c}^{(1)}\|_\infty \leq 2c$ and hence
$\tilde{g}_c$ is $2c$-Lipschitz.  Observing that each entry of $\Sigma'$ is in $[-1,1]$, we have
that $\|\Sigma'\|_2 = O_{d,\eps}(1)$.  Hence by taking  $\Sigma \in \R^{r \times r}$ to be
a psd matrix that is sufficiently close to $\Sigma'$ with respect to
$\| \cdot \|_2$, we get that $\Sigma$ has all its coefficients
rational numbers with numerator and denominator of  magnitude $O_{d,\eps}(1)$, and from
Lemma \ref{lem:round-matrix}  we get that
$|\E_{G \sim N(0^r,\Sigma)}[\tilde{g}_c(G)] - \E_{G' \sim N(0^r,\Sigma')}[\tilde{g}_c(G')] \leq O(\eps).$
Thus Theorem \ref{thm:finite-dim} is proved.}
 \qed

\subsection{Solving the $O_{d,\eps}(1)$-dimensional problem in $O_{d,\eps}(1)$ time} \label{sec:alg}

The last step is to establish the following:
\begin{theorem} \label{thm:alg}
There is a deterministic $O_{d,\eps}(1)$-time algorithm
which, given as input the output $r, \Sigma, \tilde{g}_c$ of
Theorem \ref{thm:finite-dim} and the value of $\eps > 0$,
outputs a value $\nu$ such that
\[
\left| \nu - \mathbf{E}_{(G_1,\dots,G_r)  \sim N(0^r,\Sigma)} [\tilde{g}_c(G_1, \ldots, G_r)]\right| \le \epsilon.
\]
\end{theorem}

We only sketch the proof since the details of the argument are tedious and we are content with an $O_{d,\eps}(1)$ time bound.
The first observation is that since each entry of the covariance matrix $\Sigma$ is at most 1 in magnitude and $\tilde{g}_c$ is
everywhere bounded in $[0,1]$,  it suffices to
estimate the expected value conditioned on $(G_1,\dots,G_r)\sim N(0^r,\Sigma)$ lying in an origin-centered cube $[-Z,Z]^r$
for some $Z=O_{d,\eps}(1).$  Given this, it is possible to simulate this conditional normal distribution with a discrete probability
distribution $X$ supported on a finite set of points in $[-Z,Z]^r$.  Recall that $\Sigma$ has entries as specified in Theorem \ref{thm:finite-dim} (rational numbers of magnitude $O_{d,\eps}(1)$).
Since $\tilde{g}_c$ is $2c$-Lipschitz as noted earlier,
by using a sufficiently fine grid of $O_{d,\eps}(1)$
points in $[-Z,Z]^r$ and (deterministically) estimating the probability that $N(0^r,\Sigma)$ assigns to each grid point to a sufficiently small $1/O_{d,\eps}(1)$ additive error, Theorem \ref{thm:alg} reduces to the following claim:

\begin{claim} \label{claim:end}
Given any point $x \in [-Z,Z]^r$ and any accuracy parameter $\xi > 0$,
the function $\tilde{g}_c(x)$ can be computed to within additive accuracy $\pm \xi$ in time $O_{d,\eps,\xi}(1)$.
\end{claim}
\begin{proof}
We only sketch the proof of Claim \ref{claim:end} here as the argument is routine and the
details are tedious.  We first note that (as can be easily verified from the
description of $B_c$ given earlier) there
is an (easily computed) value $W = O_{d,\eps}(1)$ such that if $S = \{x \in \R^r : \Vert x \Vert_{\infty} > W \}$, then we have
$
\int_{z \in S}B_c(z) \le \xi/2.$
As a consequence, since $g$ is everywhere bounded in $[0,1]$, we get that
\begin{equation}\label{eq:prec1}
 \left| \int_{y \in \R^r} g(x-y)  \cdot B_c(y) dy -  \int_{y \in (\R^r \setminus S)} g(x-y)  \cdot B_c(y) dy \right| \le \xi/2,
\end{equation}
and hence it suffices to estimate
\begin{equation} \label{eq:S}
\int_{y \in [-W,W]^r} g(x-y)  \cdot B_c(y) dy
\end{equation}
to within an additive $\pm \xi/2.$
Now observe that $\|B_c\|_{\infty} = O_{d,\eps}(1)$ and $\|B_c^{(1)}\|_{\infty}=O_{d,\eps}(1).$ It follows that given any $y \in [-W,W]^r$ and any accuracy parameter $\rho > 0$,
we can compute $B_c(y)$ to additive accuracy $\pm \rho$ in time $O_{d,\eps,\rho}(1)$.
Recalling that $g(x)=\sign_{0,1}(\phi(x))$ where $\phi(x)$ is a degree-$d$ polynomial
whose coefficients are $O_{d,\eps}(1)$-size rational numbers, it follows that by
taking a sufficiently fine grid of $O_{d,\eps,\xi,W}(1)$ points in $[-W,W]^r$,
we can use such a grid to estimate (\ref{eq:S}) to an additive $\pm \xi/2$ in $O_{d,\eps,\xi,W}(1)$
time as desired.
\end{proof}

\fi

% % %%%%%%%%%%%%%%%%%%
% % end of combine.tex
% % %%%%%%%%%%%%%%%%%%

%% file: from-Gaussian-to-Boolean.tex
%\newpage

\section{Deterministic approximate counting for degree-$d$ polynomials
over $\{-1,1\}^n$}
\label{sec:from-Gaussian-to-Boolean}

In this section we use Theorem \ref{thm:degd-main-gauss} to prove Theorem \ref{thm:main}.  Since the
arguments here are identical to those used in \cite{DDS13:deg2count} (where an algorithm for deterministic approximate counting of degree-2 PTF satisfying assignments over $N(0,1)^n$ is used to
obtain an algorithm for satisfying assignments over $\{-1,1\}^n$), we only sketch
the argument here.

We recall the ``regularity lemma for PTFs'' of \cite{DSTW:10}.
This lemma says that every degree-$d$ PTF $\sign(p(x))$ over $\{-1,1\}^n$ can
be expressed as a shallow decision tree with variables at the internal
nodes and degree-$d$ PTFs at the leaves, such that a random
path in the decision tree is quite likely to reach a leaf that has a
``close-to-regular'' PTF.
As explained in \cite{DDS13:deg2count}, the \cite{DSTW:10} proof actually provides an
efficient deterministic procedure for constructing such a decision tree given $p$ as input,
and thus we have the following lemma (see Theorem 36 of \cite{DDS13:deg2count} for a
detailed explanation of how Theorem \ref{thm:algorithmic-regularity} follows
from the results of \cite{DSTW:10}):

\begin{theorem} \label{thm:algorithmic-regularity}
Let $p(x_1,\dots,x_n)$ be a multilinear degree-$d$ PTF.
Fix any $\tau>0$.  There is an algorithm $A_{\mathrm{Construct-Tree}}$
which, on input $p$ and a parameter $\tau > 0$, runs in
{$\poly(n,2^{\depth(d,\tau)})$}
time and outputs a decision tree $\T$ of depth
\[
\depth(d,\tau) := {\frac 1 \tau} \cdot \left(d \log {\frac 1 \tau}
\right)^{O(d)},
\]
where each internal node of the tree is labeled with a variable and
each leaf $\rho$ of the tree is labeled with a pair $(p_\rho,\mathrm{label}
(\rho))$
where $\mathrm{label}(\rho) \in \{+1,-1,\text{``fail''},\text{``regular''}\}.$
The tree $\T$ has the following properties:

\begin{enumerate}

\item Every input $x \in \{-1,1\}^n$ to the tree reaches a leaf
$\rho$ such that $p(x)=p_\rho(x)$;

\item If leaf $\rho$ has $\mathrm{label}(\rho) \in \{+1,-1\}$ then
$\Pr_{x \in \{-1,1\}^n}[\sign(p_\rho(x)) \neq \mathrm{label}(\rho)]
\leq \tau$;

\item If leaf $\rho$ has $\mathrm{label}(\rho) = \text{``regular''}$
then $p_\rho$ is $\tau$-regular; and

\item With probability at most $\tau$, a random path from the
root reaches a leaf $\rho$ such that $\mathrm{label}(\rho)=\text{``fail''}.
$

\end{enumerate}

\end{theorem}

\noindent {\bf Proof of Theorem~\ref{thm:main}:}  The
algorithm for approximating $\Pr_{x \in \{-1,1\}^n}[
p(x) \geq 0]$ to $\pm \eps$ works as follows.  It
first runs $A_{\mathrm{Construct-Tree}}$
with its ``$\tau$'' parameter set to $\Theta((\eps/d)^{{4d+1}})$
to construct the decision tree $\T$.  It then iterates over
all leaves $\rho$ of the tree.  For each leaf $\rho$ at depth
$d_\rho$ that has $\mathrm{label}(\rho)=+1$ it adds $2^{-d_\rho}$ to
$v$ (which is initially zero), and for each leaf $\rho$ at depth
$d_\rho$ that has $\mathrm{label}(\rho)=\text{``regular''}$
it runs the algorithm of Theorem~\ref{thm:degd-main-gauss} on $p_\rho$
(with its ``$\eps$'' parameter set to $\Theta((\eps/d)^{{4d+1}})$) to
obtain a value $v_\rho \in [0,1]$ and adds $v_\rho \cdot 2^{-d_\rho}$ to $v$.
It outputs the value $v \in [0,1]$ thus obtained.

Theorems~\ref{thm:algorithmic-regularity} and~\ref{thm:degd-main-gauss}
imply that the running time is as claimed.
To establish correctness of the algorithm we will use the
``invariance principle'' of \cite{MOO10} (see Theorem~2.1):

\begin{theorem}[\cite{MOO10}]
\label{thm:invariance} Let $p(x) = \littlesum_{S\subseteq[n], |S| \leq d}
p_S x_S $ be a degree-$d$
multilinear polynomial over $\{-1,1\}^n$ with
$\Var[p]=1$. Suppose each coordinate $i\in[n]$ has $\Inf_i(p) \leq \tau$. Then
\[ \sup_{t \in \R}|\Pr_x[p(x)\leq t] -
\Pr_{\mathcal{G} \sim N(0,1)^n}[p(\mathcal{G})\leq t]| \leq O(d\tau^{1/(
{4d+1})}).\]
 \end{theorem}

By Theorem~\ref{thm:algorithmic-regularity}, the leaves of $\T$ that are
marked $+1$, $-1$ or ``fail''
collectively contribute at most $\Theta((\eps/d)^{{4d+1}}) \leq \eps/2$
to the error
of the output value $v$.  Theorem~\ref{thm:invariance} implies that
each leaf $\rho$ at depth $d_\rho$ that is marked ``regular''
contributes at most
$2^{-d_\rho} \cdot \eps/2$ to the error, so the total contribution
from all such leaves is at most $\eps/2$.
This concludes the proof of Theorem~\ref{thm:main}.
\qed

% %%%%%%%%%%%%%%%%%%%%%%%%%%%%%%%%%%%
% end of from-Gaussian-to-Boolean.tex
% %%%%%%%%%%%%%%%%%%%%%%%%%%%%%%%%%%%

%% file: moments.tex
\ifnum\confversion=1
\section{Application: A deterministic FPT for approximating absolute moments} \label{sec:moments}
\fi

\ifnum\confversion=0
\section{Application of Theorem \ref{thm:main}:  A fixed-parameter deterministic multiplicative
approximation algorithm for absolute moments} \label{sec:moments}
\fi

Consider the following computational problem, which we call {\sc Absolute-Moment}:
Given a degree-$d$ polynomial $p(x_1,\dots,x_n)$
and an integer parameter $k \geq 1$, compute \blue{the} value $\E_{x \in \{-1,1\}^n}[|p(x)|^k]$
of the \emph{$k$-th absolute moment} of $p$.
It is clear that the \emph{raw moment} $\E[p(x)^k]$ can be computed
in roughly $n^k$ time by expanding out the polynomial $p(x)^k$,
performing multilinear reduction, and outputting the constant term.
Since the $k$-th raw moment equals the $k$-th absolute moment for even $k$,
this gives an $n^k$ time algorithm for
\noindent {\sc Absolute-Moment} for even $k$.  However,
as shown in \cite{DDS13:deg2count},
even for $d=2$
the {\sc Absolute-Moment} problem is \#P-hard for any odd $k \geq 1$, and thus
it is natural to seek approximation algorithms.

Using the hypercontractive inequality \cite{Bon70,Bec75}
it is not difficult to show that the obvious randomized algorithm
(draw uniform points from $\{-1,1\}^n$ and use
them to empirically estimate $\E_{x \in \{-1,1\}^n}[|p(x)|^k]$)
with high probability
gives a $(1 \pm \eps)$-accurate estimate of
the $k$-th absolute moment of $p$ in
in $\poly(n^d,2^{dk \log k}, 1/\eps)$
time.  In this section we observe that Theorem~\ref{thm:main}
yields a
\emph{deterministic} fixed-parameter-tractable $(1 \pm \eps)$-multiplicative
approximation algorithm for {\sc Absolute-Moment}:

\begin{theorem} \label{thm:compute-kth-moment}
There is a deterministic algorithm which, given any degree-$d$
polynomial $p(x_1,\dots,x_n)$ over $\{-1,1\}^n$, any integer
$k \geq 1$, and any $\eps > 0$,
runs in
$O_{d,k,\eps}(1) \cdot \poly(n^d)$ time
and outputs a value $v$ that multiplicatively
$(1 \pm \eps)$-approximates the $k$-th absolute moment:
\ifnum\confversion=0
\[
v \in \left[
(1-\eps) \E_{x \in \{-1,1\}^n}[|p(x)|^k],
(1+\eps) \E_{x \in \{-1,1\}^n}[|p(x)|^k]
\right].
\]
\fi
\ifnum\confversion=1
$v \in [
(1-\eps) \cdot $ $\E_{x \in \{-1,1\}^n}[|p(x)|^k],$ $
(1+\eps) \cdot $ $\E_{x \in \{-1,1\}^n}[|p(x)|^k]
].
$
\fi
\end{theorem}

\ifnum\confversion=0

Theorem \ref{thm:compute-kth-moment} is a generaliation of the
$d=2$ special case which was proved in \cite{DDS13:deg2count}
(using the deterministic approximate counting result for
degree-2 PTFs which is the main result of that paper).
The proof is closely analogous to the degree-2 case so we only
sketch it below; see \cite{DDS13:deg2count} for a detailed argument.

The first step is the following easy observation:

\begin{observation} \label{obs:moments-big}
Let $p(x)$ be a degree-$d$ polynomial over $\{-1,1\}^n$ that
has $\E_{x \in \{-1,1\}^n}[p(x)^2] = 1.$
Then for all $k \geq 1$ we have that the $k$-th absolute moment
$\E_{x \in \{-1,1\}^n}[|p(x)|^k]$ is at least $c_d$ where $c_d>0$
is some universal constant (depending only on $d$).
\end{observation}

Given an input degree-$d$ polynomial $p(x_1,\dots,x_n)$, we may divide
by $\|p\|_2$ to obtain a scaled version $q=p/\|p\|_2 $ which has
$\|q\|_2=1.$  Observation~\ref{obs:moments-big} implies that
an additive $\pm \eps$-approximation to
$\E[|q(x)|^k]$ is also a multiplicative $(1\pm O_d(\eps))$-approximation
to $\E[|q(x)|^k]$.
Multiplying the approximation by $\|p\|_2^k$ we obtain a multiplicative
$(1 \pm O_d(\eps))$-approximation to $\E[|p(x)|^k]$.  Thus to
prove Theorem~\ref{thm:compute-kth-moment} it suffices to give a
deterministic algorithm which finds an additive $\pm \eps$-approximation
to $\E[|q(x)|^k]$ for degree-$d$ polynomials with $\|q\|_2=1$.
This follows from Theorem~\ref{thm:degd-computemoments} below:

\begin{theorem} \label{thm:degd-computemoments}
Let $q(x)$ be an input degree-$d$ polynomial over $\{-1,1\}^n$
with $\E[q(x)^2]=1.$
There is an algorithm $A_{\mathrm{moment}}$ that, on input
$k \in \Z^+$, $q$, and $\eps > 0$, runs in time
$O_{k,d,\eps}(1) \cdot \poly(n^d)$
and outputs a value $\tilde{\mu}_k$ such that
\[
\left|\tilde{\mu}_k - \E_{x \in \{-1,1\}}[|q(x)|^k] \right|
\leq \eps.
\]
\end{theorem}

The idea behind the proof of Theorem \ref{thm:degd-computemoments}
is simple.  By Theorem \ref{thm:main}, we can estimate
$\Pr_{x \sim \{-1,1\}^n}[q(x) \geq t]$ to high accuracy
for any $t$ of our choosing.  Doing this repeatedly for different
choices of $t$, we can get a detailed picture
of where the probability mass of the random variable
$q(x)$ lies (for $x$ uniform over $\{-1,1\}^n$), and with this detailed
picture it is straightforward to estimate the $k$-th moment.

In a bit more detail, let $\gamma_q(t)$ denote the probability mass
function of $q(x)$ when $x$ is distributed uniformly over
$\{-1,1\}^n.$  We may write the $k$-th absolute moment of $q$
as
\begin{equation}
\E_{x \in \{-1,1\}^n}[|q(x)|^k]
=
\int_{-\infty}^\infty |t|^k \gamma_q(t) dt.
\end{equation}

Using standard tail bounds on polynomials over $\{-1,1\}^n$, for
a suitable choice of $M=O_{k,d,\eps}(1)$ we have that
\[
\E_{x \in \{-1,1\}^n}[|q(x)|^k]
\in \left[
\int_{-M}^{M} |t|^k \gamma_q(t) dt,
\int_{-M}^{M} |t|^k \gamma_q(t) dt + \eps/4\right]
,
\]
and hence to approximate $\E_{x \sim N(0,1)^n}[|q(x)|^k]$ to an
additive $\pm \eps,$ it suffices to approximate
$ \int_{-M}^M |t|^k \gamma_q(t) dt$ to an additive $\pm 3\eps/4.$

We may write
$\int_{-M}^{M} |t|^k \gamma_q(t) dt $ as
\[
\sum_{j=1-M/\Delta}^{M/\Delta}
\int_{(j-1)\Delta}^{j \Delta} |t|^k \gamma_q(t) dt
\]
(here $\Delta$ should be viewed as a small positive value).
When $|j|$ is small the summand
$\int_{(j-1)\Delta}^{j \Delta} |t|^k \gamma_q(t) dt$
may be well-approximated by zero, and when $|j|$ is not small the summand
$\int_{(j-1)\Delta}^{j \Delta} |t|^k \gamma_q(t) dt$
may be well-approximated by $|j \Delta|^k q_{j,\Delta}$, where
\[
q_{j,\Delta} \eqdef \Pr_{x \in \{-1,1\}^n}[q(x) \in [(j-1)\Delta,j\Delta).
\]
Using Theorem \ref{thm:main} twice we may approximate $q_{j,\Delta}$
to high accuracy.  With a suitable choice of
$\Delta=O_{k,d,\eps}(1)$ and the cutoff for $j$ being ``small,''
it is possible to approximate
$ \int_{-M}^M |t|^k \gamma_q(t) dt$ to an additive $\pm 3\eps/4,$
and thus obtain Theorem \ref{thm:degd-computemoments}, following this approach.
We leave the detailed setting of parameters to the interested reader.

\fi

% %%%%%%%%%%%%%%%%%%
% end of moments.tex
% %%%%%%%%%%%%%%%%%%